\documentclass[12pt]{article}%
\usepackage{eurosym}
\usepackage{amsmath,amssymb,amsthm}
\usepackage{geometry}
\usepackage{setspace}
\usepackage{graphicx}
\usepackage[authoryear,round]{natbib}
\usepackage{hyperref}
\hypersetup{hidelinks}
\usepackage{cancel}
\usepackage{xcolor}
\usepackage{tikz}
\usepackage{pgfplots}
\usepackage{amsmath}
\usepackage{amsfonts}
\usepackage{amssymb}%
\providecommand{\U}[1]{\protect\rule{.1in}{.1in}}
\pgfplotsset{compat=1.18}
\definecolor{mycolor}{rgb}{1.0, 0.9, 0.2}

\makeatletter
\@ifundefined{theorem}{\newtheorem{theorem}{Theorem}}{}
\@ifundefined{lemma}{\newtheorem{lemma}{Lemma}}{}
\@ifundefined{assumption}{\newtheorem{assumption}{Assumption}}{}
\@ifundefined{remark}{\newtheorem{remark}{Remark}}{}
\makeatother

\begin{document}

\title{Nonlinear Boosting with Multiple Testing in High-Dimensional Generalised
Linear Models with Binary Responses}
\author{\textbf{Charisios Grivas}\thanks{Department of Mathematics, Aalborg
University, Denmark (\href{charisios.grivas@gmail.com}{cgrivas@math.aau.dk}).}\\Aalborg University
\and \textbf{George Kapetanios}\thanks{Department of Banking and Finance, King's
College London, U.K.
(\href{george.kapetanios@kcl.ac.uk}{george.kapetanios@kcl.ac.uk}). }\\King's College London
\and \textbf{Zacharias Psaradakis}\thanks{Birkbeck Business School, Birkbeck,
University of London, U.K.
(\href{z.psaradakis@bbk.ac.uk}{z.psaradakis@bbk.ac.uk}).}\\Birkbeck, University of London
\and \textbf{Vasilis Sarafidis}\thanks{Brunel Business School, Brunel University of
London, U.K.
(\href{vasilis.sarafidis@brunel.ac.uk}{vasilis.sarafidis@brunel.ac.uk}).}\\Brunel University London
\and \textbf{Mari\'{a}n V\'{a}vra}\thanks{Research Department, National Bank of
Slovakia and Institute of Forecasting, Slovak Academy of Sciences, Slovak Republic  (\href{marian.vavra@nbs.sk}{marian.vavra@nbs.sk}%
).}\\National Bank of Slovakia\\Slovak Academy of Sciences
\and \textbf{Alexia Ventouri}\thanks{Department of Banking and Finance, King's
College London, U.K.
(\href{alexia.ventouri@kcl.ac.uk}{alexia.ventouri@kcl.ac.uk}). }\\King's College London}
\date{\today}
\maketitle

%##############################################################################
\clearpage\newpage
%###############################################################################

\begin{abstract}
\noindent This paper proposes a nonlinear boosting with multiple testing (BMT)
approach to variable selection in high-dimensional generalised linear models
with binary responses. At each stage of the BMT procedure, the model is
updated by adding only the most significant covariate, conditional on those already selected in previous stages, while taking into account the multiple testing nature of the
problem. It is shown that, under the stated conditions, the BMT procedure selects
all covariates whose true coefficients are nonzero, and no other covariates,
with probability tending to one. Furthermore, the procedure enjoys an
oracle property, in the sense that the post-BMT maximum likelihood estimator
of the parameters of the model is asymptotically equivalent to an oracle
estimator that knows the correct sparse model in advance. Monte Carlo
experiments demonstrate that BMT outperforms competing methods, delivering
high covariate-selection accuracy and low parameter estimation error. An
empirical example illustrates that BMT delivers a predictive model for the
probability that U.S. inflation exceeds a given threshold over a 12-month
horizon which has very good out-of-sample performance.\bigskip\ \newline%

\end{abstract}

%##############################################################################
\clearpage\newpage
%###############################################################################

\section{Introduction}

\label{sec:introduction}
%========================

Statistical models with binary response variables and a large number of
potential covariates (regressors) are widely used in economics and finance,
applications including, inter alia, default or firm-exit prediction, credit
scoring, prediction of the probability of recession or deflation, and analyses
of treatment adoption or policy compliance. In such settings, the researcher
typically encounters a large number of candidate covariates (possibly even
larger than the number of observations available on them), strong
cross-sectional dependence, and only a small number of truly relevant
covariates that contribute to the response. As a consequence, the problem of
variable selection in high-dimensional nonlinear settings is of particular
interest and importance. By identifying the subset of relevant covariates,
variable selection can improve estimation accuracy and lead to models which
are parsimoniously parameterised and easily interpretable.

Existing approaches to high-dimensional model selection are dominated by two
classes of methods. The first consists of penalised-likelihood procedures,
such as the least absolute shrinkage and selection operator (LASSO)
\citep{tibshirani1996}, the elastic net \citep{zou2005regularization} and
other related approaches \citep[e.g.,][]{fan2001scan,park2007lassologit}, to
name a few. These methods are computationally attractive and often achieve
high true positive rates (sensitivity). However, they are known to have non-negligible false
positive rates (fallout) in the presence of pseudo-signals, that is, covariates which do
not belong to the data-generating process (DGP) but are correlated with signal
covariates that do. Such pseudo-signal effects are extremely common in
applications involving macroeconomic and financial time series, many of which
are partially driven by a small number of latent common factors. The second
class of variable selection approaches comprises sure independent screening
procedures \citep[e.g.,][]{fan2008sure,FanSong2010} and multiple-testing
approaches such as the one covariate at a time multiple testing (OCMT)
approach \citep{chudik2018ocmt}. While these methods offer explicit error
control, they rely on marginal or weakly conditional statistics and,
therefore, inherit similar vulnerabilities with respect to pseudo-signal selection.

Building on the \emph{boosting with multiple testing} (BMT) framework
introduced for linear regression models in \citet{kapetanios2026}, the present
paper develops a \emph{nonlinear BMT} procedure appropriate for
high-dimensional generalised linear models (GLMs) with binary responses. The
procedure combines two key features. First, selection proceeds in a forward
stagewise manner analogous to that associated with boosting
\citep[e.g.,][]{BuhlmannHothorn2007}, adding at most one covariate per
iteration based on a conditional likelihood-based test statistic. Second, at
each stage, candidate covariates are filtered using a family-wise
multiple-testing threshold similar to that employed in OCMT, which serves both
as a screening device and as a stopping rule.

A distinctive feature of the nonlinear BMT method is that it does not rely on
sparsity-inducing penalties or marginal screening. Instead, identification and
selection are driven by conditional likelihood comparisons that adapt to the
current model. This makes the procedure particularly well suited to
environments with factor-driven dependence, where many covariates are
individually correlated with the response but only a small subset carries
independent explanatory power. The procedure is intentionally conservative, aiming to exclude pseudo-signals and noise variables even when covariates are highly collinear. This is partly achieved by
selecting a single additional covariate at each stage (based on a marginal
significance test), even though many candidate covariates may (individually)
have statistically significant coefficients. At each subsequent stage, the
remaining covariates are re-assessed conditionally on the expanded model, and
again at most one additional covariate is selected.

The main theoretical contribution is to establish finite-stage selection results for nonlinear BMT that parallel, but do not rely on, linear-model arguments. We
adopt a general binary-response GLM setting with a smooth link function and
allow the number of candidate covariates to grow with the sample size ---
while the number of signals remains fixed and finite. We establish three types
of results. First, we consider the behaviour of stagewise likelihood-based
test statistics along the BMT path in an approximating model that contains all
true signals and a finite number of pseudo-signals. Our results establish
stochastic control of false selections and justify the stopping rule on which
BMT relies. Second, we introduce explicit stagewise dominance conditions under
which signals are selected before any pseudo-signal or noise covariates. This
yields exact model recovery in both single-signal and multiple-signal cases,
in the sense that BMT asymptotically selects all signals and no other
covariate with probability tending to one. Third, conditional on correct model
selection, we show that the post-BMT maximum-likelihood estimator (MLE) of the
GLM parameters is asymptotically equivalent to the infeasible oracle estimator
that knows the signals in advance. In other words, the coefficients on the
signal variables are estimated as accurately as when the correct sparse model is known.

In addition to theoretical results, we provide extensive simulation evidence
that demonstrates the good performance of the nonlinear BMT procedure in
finite samples and its superiority over some competing methods. Specifically,
in a binary logistic setting with varying degrees of sparsity and collinearity
among the covariates, BMT delivers consistently
strong performance across a wide range of dimensionalities, combining highly
accurate selection of signals, model sizes that closely track the true sparsity level of the DGP and low parameter estimation error. Moreover, BMT outperforms
OCMT and LASSO, which tend to be less accurate in selecting signals and yield
overparameterised models and higher estimation errors.

We illustrate the practical use of nonlinear BMT by considering a predictive model for the probability that the U.S. inflation rate will, on
average, be higher than 2.5\% over a 12-month horizon. Using a large set of
variables taken from the FRED-MD database of \cite{Mccracken2021}, BMT selects predictors that yield highly accurate out-of-sample
identification of high-inflation periods. This is especially true when the
predictors and their estimated coefficients are updated in a recursive manner.
While models selected using OCMT and LASSO achieve marginally better in-sample fit, they rely on approximately seven and ten times as many predictors as BMT, respectively, and exhibit poorer out-of-sample performance.

Although the paper focuses on binary-response GLMs, the proposed methodology is considerably more general. Specifically, Appendix~\ref{app:addGLM} develops a BMT framework for high-dimensional GLMs with responses drawn from an one-parameter exponential family, thereby encompassing a wide range of general models used in applied work. We show that the central theoretical properties of the procedure, including uniform control of the stagewise test statistics, exact model recovery and oracle post-selection inference, continue to hold in this broader setting under suitable conditions. To maintain a clear exposition, the binary-response case is treated in the main text, while the more general theory is relegated to the appendix.

The remainder of the paper is organised as follows. Section~\ref{sec:glm_bmt}
introduces the general GLM framework, defines the stagewise submodels, and
formalises the nonlinear BMT procedure. Section~\ref{sec:assumptions} contains
the high-level assumptions used in the analysis. Section~\ref{sec:approx}
provides approximating-model results governing the behaviour of the test
statistics and the stopping rule used in the BMT algorithm.
Section~\ref{sec:exact}~contains exact-recovery results for the single-signal
and multiple-signal cases, establishes the oracle
properties of the post-selection parameter estimator and verifies the required
conditions for binary logit and probit models. Section~\ref{sec:primitive}
provides primitive sufficient conditions for a key dominance condition
required for exact model recovery and oracle inference, and verifies one of
these for binary logit and probit models. Section~\ref{sec:Monte_Carlo}
provides simulation evidence on the finite-sample properties of nonlinear BMT
and competing procedures. Section~\ref{sec:Example} discusses an empirical
illustration that highlights the usefulness of BMT in a predictive model for
the probability of U.S. inflation exceeding a given threshold value over a
specified time horizon. Section~\ref{sec:conclusion} summarises and concludes.
Proofs and technical lemmas are collected in Appendix~\ref{app:prelim}. The general exponential-family extension is given in Appendix~\ref{app:addGLM}, Appendix~\ref{app:binary_lasso} compares BMT with binary LASSO, and Appendix~\ref{app:primitive_wald} gives primitive sufficient conditions for the uniform Wald approximation. Additional simulation results are available from the authors on request.

\section{Setting and nonlinear BMT}

\label{sec:glm_bmt}
%========================

This section describes our GLM setup, which follows the classical formulation
of \citet{mccullagh1989glm}. The single-index conditional mean below is
maintained for the exact-recovery results. We nevertheless use quasi-likelihood
notation, both to permit robust inference and because the stagewise submodels
are generally misspecified when they omit active covariates. Their population
targets are therefore pseudo-true parameters in the sense of
\citet{white1982misspec,white94}. The stagewise submodels and the BMT algorithm
are defined below.

\subsection{GLM and inference}

\label{sec:glm_model}

The parametric model of interest is the single-index GLM
\begin{equation}
\mathbb{E}(y_{t}|\boldsymbol{x}_{t})=G(\boldsymbol{x}_{t}^{\prime
}\boldsymbol{\beta}),\quad t=1,2,\dots,T, \label{eq:glm_single_index}%
\end{equation}
where $y_{t}$ is a response variable taking values in $\{0,1\}$,
$\boldsymbol{x}_{t}=(x_{1,t},\dots,x_{p,t})^{\prime}\in\mathbb{R}^{p}$ is a
covariate vector whose dimension $p$ may exceed the number of available
observations $T$, $\boldsymbol{\beta}=(\beta_{1},\ldots,\beta_{p})^{\prime}%
\in\mathbb{R}^{p}$ is an unknown parameter, and $G$ is a known smooth and
monotone (inverse) link function defined on $\mathbb{R}$.\footnote{Without loss of
generality, an intercept is not included explicitly in the index
$\boldsymbol{x}_{t}^{\prime}\boldsymbol{\beta}$. The covariate corresponding
to an intercept is typically treated as an \textquotedblleft
always-in\textquotedblright\ control.} In high-dimensional settings where $p$
is large, it is typically assumed that only a small number of covariates
contribute to the conditional mean of the response, which amounts to the
parameter $\boldsymbol{\beta}$ and its (unknown) true value, denoted
$\boldsymbol{\beta}_{0}=(\beta_{0,1},\ldots,\beta_{0,p})^{\prime}\in
\mathbb{R}^{p}$, being sparse.

In what follows, $F_{p}=\{1,2,\ldots,p\}$ denotes the set of candidate
covariates. Let $S^{(0)}\subseteq F_p$ be a fixed, possibly empty, set of
always-in controls. The \emph{true active set} relevant for selection is the
set of active non-always-in covariates,
\begin{equation}
S_{0}=\{j\in F_{p}\setminus S^{(0)}:\beta_{0,j}\neq0\},\qquad
\left\vert S_{0}\right\vert =k.
\label{eq:true_set}%
\end{equation}
It is maintained throughout that $S_{0}$ is non-empty with cardinality $k$
that does not depend on $T$. By contrast, the cardinality $p=p_{T}$ of $F_{p}$
is allowed to grow with $T$. A candidate covariate $x_{j,t}$ is said to be a
\emph{signal} variable if $j\in S_{0}$, a \emph{pseudo-signal} (or
\emph{proxy}) variable if $j\notin S^{(0)}\cup S_{0}$ and
$\mathrm{Cov}(x_{j,t},x_{i,t})\neq0$ for some $i\in S_{0}$, and a
\emph{noise} variable if $j\notin S^{(0)}\cup S_{0}$ and
$\mathrm{Cov}(x_{j,t},x_{i,t})=0$ for all $i\in S_{0}$.

Inference on the parameter $\boldsymbol{\beta}$ in (\ref{eq:glm_single_index})
is typically likelihood-based. The following standard $M$-estimation discussion
applies to a fixed-dimensional fitted model, including each admissible
stagewise submodel because $k_{\max}$ is fixed. To fix notation,
let $l_{t}(\eta_{t},y_{t})$ denote a postulated log-density of $y_{t}$ given
$\boldsymbol{x}_{t}$ evaluated at index $\eta_{t}=\eta_{t}(\boldsymbol{\beta
})=\boldsymbol{x}_{t}^{\prime}\boldsymbol{\beta}$ and response $y_{t}$; put
$\dot{l}_{t}(\eta,y)=\partial l_{t}(\eta,y)/\partial\eta$ and $\ddot{l}%
_{t}(\eta,y)=\partial^{2}l_{t}(\eta,y)/\partial\eta^{2}$. A (quasi-) MLE
$\widehat{\boldsymbol{\beta}}=(\widehat{\beta}_{1},\ldots,\widehat{\beta}%
_{p})^{\prime}$ is obtained as a maximiser with respect to $\boldsymbol{\beta
}$ of the average (quasi-) log-likelihood
\begin{equation}
L_{T}(\boldsymbol{\beta})=T^{-1}\sum_{t=1}^{T}l_{t}(\eta_{t},y_{t}).
\label{eq:criterion}%
\end{equation}

For such a fixed-dimensional model, if $l_{t}(\eta_{t},y_{t})$ is correctly
specified (which amounts to correct specification of the response probability
$\mathbb{P}(y_{t}=1|\boldsymbol{x}_{t})$), the maximiser of the empirical
criterion (\ref{eq:criterion}) is $\sqrt{T}$-consistent for the corresponding
true parameter and asymptotically normal. Under quasi-likelihood or other
criterion misspecification, the same conclusion holds for the pseudo-true
parameter, denoted $\boldsymbol\beta^*$, under the usual $M$-estimation
conditions \citep{white1982misspec,white94}. When the (quasi-) score
sequence is serially uncorrelated, a consistent estimate of the
asymptotic covariance matrix of
$\sqrt T(\widehat{\boldsymbol{\beta}}-\boldsymbol{\beta}^{\ast})$ is the
one-period sandwich matrix
\begin{equation}
\left(  T^{-1}\sum_{t=1}^{T}\boldsymbol{H}_{t}(\widehat{\boldsymbol{\beta}%
})\right)  ^{-1}\left(  T^{-1}\sum_{t=1}^{T}\boldsymbol{s}_{t}%
(\widehat{\boldsymbol{\beta}})\boldsymbol{s}_{t}(\widehat{\boldsymbol{\beta}%
})^{\prime}\right)  \left(  T^{-1}\sum_{t=1}^{T}\boldsymbol{H}_{t}%
(\widehat{\boldsymbol{\beta}})\right)  ^{-1},\label{eq:cov_est}%
\end{equation}
where $\boldsymbol{s}_{t}(\boldsymbol{\beta})=\partial l_{t}(\eta_{t}%
,y_{t})/\partial\boldsymbol{\beta}=\boldsymbol{x}_{t}\dot{l}_{t}(\eta
_{t},y_{t})$ and $\boldsymbol{H}_{t}(\boldsymbol{\beta})=\partial^{2}%
l_{t}(\eta_{t},y_{t})/\partial\boldsymbol{\beta}\partial\boldsymbol{\beta
}^{\prime}=\boldsymbol{x}_{t}\boldsymbol{x}_{t}^{\prime}\ddot{l}_{t}(\eta
_{t},y_{t})$. Robustness with respect to dependence in the (quasi-) scores
$\boldsymbol{s}_{t}(\boldsymbol{\cdot})$ across observations may be achieved
by replacing the second factor in (\ref{eq:cov_est}) with $T^{-1}\sum
_{t=1}^{T}\sum_{i=1}^{T}K(|t-i|/\varpi)\boldsymbol{s}_{t}%
(\widehat{\boldsymbol{\beta}})\boldsymbol{s}_{i}(\widehat{\boldsymbol{\beta}%
})^{\prime}$, where $K:\mathbb{R}\rightarrow\lbrack-1,1]$ is a suitable kernel
function and $\varpi=\varpi_{T}>0$ is a bandwidth parameter that is an
increasing function of $T$. Under the standard GLM assumptions of correct
specification and conditionally independent responses, this asymptotic
covariance matrix reduces to the inverse
of the average observed information
$-T^{-1}\sum_{t=1}^{T}\boldsymbol{H}_{t}(\widehat{\boldsymbol{\beta}})$.
The covariance matrix of $\widehat{\boldsymbol{\beta}}$ itself is estimated
by $T^{-1}$ times the relevant asymptotic covariance matrix; stagewise standard
errors below use this $T^{-1/2}$ scaling.

\subsection{Stagewise submodels and pseudo-true parameters}

\label{sec:pseudotrue}

Let $S\subseteq F_{p}$ denote a conditioning set of selected covariates and
$\boldsymbol{x}_{S,t}$ be the subvector of $\boldsymbol{x}_{t}$ indexed by
$S$. For a candidate $j\notin S$, the associated \emph{augmented submodel} is
defined as the GLM with parameter $\boldsymbol{\vartheta}_{j}%
=(\boldsymbol{\gamma}_{j}^{\prime},\theta_{j})^{\prime}\in\mathbb{R}^{|S|+1}$
and index
\begin{equation}
\eta_{j,t}(\boldsymbol{\vartheta}_{j};S)=\boldsymbol{x}_{S,t}^{\prime
}\boldsymbol{\gamma}_{j}+\theta_{j}\,x_{j,t}. \label{eq:aug_index}%
\end{equation}
Letting $L_{T,j}(\boldsymbol{\vartheta}_{j};S)=T^{-1}\sum_{t=1}^{T}l_{t}%
(\eta_{j,t}(\boldsymbol{\vartheta}_{j};S),y_{t})$ be the average (quasi-)
log-likelihood for the augmented submodel, denote a (possibly local) maximiser
of it by
\begin{equation}
\widehat{\boldsymbol{\vartheta}}_{j}(S)\in\arg\max_{\boldsymbol{\vartheta}%
_{j}}L_{T,j}(\boldsymbol{\vartheta}_{j};S), \label{eq:mle_aug}%
\end{equation}
and write $\widehat{\theta}_{j}(S)$ for the component of
$\widehat{\boldsymbol{\vartheta}}_{j}(S)$ corresponding to $x_{j,t}$.

For theoretical comparisons, we also require the \emph{pseudo-true} parameter
$\boldsymbol{\vartheta}_{j}^{\ast}(S)$, defined as the maximiser of the
expected (quasi-) log-likelihood, i.e.,
\begin{equation}
\boldsymbol{\vartheta}_{j}^{\ast}(S)\in\arg\max_{\boldsymbol{\vartheta}_{j}%
}\;\mathbb{E}\bigl[l_{t}(\eta_{j,t}(\boldsymbol{\vartheta}_{j};S),y_{t}%
)\bigr], \label{eq:pseudotrue}%
\end{equation}
where the expectation is taken under the true conditional distribution of
$y_{t}$ given $\boldsymbol{x}_{t}$; we write $\theta_{j}^{\ast}(S)$ for the
component of $\boldsymbol{\vartheta}_{j}^{\ast}(S)$ corresponding to $x_{j,t}%
$. Under correct specification and when
$S\supseteq S^{\dagger}=S^{(0)}\cup S_{0}$, $\theta_{j}^{\ast}(S)=0$
for every $j\notin S^{\dagger}$; under misspecification,
$\theta_{j}^{\ast}(S)$ may be nonzero and captures the population
\textquotedblleft proxy\textquotedblright\ effect of $x_{j,t}$ after
conditioning on $\boldsymbol{x}_{S,t}$.

For subsequent theoretical results, the maximiser in \eqref{eq:pseudotrue} is
assumed to be unique for every admissible pair $(S,j)$. If the population
criterion has multiple maximisers, the results should be read under a fixed
deterministic tie-breaking rule and the additional requirement that all tied
maximisers yield the same candidate coefficient and the same relevant
curvature and variance quantities. This convention makes
$\theta_j^*(S)$ and the population noncentrality defined below single-valued.

\subsection{Stagewise Wald statistics and score implementation}

\label{sec:wald}

For a given $S$ and $j\notin S$, define the (absolute) \emph{stagewise Wald
statistic} for testing $H_{0}:\theta_{j}=0$ in the augmented submodel
\eqref{eq:aug_index} as
\begin{equation}
W_{T,j}(S)=\frac{\bigl|\widehat{\theta}_{j}(S)\bigr|}{\widehat{\mathrm{se}%
}_{j}(S)}. \label{eq:wald_stat}%
\end{equation}
Throughout the theoretical analysis,
$\widehat{\mathrm{se}}_{j}(S)$ is based on an estimator that is uniformly
consistent for
$T^{-1/2}\{V_{\theta\theta,j}^{\ast}(S)\}^{1/2}$, where
$V_{\theta\theta,j}^{\ast}(S)$ is obtained from the long-run sandwich
matrix in Assumption~\ref{ass:A3}. The inverse observed information may be
used when the information equality holds for the particular augmented
submodel, and the one-period sandwich may be used when its stagewise scores
are serially uncorrelated. Otherwise, a uniformly consistent long-run or heteroskedasticity and autocorrelation consistent (HAC)
variance estimator is required. The uniform rate needed when population
noncentralities diverge is part of Assumption~A3.4; pointwise relative
consistency alone is not sufficient. Subject to this requirement, the ordering
arguments do not depend on the particular implementation of the standard
error.

In addition to \eqref{eq:wald_stat}, we will use a Lagrange multiplier
(LM)/score statistic that is convenient for probability bounds. Let
$\widehat{\boldsymbol{\vartheta}}_{0}(S)$ denote the restricted (quasi-) MLE
obtained subject to $\theta_j=0$, let
$\widehat\eta_t(S)=\eta_{j,t}(\widehat{\boldsymbol{\vartheta}}_{0}(S);S)$, and
put $\widehat p_t(S)=G\{\widehat\eta_t(S)\}$. A convenient residualised score form is available when the restricted
Bernoulli submodel is correctly specified. In that case, define
\[
\widehat a_t(S)
=
\frac{G'\{\widehat\eta_t(S)\}}
{\widehat p_t(S)\{1-\widehat p_t(S)\}},
\qquad
\widehat w_t(S)
=
\frac{G'\{\widehat\eta_t(S)\}^{2}}
{\widehat p_t(S)\{1-\widehat p_t(S)\}}.
\]
Let $\widetilde{x}_{j,t}(S)=x_{j,t}-\boldsymbol{x}_{S,t}^{\prime}
\widehat{\boldsymbol{\pi}}_{j}(S)$ be the residual from the weighted
least-squares projection of $x_{j,t}$ on $\boldsymbol{x}_{S,t}$ using weights
$\widehat w_t(S)$. The corresponding normalised score statistic is
\begin{equation}
\widetilde{W}_{T,j}(S)
=
\frac{\left|T^{-1/2}\sum_{t=1}^{T}\widetilde{x}_{j,t}(S)
\widehat a_t(S)\{y_t-\widehat p_t(S)\}\right|}
{\left[T^{-1}\sum_{t=1}^{T}\widetilde{x}_{j,t}^{2}(S)
\widehat w_t(S)\right]^{1/2}}.
\label{eq:score_form}%
\end{equation}
For the logit link, $\widehat a_t(S)=1$ and
$\widehat w_t(S)=\widehat p_t(S)\{1-\widehat p_t(S)\}$. For a non-canonical link when the restricted submodel is misspecified, the LM
statistic should instead be constructed directly from the restricted score,
the Schur complement of the actual restricted sample Hessian, and the relevant
sandwich or long-run score covariance; the simple weights above need not equal
the curvature weights. The formal theory below is Wald-based, so no global
score representation is required.

The formal results below are stated for the Wald statistic $W_{T,j}(S)$. If
the score statistic in \eqref{eq:score_form} is used instead, the same
arguments apply after imposing the analogue of Assumption~\ref{ass:A3}.4 with
$\widetilde W_{T,j}(S)$ replacing $W_{T,j}(S)$. Appendix~\ref{app:wald_score}
records the standard local Wald--score equivalence for population-null and
local alternatives; no global equivalence is required for exact-recovery
results.

\subsection{Nonlinear BMT algorithm}

\label{sec:algo}

To give a formal description of the BMT procedure, fix an initial set
(possibly empty) of \textquotedblleft always-in\textquotedblright\ controls
$S^{(0)}$ and a maximum number of stages $k_{\max}$ (with $k_{\max}\geq k$).
For $\ell\geq0$, given the current set $S^{(\ell)}$, the nonlinear BMT
algorithm comprises four steps.

\begin{enumerate}
\item \emph{Restricted fit}\textbf{.} Estimate the parameters of a GLM with
covariates indexed by $S^{(\ell)}$ and obtain the fitted objects of interest
(residual-like terms and weights).

\item \emph{Candidate screening}\textbf{.} For each $j\notin S^{(\ell)}$,
compute the stagewise statistic $W_{T,j}\bigl(S^{(\ell)}\bigr)$. A score
statistic $\widetilde{W}_{T,j}\bigl(S^{(\ell)}\bigr)$ may be used instead
when the uniform approximation condition is imposed for that statistic.

\item \emph{Multiple-testing filter}\textbf{.} For a family-wise threshold
$c_{T}>0$, form the passing set
\begin{equation}
\mathcal{J}^{(\ell)}=\Bigl\{j\notin S^{(\ell)}:W_{T,j}\bigl(S^{(\ell
)}\bigr)\geq c_{T}\Bigr\}. \label{eq:pass_set}%
\end{equation}

\item \emph{Stagewise selection}\textbf{.} If $\mathcal{J}^{(\ell
)}=\varnothing$, stop and return $\widehat{S}=S^{(\ell)}$. Otherwise, select
\[
j^{(\ell)}\in\arg\max_{j\in\mathcal{J}^{(\ell)}}W_{T,j}\bigl(S^{(\ell
)}\bigr),\qquad S^{(\ell+1)}=S^{(\ell)}\cup\{j^{(\ell)}\},
\]
and continue until the algorithm stops or $k_{\max}$ successful selections
have been made.
\end{enumerate}

Theoretical results in subsequent sections establish: (i) high-probability
stochastic bounds for the test statistics along the selection path
(approximating-model results in Theorems~\ref{thm:approx_stats} and
\ref{thm:approx_stop}); (ii) sufficient conditions ensuring that the stagewise
maximiser is always a signal until $S_{0}$ is exhausted (exact recovery
results in Theorems~\ref{thm:k1_recovery} and \ref{thm:fixedk_recovery}),
leading to oracle post-selection inference (Theorem~\ref{thm:oracle}).
%========================

\section{Assumptions}

\label{sec:assumptions}
%========================

This section collects the assumptions used in the theoretical analysis. They
are organised to mirror the logical roles of the main theorems. Assumptions
\ref{ass:A1}--\ref{ass:A3} are sufficient for the approximating-model results,
while Assumptions \ref{ass:A4}--\ref{ass:A5} introduce additional structure
required for exact model recovery and oracle inference. Henceforth, limits in order symbols are taken as $T\rightarrow\infty$, unless stated otherwise.

%------------------------
\renewcommand{\theassumption}{A\arabic{assumption}}\renewcommand{\theHassumption}{A.\arabic{assumption}}

\begin{assumption} \label{ass:A1}\qquad

\begin{enumerate}
\item[A1.1] The triangular array
$\{(y_t,\boldsymbol x_t')\}_{t=1}^T$ is row-wise strictly stationary and
strongly mixing. Let $\alpha_T(h)$ denote its strong-mixing coefficient,
uniformly over rows. Under A2-ET, there are constants
$C_\alpha,c_\alpha>0$ such that
$\sup_T\alpha_T(h)\le C_\alpha\exp(-c_\alpha h)$ for every $h\ge1$.
Under A2-PM, with $q$ and $\delta$ as defined there,
\[
\sum_{h=1}^{\infty}h^{q/2-1}
\sup_T\alpha_T(h)^{\delta/(q+\delta)}<\infty.
\]

\item[A1.2] The number of covariates $|F_p|=p=p_T$ may diverge with $T$,
while the number of non-always-in signals $|S_0|=k$ and the maximum number of
successful selections $k_{\max}$ are fixed, positive and finite. The number
of always-in covariates $s_0=|S^{(0)}|$ is fixed and finite and may be zero.
The always-in set and the signal set are disjoint:
$S^{(0)}\cap S_0=\varnothing$.\medskip
\end{enumerate}
\end{assumption}

For later use, let $K=k_{\max}+1$, $d_0=s_0+k_{\max}+1$ and
$\ell_T=\log(p\vee T)$, and define the admissible class
$\mathcal A_T=\{(S,j):S^{(0)}\subseteq S\subseteq F_p,
|S\setminus S^{(0)}|\le k_{\max},\ j\notin S\}$.

For an admissible pair $m=(S,j)\in\mathcal A_T$, write
$l_{m,t}(\boldsymbol\vartheta)=
l_t\{\eta_{j,t}(\boldsymbol\vartheta;S),y_t\}$. Let
$F_{1,m,t}$, $F_{2,m,t}$ and $F_{3,m,t}$ be coordinatewise envelopes, over
$\boldsymbol\vartheta\in\mathcal B_{j,S}$ with the compact sets defined
in A3.2, for the score, negative Hessian
and third derivative of $l_{m,t}$, respectively. Let
$F_{\Omega,0,m,t}$ envelope the coordinates of the score outer product and
let $F_{\Omega,1,m,t}$ envelope their first parameter derivatives. For
$\nu>0$, write
$\|Z\|_{\psi_\nu}=\inf\{c>0:\mathbb E\exp[(|Z|/c)^\nu]\le2\}$.

\begin{assumption}
\label{ass:A2} One of the following sets of conditions is satisfied:

\begin{enumerate}
\item[A2-ET] Uniformly over admissible $m$ and $t$,
\[
\|F_{1,m,t}\|_{\psi_2}
+\|F_{2,m,t}\|_{\psi_1}
+\|F_{\Omega,0,m,t}\|_{\psi_1}
+\|F_{3,m,t}\|_{\psi_{2/3}}
+\|F_{\Omega,1,m,t}\|_{\psi_{2/3}}
\le C_F<\infty.
\]

\item[A2-PM] There exist $q>\max\{4,d_0\}$ and $\delta>0$ such that,
uniformly over admissible $m$ and $t$,
\[
\mathbb E\!\left(
F_{1,m,t}^{q+\delta}+F_{2,m,t}^{q+\delta}
+F_{\Omega,0,m,t}^{q+\delta}+F_{3,m,t}^{q+\delta}
+F_{\Omega,1,m,t}^{q+\delta}
\right)\le C_F<\infty.
\]
\end{enumerate}
\end{assumption}

\begin{assumption}
\label{ass:A3}\qquad

\begin{enumerate}
\item[A3.1] The (inverse) link function
$G:\mathbb{R}\rightarrow(0,1)$ is strictly increasing and three-times
continuously differentiable, with derivatives that are bounded on compact
subsets of $\mathbb{R}$.

\item[A3.2] For each admissible pair $(S,j)$, there is a compact parameter
set $\mathcal B_{j,S}\subset\mathbb R^{|S|+1}$ containing
$\boldsymbol\vartheta_j^*(S)$ in its interior. Uniformly over all stagewise
submodels with $|S\setminus S^{(0)}|\le k_{\max}$ and all $j\notin S$, the
norms of the pseudo-true parameters and the diameters of the sets
$\mathcal B_{j,S}$ are bounded, and the distances of the pseudo-true
parameters from the boundaries of $\mathcal B_{j,S}$ are bounded away from
zero. Writing
$Q_{j,S}(\boldsymbol\vartheta)=\mathbb E
l_{m,t}(\boldsymbol\vartheta)$, the maximisers are uniformly well separated:
for every $\epsilon>0$ there is $c(\epsilon)>0$ such that
\[
\inf_{(S,j)\in\mathcal A_T}
\left[
Q_{j,S}\{\boldsymbol\vartheta_j^*(S)\}
-
\sup_{\substack{\boldsymbol\vartheta\in\mathcal B_{j,S}:\\
\|\boldsymbol\vartheta-\boldsymbol\vartheta_j^*(S)\|\ge\epsilon}}
Q_{j,S}(\boldsymbol\vartheta)
\right]\ge c(\epsilon).
\]

\item[A3.3] For any admissible
conditioning set $S$ with $|S\setminus S^{(0)}|\leq k_{\max}$ and any
$j\notin S$, the population information matrix
\[
\boldsymbol{J}_{j}^{\ast}(S)
= -\mathbb{E}\left[
\frac{\partial^{2}l_{t}(\eta_{j,t}(\boldsymbol{\vartheta}_{j}^{\ast}(S);S),y_t)}
{\partial\boldsymbol{\vartheta}_{j}\partial\boldsymbol{\vartheta}_{j}^{\prime}}
\right]
\]
is finite and nonsingular, with smallest eigenvalue uniformly bounded away
from zero. The long-run score covariance matrix
\[
\boldsymbol{\Omega}_{j}^{\ast}(S)
=
\sum_{h=-\infty}^{\infty}
\mathbb{E}\!\left[
\boldsymbol{s}_{j,t}(\boldsymbol{\vartheta}_{j}^{\ast}(S);S)
\boldsymbol{s}_{j,t-h}(\boldsymbol{\vartheta}_{j}^{\ast}(S);S)^{\prime}
\right]
\]
exists and is nonsingular uniformly, where
\(
\boldsymbol{s}_{j,t}(\boldsymbol{\vartheta}_{j};S)
=
\frac{\partial l_t(\eta_{j,t}(\boldsymbol{\vartheta}_{j};S),y_t)}
{\partial\boldsymbol{\vartheta}_{j}}.
\)
The sandwich matrix
\(
\boldsymbol{V}_{j}^{\ast}(S)
=
\boldsymbol{J}_{j}^{\ast}(S)^{-1}
\boldsymbol{\Omega}_{j}^{\ast}(S)
\boldsymbol{J}_{j}^{\ast}(S)^{-1}
\)
is finite, and its candidate-coefficient diagonal element
$V_{\theta\theta,j}^{\ast}(S)$ is uniformly bounded away from zero and infinity.
If the stagewise scores are serially uncorrelated, the displayed long-run matrix 
$\boldsymbol{\Omega}_{j}^{\ast}(S)$ reduces to the usual one-period  covariance.

\item[A3.4] (High-level uniform Wald approximation.) For the admissible
class $\mathcal A_T$ defined above, define
\[
\mathcal E_T
=
\max_{(S,j)\in\mathcal A_T}
\left|W_{T,j}(S)-\mathrm{NC}_{j}^{\ast}(S)\right|.
\]
The following uniform approximation rates hold. Under Assumption A2-ET,
\(\mathcal E_T=O_p(\sqrt{\ell_T}).\)
Under Assumption A2-PM, there exists a deterministic maximal-deviation rate
$a_{T,p}^{\mathrm{PM}}$ such that
\(\mathcal E_T=O_p(a_{T,p}^{\mathrm{PM}}).\)
Primitive sufficient conditions and explicit rates are given in
Appendix~\ref{app:primitive_wald}.
\end{enumerate}
\end{assumption}

For a given conditioning set $S$, define the population \emph{normalised
conditional strength}, or population Wald noncentrality, of covariate
$x_{j,t}$ in the augmented submodel with index \eqref{eq:aug_index} as
\begin{equation}
\mathrm{NC}_{j}^{\ast}(S)
=
\frac{\sqrt{T}\,|\theta_{j}^{\ast}(S)|}
{\{V_{\theta\theta,j}^{\ast}(S)\}^{1/2}}.
\label{eq:nc_def_prim}%
\end{equation}
The factor $\sqrt{T}$ is essential: $\mathrm{NC}_{j}^{\ast}(S)$ is the
population analogue of the absolute Wald statistic, not the unscaled
pseudo-true coefficient. When the particular augmented submodel is correctly specified and its
observations are conditionally independent,
$V_{\theta\theta,j}^{\ast}(S)$ reduces to the corresponding diagonal
element of the inverse information matrix.

\begin{assumption}
\label{ass:A4}
There exists a deterministic sequence $d_T>0$ such that
\begin{equation}
\frac{\mathcal E_T}{d_T}=o_p(1),
\label{eq:dom_error_margin}%
\end{equation}
and, for any conditioning set $S$ satisfying
$S^{(0)}\subseteq S\subseteq S^{(0)}\cup S_0$ and $|S\cap S_0|<k$,
\begin{equation}
\min_{j\in S_{0}\setminus S}\mathrm{NC}_{j}^{\ast}(S)
\ge
\max_{\substack{m\notin S_{0}\\m\notin S}}
\mathrm{NC}_{m}^{\ast}(S)+d_T.
\label{eq:stg_dom}%
\end{equation}
In addition, every remaining signal is above the testing threshold by the same
population margin:
\begin{equation}
\min_{j\in S_{0}\setminus S}\mathrm{NC}_{j}^{\ast}(S)
\ge c_T+d_T.
\label{eq:stg_threshold}%
\end{equation}
\end{assumption}

\paragraph{Remark.}
Because BMT adds one covariate at each stage, the exact-recovery induction also
holds under the weaker requirement that, for every admissible $S$, the
strongest remaining signal satisfies
\(\max_{j\in S_0\setminus S}\mathrm{NC}_j^*(S)\ge
\max_{m\notin S_0,\,m\notin S}\mathrm{NC}_m^*(S)+d_T\)
and
\(\max_{j\in S_0\setminus S}\mathrm{NC}_j^*(S)\ge c_T+d_T\),
with the same condition $\mathcal E_T/d_T=o_p(1)$. Assumption~\ref{ass:A4}
uses the minimum over remaining signals and is therefore stronger: it ensures
that every remaining signal dominates every non-signal.

\begin{assumption}
\label{ass:A5}
The threshold sequence $\{c_T\}$ used in the nonlinear BMT procedure satisfies
$c_T\rightarrow\infty$ and $c_T=o(\sqrt T)$. Moreover, there exists a constant
$\bar{\eta}_{0}\in(0,1)$ such that
\begin{equation}
\Pr\!\left(\mathcal E_T\le (1-\bar{\eta}_{0})c_T\right)\rightarrow1,\quad\mathrm{as}\text{ }T\rightarrow\infty.
\label{eq:threshold_calibration}%
\end{equation}
\end{assumption}

Assumption \ref{ass:A1} replaces pointwise ergodicity by dependence
conditions that support maximal inequalities. Under A2-PM, the displayed
summability condition is satisfied, for example, if
$\sup_T\alpha_T(h)\le C h^{-a}$ with
$a>q(q+\delta)/(2\delta)$. The fixed values of $s_0$ and
$k_{\max}$ imply that every augmented GLM has dimension at most
$d_0=s_0+k_{\max}+1$. Moreover,
\[
|\mathcal A_T|
\le p\sum_{\ell=0}^{k_{\max}}\binom p\ell
\le Cp^{k_{\max}+1}=Cp^K.
\]
Thus the high-dimensional cost is a union bound over polynomially many
fixed-dimensional likelihood problems, rather than control over all sparse
directions.

Assumption \ref{ass:A2} supplies envelopes for the score, Hessian, score outer
product, its first parameter derivative and the third likelihood derivative.
For logit links, the scalar score and curvature weights are bounded. Thus, A2-ET
follows, for example, from uniformly sub-Gaussian covariates, while A2-PM
follows from sufficiently high covariate moments. For probit links, the Mills-ratio
terms are unbounded away from compact index sets; a compact-index condition or
correspondingly stronger covariate tails are therefore required.

Assumption \ref{ass:A3} remains stated at the level used by the recovery
proofs. Appendix~\ref{app:primitive_wald} shows that, for information standard
errors under the information equality or one-period sandwich standard errors
with serially uncorrelated scores, A3.4 follows from A1--A3.3, the envelopes in
A2 and the mild parameter-set regularity stated there. Under A2-ET, if
$\ell_T^3/T\to0$, then $\mathcal E_T=O_p(\sqrt{\ell_T})$. Under A2-PM, the
appendix gives $a_{T,p}^{\mathrm{PM}}=p^{K/q}$ and requires
$p^{K/q}=o(\sqrt T)$. The high-level formulation is retained because HAC
standard errors add a bandwidth-dependent long-run-variance estimation error.
This approximation includes nuisance re-estimation, Hessian approximation and
standard-error estimation; pointwise relative consistency of
$\widehat{\mathrm{se}}_j(S)$ is not sufficient when population
noncentralities diverge.

Assumption \ref{ass:A4}, which is only required for exact model recovery,
formalises stagewise separation with the correct stochastic margin. The
population dominance gap $d_T$ must dominate the maximal sample approximation
error $\mathcal E_T$, and the remaining signals must also clear the testing
threshold. A fixed positive population-coefficient gap is sufficient in common
fixed-signal settings because the Wald noncentrality in \eqref{eq:nc_def_prim}
contains the factor $\sqrt T$, but the theorem is stated in the more general
rate form needed for high-dimensional testing.

Assumption \ref{ass:A5} ensures that false positives are asymptotically eliminated. It requires the threshold to dominate the maximal uniform stochastic approximation
error. Bonferroni or Gaussian-max thresholds must, therefore, be chosen with sufficient
slack for \eqref{eq:threshold_calibration} to hold; the boundary choice
$c_T=\sqrt{2\log p}$ need not be sufficient without such slack.

\section{Approximating-model results}

\label{sec:approx}
%========================

This section contains results relating to the behaviour of the test statistics
and the stopping rule that form the basis of the nonlinear BMT procedure. The
results are formulated at the level of the \emph{approximating model}, that
is, a model whose covariates include all signals as well as one or more of the
pseudo-signals. No dominance or exact-recovery assumptions are imposed and the
conditioning set may be misspecified.

%------------------------

%\subsection{Uniform control of stagewise statistics}

%\label{sec:thm1}
%------------------------

The first result establishes uniform stochastic control of the stagewise Wald
test statistics along the BMT path.

\begin{theorem}
\label{thm:approx_stats} Suppose Assumptions~\ref{ass:A1}--\ref{ass:A3} hold.
Let $\{S^{(\ell)}\}_{\ell\geq0}$ denote any sequence of conditioning sets that
can be generated by the nonlinear BMT procedure, with $|S^{(\ell)}\setminus
S^{(0)}|\leq k_{\max}$ for all $\ell$. Then, as $T\rightarrow\infty$, the
following statements hold. \newline\noindent\textbf{(a) Exponential-tail
regime.} If Assumption A2-ET holds, then
\begin{equation}
\max_{\ell\leq k_{\max}}\max_{j\notin S^{(\ell)}}\bigl|W_{T,j}(S^{(\ell
)})-\mathrm{NC}_{j}^{\ast}(S^{(\ell)})\bigr|=O_{p}\!\left(\sqrt{\ell_T}\right).\label{eq:thm1_ET}%
\end{equation}
\noindent\textbf{(b) Polynomial-moment regime.} If Assumption A2-PM holds,
then
\begin{equation}
\max_{\ell\leq k_{\max}}\max_{j\notin S^{(\ell)}}\bigl|W_{T,j}(S^{(\ell
)})-\mathrm{NC}_{j}^{\ast}(S^{(\ell)})\bigr|=O_{p}(a_{T,p}^{\mathrm{PM}%
}).\label{eq:thm1_PM}%
\end{equation}
\noindent\textbf{(c) Implication for population-null covariates.} If, in
addition, Assumption~\ref{ass:A5} holds, then under either tail regime
\begin{equation}
\Pr\!\left(  \exists\,\ell\leq k_{\max},\ \exists\,j\notin S^{(\ell
)}:\mathrm{NC}_{j}^{\ast}(S^{(\ell)})=0\ \text{and}\ W_{T,j}(S^{(\ell)})\geq
c_{T}\right)  \rightarrow0.\label{eq:thm1_null}%
\end{equation}

\end{theorem}

%------------------------

\paragraph{Comments.}

\label{sec:thm1_discussion}
%------------------------

Theorem~\ref{thm:approx_stats} provides a uniform approximation result for the
nonlinear BMT statistics along the entire selection path. Notably, no
assumptions are made regarding the relative magnitude of the population
strengths $\mathrm{NC}_{j}^{\ast}(S)$ across covariates. The theorem,
therefore, applies to arbitrary approximating models, including cases in which
strong pseudo-signal effects are present. 

Parts (a) and (b) establish that, under either exponential-tail or
polynomial-moment conditions, the stagewise Wald statistics
concentrate uniformly around their population counterparts. Part (c) shows
that the multiple-testing threshold eliminates covariates whose population
conditional strength is zero, uniformly across stages. 
\bigskip%

%========================

%\subsection{No-null selection and stopping under exhaustion}

%\label{sec:theorem2}
%========================

The results in Theorem~\ref{thm:approx_stats} form the basis for the stopping and no-false-inclusion properties of the nonlinear BMT procedure, to which we turn our attention to next. Specifically, we consider two key properties: (i) covariates with \emph{zero population conditional strength}
are not selected (uniformly along the path); (ii) the procedure stops once the
approximating model is \emph{exhausted}, in the sense that no remaining
covariate has population conditional strength large enough to cross the
multiple-testing threshold. 

Before stating the relevant theorem, recall that, for a conditioning set $S$ and a candidate $j\notin S$,
$\mathrm{NC}_{j}^{\ast}(S)$ denotes the population normalised conditional
strength of $x_{j,t}$ (see \eqref{eq:nc_def_prim}). Define the
\emph{population-null} set at $S$ as
\begin{equation}
\mathcal{N}^{\ast}(S)=\{j\notin S:\ \mathrm{NC}_{j}^{\ast}(S)=0\},
\label{eq:nullset_def}%
\end{equation}
and the maximal remaining population strength at stage $\ell$ as
\begin{equation}
M_{\ell}^{\ast}=\max_{j\notin S^{(\ell)}}\mathrm{NC}_{j}^{\ast}\bigl(S^{(\ell
)}\bigr). \label{eq:Mell_def}%
\end{equation}

%------------------------

\begin{theorem}
\label{thm:approx_stop}
Suppose Assumptions~\ref{ass:A1}--\ref{ass:A3} and \ref{ass:A5} hold. Let
$\widehat L\leq k_{\max}$ denote the number of successful selections made
before the nonlinear BMT procedure stops, with $\widehat L=k_{\max}$ if the
prespecified stage cap is reached, and let $\widehat S=S^{(\widehat L)}$.
For notational convenience, extend the path after stopping by setting
$S^{(\ell)}=\widehat S$ for $\ell\geq\widehat L$. Then, as
$T\rightarrow\infty$:

\begin{enumerate}
\item[\textup{(i)}] (\emph{No selection of population-null covariates,
uniformly in stage})
\begin{equation}
\Pr\!\left(
\exists\,\ell\leq k_{\max},\ \exists\,j\in
\mathcal N^*\!\left(S^{(\ell)}\right)
\text{ such that BMT selects }j\text{ at stage }\ell
\right)\rightarrow0.
\label{eq:thm2_nonull}%
\end{equation}

\item[\textup{(ii)}] (\emph{Stopping under approximating-model exhaustion}) If
there exists a possibly random stage $\ell^\dagger\le k_{\max}$ such that
\begin{equation}
M_{\ell^\dagger}^*=o_p(c_T),
\label{eq:exhaust_general}%
\end{equation}
then
\begin{equation}
\Pr(\widehat L\le \ell^\dagger)\rightarrow1.
\label{eq:thm2_stop}%
\end{equation}
\end{enumerate}
Part~\textup{(i)} may equivalently be written along the realised selected
path as
\begin{equation}
\Pr\!\left(
\exists\,\ell<\widehat L:
j^{(\ell)}\in\mathcal N^*\!\left(S^{(\ell)}\right)
\right)\rightarrow0.
\label{eq:thm2_path_nonull}%
\end{equation}
\end{theorem}

%------------------------

\paragraph{Comments.}

\label{sec:thm2_comments}
%------------------------

Part (i) in both regimes is the formal \textquotedblleft
no-null\textquotedblright\ property: variables whose population conditional
strength is zero are not selected, uniformly over all stages $\ell\leq
k_{\max}$. The proof is a direct union-bound argument over $\ell$ combined
with Theorem~\ref{thm:approx_stats}(c), exploiting that $k_{\max}$ is fixed.

Part (ii) formalises the stopping logic in an approximating-model sense. The
condition $M_{\ell^{\dagger}}^{\ast}=o_p(c_{T})$ means that by stage
$\ell^{\dagger}$ the remaining population conditional strengths are
asymptotically too small to cross the multiple-testing threshold. Since
Theorem~\ref{thm:approx_stats} controls the maximal deviation of
$W_{T,j}(S^{(\ell)})$ from $\mathrm{NC}_{j}^{\ast}(S^{(\ell)})$ uniformly in
stage, it follows that the maximal \emph{sample} statistic falls below $c_{T}$
at (or before) $\ell^{\dagger}$ with probability tending to one.

%========================

\section{Exact recovery and oracle inference}

\label{sec:exact}
%========================

This section provides results relating to \emph{exact recovery} of single or
multiple signals by the nonlinear BMT algorithm. The results rely on suitable
dominance conditions and yield statements about the recovery of the true
model, that is the model containing all signals and no pseudo-signals or noise
variables. Conditional on such results, the \emph{oracle properties} 
of the post-selection MLE of the GLM parameters can be established.

%------------------------

\subsection{Single signal}

Recall that, at any stage $\ell\geq0$, the BMT procedure compares the
stagewise statistics $\{W_{T,j}(S^{(\ell)}):j\notin S^{(\ell)}\}$ and selects
the largest one among those exceeding the threshold. Exact recovery in the
single-signal case requires that the population conditional strength of the
signal dominates that of all other covariates uniformly in the conditioning
set. This is formalised in the following dominance condition.

\setcounter{assumption}{0} \renewcommand{\theassumption}{A4-\arabic{assumption}}\renewcommand{\theHassumption}{A4.\arabic{assumption}}

\begin{assumption}
\label{ass:A41}
There exists a deterministic sequence $d_{1T}>0$ such that
\begin{equation}
\frac{\mathcal E_T}{d_{1T}}=o_p(1),
\label{eq:dom_k1_error}%
\end{equation}
and, for every admissible conditioning set $S$ satisfying
$S^{(0)}\subseteq S\subseteq S^{(0)}\cup\{j_0\}$ and $j_0\notin S$,
\begin{equation}
\mathrm{NC}_{j_{0}}^{\ast}(S)
\ge
\max_{\substack{m\neq j_0\\m\notin S}}
\mathrm{NC}_{m}^{\ast}(S)+d_{1T},
\label{eq:dom_k1}%
\end{equation}
and
\begin{equation}
\mathrm{NC}_{j_{0}}^{\ast}(S)
\ge c_T+d_{1T}.
\label{eq:dom_k1_threshold}%
\end{equation}
\end{assumption}

Under the restriction in Assumption~\ref{ass:A41}, the conditioning set before
the signal is selected is $S^{(0)}$. The slightly more general notation is
retained to match the stagewise formulation used below.

We have the following exact-recovery result in this setting.

\begin{theorem}
\label{thm:k1_recovery} Suppose
Assumptions~\ref{ass:A1}--\ref{ass:A3}, \ref{ass:A41} and \ref{ass:A5} hold,
$k=|S_0|=1$, with $S_0=\{j_0\}$, and $k_{\max}\ge1$ in the nonlinear BMT
procedure. Then, as $T\rightarrow\infty$:
\medskip\noindent\newline\textbf{(a)} With probability
tending to one, the first-stage passing set is nonempty and the first selected
covariate is the signal:
\begin{equation}
\Pr\!\left(j^{(0)}=j_0\right)\rightarrow1.
\label{eq:thm3_first}%
\end{equation}

\medskip\noindent\textbf{(b)} Consequently,
\begin{equation}
\Pr\!\left(j_0\in\widehat S\right)\rightarrow1.
\label{eq:thm3_include}%
\end{equation}

\medskip\noindent\textbf{(c)} If, in addition,
\begin{equation}
\max_{\substack{m\neq j_0\\m\notin S^{(0)}}}
\mathrm{NC}_m^*\!\left(S^{(0)}\cup\{j_0\}\right)=o(c_T),
\label{eq:k1_exhaustion}%
\end{equation}
then
\begin{equation}
\Pr\!\left(\widehat S=S^{(0)}\cup\{j_0\}\right)\rightarrow1.
\label{eq:thm3_exact}%
\end{equation}
\end{theorem}

%------------------------

\paragraph{Proof outline.}

\label{sec:thm3_proofoutline}
%------------------------

We sketch the key steps in the proof here; a complete proof is given in
Appendix~\ref{app:thm3}.

\begin{enumerate}
\item \emph{Population separation with the correct margin}. By
\eqref{eq:dom_k1}, for any admissible $S$ not containing the signal,
$\mathrm{NC}_{j_0}^*(S)$ exceeds every competing population noncentrality by
at least $d_{1T}$.

\item \emph{Uniform sample approximation}. By definition of $\mathcal E_T$,
\[
W_{T,j_0}(S)\ge \mathrm{NC}_{j_0}^*(S)-\mathcal E_T,
\qquad
W_{T,m}(S)\le \mathrm{NC}_{m}^*(S)+\mathcal E_T.
\]
Since $\mathcal E_T/d_{1T}=o_p(1)$, the sample ordering inherits the
population ordering with probability tending to one.

\item \emph{Threshold crossing}. The additional condition
\eqref{eq:dom_k1_threshold} ensures that the signal is not only the largest
sample statistic but also exceeds the multiple-testing threshold.

\item \emph{Stopping after selecting the signal}. Condition
\eqref{eq:k1_exhaustion} ensures that, after the signal is included, all
remaining population noncentralities are asymptotically below the threshold.
The stopping theorem then rules out further selections.
\end{enumerate}

%========================

\subsection{Multiple signals}

\label{sec:theorem4}
%========================

Next, we extend the single-signal recovery result to the general fixed-$k$
case. The key additional difficulty is that, after some signals are selected,
the remaining signal variables must continue to dominate all pseudo-signal and
noise covariates conditionally on the currently selected set. This is captured
by the stagewise dominance condition of Assumption~\ref{ass:A4}.

The following theorem is the nonlinear analogue of the multiple-signal
separation and exact recovery theorem established in linear BMT analysis \citep[cf.][]{kapetanios2026}.

%------------------------

\begin{theorem}
\label{thm:fixedk_recovery} Suppose
Assumptions~\ref{ass:A1}--\ref{ass:A5} hold and $k_{\max}\geq k\geq1$ in the
nonlinear BMT procedure. Then, the following statements are true as
$T\rightarrow\infty$. \medskip\noindent\newline\textbf{(a) No false selections
before exhausting }$S_{0}$\textbf{:} With probability tending to one, at each
stage $\ell=0,1,\dots,k-1$,
\begin{equation}
j^{(\ell)}\in S_{0}\setminus S^{(\ell)}, \label{eq:thm4_true_each_stage}%
\end{equation}
i.e.,\ the covariate selected at stage $\ell$ is a previously unselected
signal. In particular,
\begin{equation}
\Pr\!\left(  S^{(k)}=S^{(0)}\cup S_{0}\right)  \ \rightarrow\ 1,
\label{eq:thm4_after_k}%
\end{equation}
where $S^{(k)}$ denotes the set after $k$ successful selections.
\medskip\noindent\newline\textbf{(b) Exact model recovery and stopping:}
Assume additionally that
\begin{equation}
\max_{\substack{m\notin S_{0}\\m\notin S^{(0)}}}
\mathrm{NC}_{m}^{\ast}\bigl(S^{(0)}\cup S_{0}\bigr)=o(c_T),
\label{eq:thm4_null_after_true}%
\end{equation}
i.e., once all signals are included, the remaining population conditional
strengths are asymptotically below the testing threshold. Then,
\begin{equation}
\Pr\!\left(  \widehat{S}=S^{(0)}\cup S_{0}\right)  \rightarrow\ 1,
\label{eq:thm4_exact}%
\end{equation}
i.e.,\ nonlinear BMT selects all and only the true signals and stops after $k$
stages (up to the always-in controls).
\end{theorem}

%------------------------

\paragraph{Comments.}

\label{sec:thm4_comments}
%------------------------

Part (a) is the central stagewise recovery statement: the algorithm selects a
signal at each of the first $k$ stages, hence it necessarily includes all
signals after $k$ selections. The nontrivial requirement is the uniform dominance gap in \eqref{eq:stg_dom},
which must dominate the maximal sample approximation error and must be large
enough for the remaining signals to cross the threshold at every intermediate
stage.

Part (b) adds a standard \textquotedblleft
no-residual-signal\textquotedblright\ condition after $S_{0}$ is included.
Under correct specification of the single-index GLM, the condition is
satisfied with equality to zero: once $S_{0}$ is included, any covariate
$m\notin S_{0}$ has zero population conditional strength in the augmented
model. The weaker $o(c_T)$ condition is sufficient for stopping because the
threshold dominates the sample approximation error.

%------------------------

\paragraph{Proof outline.}

\label{sec:thm4_proofoutline}
%------------------------

A complete proof is given in Appendix~\ref{app:thm4}; here, we summarise the
main argument.

\begin{enumerate}
\item \emph{Induction over stages}. Suppose that after $\ell<k$ stages the
selected non-always-in covariates are all signals, so that
$S^{(\ell)}\subseteq S^{(0)}\cup S_0$ and
$|S^{(\ell)}\cap S_0|=\ell$.

\item \emph{Population dominance at stage $\ell$}. Assumption~\ref{ass:A4}
gives a population gap of at least $d_T$ between every remaining signal and
every remaining non-signal, and it also ensures that the remaining signals
cross the threshold by at least $d_T$.

\item \emph{Uniform sample approximation}. Since $\mathcal E_T/d_T=o_p(1)$,
the sample Wald statistics preserve the population ordering with probability
tending to one. Hence a non-signal cannot maximise the stagewise statistic at
any stage before $S_0$ is exhausted, and the selected covariate must be a new
signal.

\item \emph{Stopping after exhaustion}. Once $S^{(k)}=S^{(0)}\cup S_0$,
condition \eqref{eq:thm4_null_after_true} and Theorem~\ref{thm:approx_stop}
imply that all remaining sample statistics are below $c_T$ with probability
tending to one. The algorithm therefore stops without false inclusions.
\end{enumerate}

%========================

\subsection{Oracle estimation and inference}

\label{sec:oracle}
%========================

We now establish the oracle properties of the post-nonlinear-BMT
estimator: conditional on exact model recovery, as established in
Theorem~\ref{thm:fixedk_recovery}, the post-selection MLE of the parameters of
the GLM behaves asymptotically (to first order) as if the true model that
contains only signal variables were known in advance. 
%------------------------

%\subsection{Post-selection estimator and oracle equivalence}

%\label{sec:post_mle}
%------------------------

Letting $\widehat{S}$ denote the set selected by the nonlinear BMT procedure, write
\(S^\dagger=S^{(0)}\cup S_0.\)
Define $\widehat{\boldsymbol\beta}_{\rm post}\in\mathbb R^p$ by fitting the
(quasi-) MLE on $\widehat S$ and setting all coordinates outside
$\widehat S$ equal to zero. Write
$\widehat{\boldsymbol\beta}_{{\rm post},S^\dagger}$ for its fixed-dimensional
$S^\dagger$ subvector. Let $\boldsymbol\beta_{0,S^\dagger}$ denote the true
parameter vector restricted to the oracle support $S^\dagger$, including the
always-in controls, and write $\bar d=|S^\dagger|$. The infeasible oracle estimator is computed
using the same maximisation and deterministic tie-breaking rule as the
post-BMT refit, so the two estimators coincide whenever
$\widehat S=S^\dagger$.

\setcounter{assumption}{5} \renewcommand{\theassumption}{A\arabic{assumption}}\renewcommand{\theHassumption}{A.\arabic{assumption}}

\begin{assumption}
\label{ass:A6}
The oracle criterion over $S^\dagger$ has a unique pseudo-true value
$\boldsymbol\beta_{0,S^\dagger}$. Its score
\[
\boldsymbol\psi_{0,t}
=
\frac{\partial}{\partial\boldsymbol\beta_{S^\dagger}}
 l_t(\boldsymbol x_{S^\dagger,t}'\boldsymbol\beta_{0,S^\dagger},y_t)
\]
satisfies
\[
T^{-1/2}\sum_{t=1}^T\boldsymbol\psi_{0,t}
\Rightarrow N(\boldsymbol{0},\boldsymbol\Omega_0),
\]
where $\boldsymbol\Omega_0$ is the relevant long-run score covariance matrix ($\Rightarrow$ signifies convergence in distribution as
$T\rightarrow\infty$).
The oracle Hessian
\[
\boldsymbol J_0
=
-\mathbb E\left[
\frac{\partial^2}{\partial\boldsymbol\beta_{S^\dagger}\partial
\boldsymbol\beta_{S^\dagger}'}
l_t(\boldsymbol x_{S^\dagger,t}'\boldsymbol\beta_{0,S^\dagger},y_t)
\right]
\]
is nonsingular, and the oracle estimator $\widehat{\boldsymbol\beta}_{\rm or}$
satisfies the asymptotic linear expansion
\[
\sqrt T(\widehat{\boldsymbol\beta}_{\rm or}-\boldsymbol\beta_{0,S^\dagger})
=
\boldsymbol J_0^{-1}T^{-1/2}\sum_{t=1}^T\boldsymbol\psi_{0,t}+o_p(1).
\]
\end{assumption}

We have the following result.

%------------------------

\begin{theorem}
\label{thm:oracle} Suppose
Assumption~\ref{ass:A6} holds and
\(\Pr(\widehat S=S^\dagger)\rightarrow1\) as $T\rightarrow\infty$.
Then,
\begin{equation}
\widehat{\boldsymbol\beta}_{{\rm post},S^\dagger}-\boldsymbol\beta_{0,S^\dagger}
=o_p(1),
\label{eq:oracle_consistency}%
\end{equation}
and
\begin{equation}
\sqrt{T}\bigl(\widehat{\boldsymbol\beta}_{{\rm post},S^\dagger}-
\boldsymbol\beta_{0,S^\dagger}\bigr)
\Rightarrow
N\!\left(\boldsymbol{0},\boldsymbol J_0^{-1}\boldsymbol\Omega_0\boldsymbol J_0^{-1}\right).
\label{eq:oracle_an}%
\end{equation}
Consequently, the post-BMT estimator has the same first-order limiting
distribution as the infeasible oracle estimator that knows $S^\dagger$ in advance.
\end{theorem}

%------------------------

\paragraph{Proof outline.}

\label{sec:thm5_proofsketch}
%------------------------

On the event $\{\widehat S=S^\dagger\}$, which occurs with probability
tending to one, the fixed-dimensional post-selection subvector equals the
oracle estimator:
$\widehat{\boldsymbol\beta}_{{\rm post},S^\dagger}
=\widehat{\boldsymbol\beta}_{\rm or}$.
Assumption~\ref{ass:A6} gives the fixed-dimensional asymptotic linear
representation and central limit theorem for the oracle estimator. The
high-probability equality between the post-BMT and oracle estimators then
transfers the oracle limit distribution to the post-BMT estimator. A complete
proof is given in Appendix~\ref{app:thm5}.

%------------------------

\subsection{Logit and probit models}

\label{sec:logit_probit}
%------------------------

We verify that the conditions of Theorems~\ref{thm:approx_stats}%
--\ref{thm:oracle} are satisfied for logit and probit models, where the (inverse) link function is the logistic distribution function $\Lambda(z)=1/(1+e^{-z}%
)$ and the standard normal distribution function $\Phi$, respectively.

\begin{lemma}
\label{lem:logit_probit} Let $G=\Lambda$ or $G=\Phi$.
Suppose Assumptions~\ref{ass:A1}--\ref{ass:A2} and conditions A3.2--A3.3
hold. Suppose also that the pseudo-true parameters of all admissible stagewise
submodels lie in a common compact set, that their associated indices lie in a
fixed compact interval, and that the additional parameter-set, growth and
standard-error conditions of Theorem~\ref{thm:primitive_wald} hold. Then A3.4
holds with the rates stated in that theorem, and hence
Assumption~\ref{ass:A3} holds. If, in addition, the postulated conditional
mean is correctly specified and
$S\supseteq S^{\dagger}=S^{(0)}\cup S_0$, then
$\theta_m^*(S)=0$ and hence $\mathrm{NC}_m^*(S)=0$ for every candidate
$m\notin S$.
\end{lemma}

\paragraph{Comment.}

For logit, the scalar score and curvature weights are bounded, so the
derivative-envelope conditions in Assumption~\ref{ass:A2} follow from the
corresponding covariate tails. For probit, compactness of the index controls
the Mills-ratio terms. Conditions A3.2--A3.3 retain their population
identification and nonsingularity roles, while Appendix~\ref{app:primitive_wald}
provides the uniform Wald approximation. The vanishing of
$\mathrm{NC}_{m}^{\ast}(S)$ after inclusion of $S^{\dagger}$ follows from
correct specification of the conditional mean. \medskip

Together with Theorems~\ref{thm:approx_stats}--\ref{thm:oracle}, the
results in Lemma~\ref{lem:logit_probit} show that the nonlinear BMT procedure
results in consistent model selection and oracle inference for
high-dimensional binary-response models under transparent and verifiable conditions.

%========================

\section{Primitive conditions for stagewise dominance}

\label{sec:primitive}

In this section, we provide primitive sufficient conditions for the stagewise
dominance conditions (\ref{eq:stg_dom}) and (\ref{eq:dom_k1}) required for the
exact-recovery results of Theorems~\ref{thm:k1_recovery} and
\ref{thm:fixedk_recovery}. We consider two distinct routes relying on
assumptions that serve different empirical and conceptual purposes. \emph{The
local (small-index) regime} (Assumptions \ref{ass:L1}--\ref{ass:L3} below) is
most appropriate when the single-index remains uniformly moderate, so that the
nonlinear GLM admits a stable linear approximation --- in this case dominance
reduces to a partial-correlation gap. The \emph{global regime} is designed for settings with strong dependence
and potential proxy covariates, where dominance must be preserved after
conditioning on previously selected variables. It imposes stagewise
\emph{conditional proxy weakness} through Assumption~\ref{ass:GM2}.
Assumption~\ref{ass:G1} gives a primitive moment-based sufficient condition,
whereas Assumption~\ref{ass:G2} is a partial-projection diagnostic that is
useful when additional model-specific restrictions link it to
Assumption~\ref{ass:G1}.

The local conditions and the global-regime condition are independent
alternatives. Either Assumptions~\ref{ass:L1}--\ref{ass:L3}, or
Assumption~\ref{ass:GM2} verified through an appropriate primitive route, is
sufficient for stagewise dominance. Theorems~\ref{thm:k1_recovery} and
\ref{thm:fixedk_recovery} require only one of these alternatives. Recalling the definition of the
population normalised conditional strength $\mathrm{NC}_{j}^{\ast}(S)$ of
covariate $x_{j,t}$, given in \eqref{eq:nc_def_prim}, we note that a dominance
condition such as Assumption \ref{ass:A4} requires that, for any admissible
conditioning set $S$ with $|S\cap S_{0}|\leq k-1$, the remaining signals
dominate all non-signals in $\mathrm{NC}_{j}^{\ast}(S)$ by a uniform gap.

%------------------------

\subsection{Local regime linearisation}

\label{sec:local_regime}
%------------------------

The local regime exploits the fact that, when the index $\boldsymbol{x}%
_{t}^{\prime}\boldsymbol{\beta}_{0}$ is uniformly small, the (inverse) link $G$ can be
linearised and the GLM behaves like a weighted linear model. We make the
following three assumptions.

\setcounter{assumption}{0} \renewcommand{\theassumption}{L\arabic{assumption}}\renewcommand{\theHassumption}{L.\arabic{assumption}}

\begin{assumption}
\label{ass:L1} There are a constant $\delta_0>0$ and a deterministic sequence
$0<\delta_T\leq\delta_0$ such that
\[
\sup_{1\leq t\leq T}|\boldsymbol{x}_{t}^{\prime}\boldsymbol{\beta}_{0}|
\leq\delta_T
\]
with probability tending to one. The sequence $\delta_T$ may be fixed at a
sufficiently small value or may converge to zero in a local triangular array.
\end{assumption}

\begin{assumption}
\label{ass:L2} On $[-\delta_0,\delta_0]$, $G$
satisfies $0<\underline{g}\leq G^{\prime}(z)\leq\bar{g}<\infty$ and
$|G^{\prime\prime}(z)|\leq M<\infty$ for some constants $\underline{g}$,
$\bar{g}$ and $M$.
\end{assumption}

\begin{assumption}
[Scaled local partial-correlation gap]\label{ass:L3} For every admissible
conditioning set $S^{(0)}\subseteq S\subseteq S^{(0)}\cup S_{0}$ with
$|S\cap S_{0}|\leq k-1$, let $r_{j,t}(S)$ denote the residual from the
least-squares projection of $x_{j,t}$ on $\boldsymbol{x}_{S,t}$ and put
\(M_j(S)=\left|\mathbb E[r_{j,t}(S)\boldsymbol x_t'\boldsymbol\beta_0]\right|.\)
There exist constants $0<C_L\leq C_U<\infty$ and a deterministic sequence
$\underline{\kappa}_T>0$ such that
\begin{equation}
C_L\min_{j\in S_0\setminus S}M_j(S)
\geq
C_U\max_{\substack{m\notin S_0\\m\notin S}}M_m(S)
+\underline{\kappa}_T.
\label{eq:local_gap}%
\end{equation}
\end{assumption}

The next lemma provides the link between the local regime and stagewise dominance.

\begin{lemma}
[Local regime implies dominance]\label{lem:local_implies_A4} Under Assumptions
A3.3 and \ref{ass:L1}--\ref{ass:L3}, suppose that the uniform local expansion
takes the form
\begin{equation}
\sup_{\substack{S^{(0)}\subseteq S\subseteq S^{(0)}\cup S_0\\j\notin S}}
\left|
\frac{|\theta_j^*(S)|}{\{V_{\theta\theta,j}^*(S)\}^{1/2}}
-C_j(S)M_j(S)
\right|
\leq r_T,
\qquad
C_L\leq C_j(S)\leq C_U,
\label{eq:local_uniform_expansion}
\end{equation}
where $r_T$ is deterministic and
$r_T=o(\underline\kappa_T)$. If
$c_T+\mathcal E_T=o_p(\sqrt T\,\underline\kappa_T)$, then
Assumption~\ref{ass:A4} holds with
$d_T=\sqrt T\,\underline\kappa_T/4$.
\end{lemma}

\paragraph{Comment.}

The local regime reduces dominance to a \emph{linear} gap in partial
correlations with the index. In binary models, it is most plausible
when the response probability stays away from $0$ and $1$ and the true
coefficients are small enough that $\boldsymbol{x}_{t}^{\prime}%
\boldsymbol{\beta}_{0}$ remains uniformly moderate. Assumption~\ref{ass:L1}
is therefore either a fixed small-index condition for bounded designs or a
local triangular-array condition. The rate requirement
$r_T=o(\underline\kappa_T)$ makes explicit that the linearisation error must be
small relative to the partial-correlation gap.

%------------------------

\subsection{Global regime proxy weakness}

\label{sec:global_regime}
%------------------------

Next, we formalise the global regime condition (Assumption \ref{ass:GM2}),
then give a moment-based sufficient condition and a complementary
partial-projection diagnostic.

\setcounter{assumption}{0} \renewcommand{\theassumption}{GM\arabic{assumption}}\renewcommand{\theHassumption}{GM.\arabic{assumption}}

\begin{assumption}
\label{ass:GM2} There exists $\rho\in(0,1)$
such that, for any admissible $S^{(0)}\subseteq S\subseteq S^{(0)}\cup S_0$
with $|S\cap S_0|\le k-1$,
\begin{equation}
\max_{\substack{m\notin S_0\\m\notin S}}
\frac{|\theta_m^*(S)|}{\{V_{\theta\theta,m}^*(S)\}^{1/2}}
\le
\rho
\min_{j\in S_0\setminus S}
\frac{|\theta_j^*(S)|}{\{V_{\theta\theta,j}^*(S)\}^{1/2}}.
\label{eq:GM2_theta}%
\end{equation}
In addition, the remaining signals are strong enough in Wald-noncentrality
scale:
\begin{equation}
b_T
:=
\inf_{S}
\min_{j\in S_0\setminus S}\mathrm{NC}_j^*(S)
\quad\text{satisfies}\quad
\frac{c_T+\mathcal E_T}{b_T}=o_p(1),
\label{eq:GM2_signal_strength}%
\end{equation}
where the infimum is over the same admissible conditioning sets.
\end{assumption}

Stagewise dominance follows from (\ref{eq:GM2_theta}) by virtue of the next result.

\begin{lemma}
[Global regime implies dominance]\label{lem:GM2_implies_A4} Under
Assumptions~A3.3 and \ref{ass:GM2}, the stagewise dominance condition of
Assumption~\ref{ass:A4} is satisfied with
$d_T=(1-\rho)b_T/2$.
\end{lemma}

%\subsubsection{Conditional correlation attenuation}

%\label{sec:G1}
The first primitive form of the global regime condition controls proxies
through a conditional score-moment restriction. For an admissible $S$, let
$\boldsymbol\gamma^*(S)$ maximise
$\mathbb E[l_t(\boldsymbol x_{S,t}'\boldsymbol\gamma,y_t)]$, put
$\eta_{S,t}^*=\boldsymbol x_{S,t}'\boldsymbol\gamma^*(S)$,
$\xi_t(S)=\dot l_t(\eta_{S,t}^*,y_t)$, and
$w_t(S)=-\mathbb E[\ddot l_t(\eta_{S,t}^*,y_t)\mid\boldsymbol x_t]$.
Let $\widetilde x_{j,t}(S)$ be the residual defined by
$\mathbb E[w_t(S)\boldsymbol x_{S,t}\widetilde x_{j,t}(S)]=\boldsymbol0$.
The relevant weighted Gram matrices are assumed nonsingular.

\setcounter{assumption}{0} \renewcommand{\theassumption}{G\arabic{assumption}}\renewcommand{\theHassumption}{G.\arabic{assumption}}

\begin{assumption}
\label{ass:G1} There exists $\rho
\in(0,1)$ such that, for all admissible $S$,
\begin{equation}
\max_{\substack{m\notin S_{0}\\m\notin S}}\Bigl|\mathbb{E}\bigl[\widetilde{x}_{m,t}(S)\,\xi
_{t}(S)\bigr]\Bigr|\ \leq\ \rho\,\min_{j\in S_{0}\setminus S}\Bigl|\mathbb{E}%
\bigl[\widetilde{x}_{j,t}(S)\,\xi_{t}(S)\bigr]\Bigr|. \label{eq:G1}%
\end{equation}

\end{assumption}

The following lemma gives conditions under which (\ref{eq:G1}) is sufficient
for (\ref{eq:GM2_theta}) to hold.

\begin{lemma}
\label{lem:G1_implies_GM2} Suppose
Assumption~\ref{ass:A3} holds and there exist constants
$0<C_L\le C_U<\infty$ such that, uniformly over admissible $S$ and $j\notin S$,
\[
C_L\left|\mathbb E[\widetilde x_{j,t}(S)\xi_t(S)]\right|
\le
\frac{|\theta_j^*(S)|}{\{V_{\theta\theta,j}^*(S)\}^{1/2}}
\le
C_U\left|\mathbb E[\widetilde x_{j,t}(S)\xi_t(S)]\right|.
\]
If Assumption~\ref{ass:G1} holds with a constant $\rho$ satisfying
$\rho C_U/C_L<1$, and if the signal-strength condition
\eqref{eq:GM2_signal_strength} holds, then Assumption~\ref{ass:GM2} holds.
\end{lemma}

%\subsubsection{Conditional projection bound (proxy
%\texorpdfstring{$\mathrm{R}^{2}$}{R^2})}

%\label{sec:G2}

The second primitive form of the global regime condition controls proxies
through a weighted partial-projection restriction. For $R=S_0\setminus S$,
let $\widetilde{\boldsymbol x}_{R,t}(S)$ collect the curvature-weighted
residuals of the remaining signals after projection on $\boldsymbol x_{S,t}$,
and let $\mathcal P^w_{R\mid S}\widetilde x_{m,t}(S)$ denote the fitted value
from the population weighted linear projection of
$\widetilde x_{m,t}(S)$ on $\widetilde{\boldsymbol x}_{R,t}(S)$.

\begin{assumption}
\label{ass:G2} The curvature weights are positive for the designs to which
this condition is applied. There exists $\rho\in(0,1)$ such that, for every
admissible $S$,
\begin{equation}
\sup_{\substack{m\notin S_{0}\\m\notin S}}
\frac{\mathbb E\!\left[w_t(S)
\{\mathcal P^w_{R\mid S}\widetilde x_{m,t}(S)\}^{2}\right]}
{\mathbb E\!\left[w_t(S)\widetilde x_{m,t}^{2}(S)\right]}
\leq\rho,
\qquad R=S_0\setminus S.
\label{eq:G2}%
\end{equation}
Thus the weighted partial $R^2$ of an inactive candidate with the remaining
signals, after conditioning on the current model, is bounded away from one.
\end{assumption}

The bound in \eqref{eq:G2} is useful only when it can be shown to deliver the
moment attenuation in \eqref{eq:G1}; by itself it does not control the sign or
size of the score moments.

\begin{remark}
\label{rem:G2_diagnostic}
Assumption~\ref{ass:G2} is therefore a design diagnostic rather than an
independent sufficient condition for stagewise dominance. When additional
model-specific restrictions imply \eqref{eq:G1} with a constant satisfying
the restriction in Lemma~\ref{lem:G1_implies_GM2}, that lemma, together with
\eqref{eq:GM2_signal_strength}, yields Assumption~\ref{ass:GM2}. No such
implication is claimed from Assumption~\ref{ass:G2} alone.
\end{remark}

%------------------------

\subsection{Logit and probit links}

\label{sec:link_verification_GM2}
%------------------------

Finally, we verify the key analytic condition used in
Lemma~\ref{lem:G1_implies_GM2} for logit and probit links.

\begin{lemma}
\label{lem:logit_probit_mapping} For the logit and probit links, suppose
Assumption~\ref{ass:A3} holds and the relevant indices lie in a compact set.
Suppose further that the profiled population score for the candidate
coefficient has Schur-complement curvature uniformly bounded away from zero
and infinity on the line segment between $0$ and $\theta_j^*(S)$, uniformly
over admissible $S$ and $j$. Then, there exist constants
$0<C_L\le C_U<\infty$ such that, uniformly over admissible $S$ and $j$,
\[
C_L\left|\mathbb{E}[\widetilde{x}_{j,t}(S)\xi_t(S)]\right|
\le
\frac{|\theta_j^*(S)|}{\{V_{\theta\theta,j}^*(S)\}^{1/2}}
\le
C_U\left|\mathbb{E}[\widetilde{x}_{j,t}(S)\xi_t(S)]\right|.
\]
\end{lemma}

In logit and probit models, Assumption~\ref{ass:G1}, or Assumption~\ref{ass:G2}
when it is verified to imply the moment attenuation in \eqref{eq:G1}, implies
Assumption \ref{ass:GM2}, which in turn implies Assumption \ref{ass:A4}
(Lemma~\ref{lem:G1_implies_GM2}; see also
Remark~\ref{rem:G2_diagnostic}). Alternatively,
the local small-index conditions
(\ref{ass:L1}--\ref{ass:L3}) yield dominance directly
(Lemma~\ref{lem:local_implies_A4}). Hence, stagewise dominance is verifiable
via multiple primitive routes, depending on whether the setting is closer to a
local-linear or global proxy-controlled regime.

\paragraph{Comparison with binary LASSO.}
Appendix~\ref{app:binary_lasso} compares the stagewise population ordering used
by BMT with the generalised irrepresentability condition for binary LASSO. The
comparison is not one-sided. Theorem~\ref{thm:app_lasso_bmt_success} gives a
Gaussian proxy design in which BMT recovers the true support although binary
LASSO violates generalised irrepresentability, whereas
Theorem~\ref{thm:app_lasso_lasso_success} gives a design in which binary LASSO
is sign-consistent but an inactive covariate has the largest first-stage BMT
statistic. Thus, neither condition implies the other without further
restrictions.

Theorem~\ref{thm:app_lasso_gi_bmt} gives a qualified positive relation. In a
local binary model, a strict generalised-irrepresentability margin implies
one-signal-at-a-time BMT dominance when conditioning on previously selected
signals does not exhaust that margin in the direction of the remaining active
score. This directional-stability condition allows correlated active
covariates and mixed-sign projection coefficients. The result concerns the
strongest remaining signal, which is sufficient for a procedure that adds one
covariate at each stage; the stronger all-remaining-signals formulation in
Assumption~\ref{ass:A4} additionally requires the active-score balance
condition stated in Appendix~\ref{app:binary_lasso}.

%##############################################################################

\section{Monte Carlo simulations}

\label{sec:Monte_Carlo}

The purpose of this section is to assess the finite-sample properties of the
BMT selection procedure via Monte Carlo simulations. We begin by describing
the experimental design and evaluation metrics, and proceed to discuss the
results of the simulations.\footnote{Simulations were run on a Linux server
(AMD EPYC~7302, 64 threads, 503~GiB RAM, Ubuntu~24.04.3~LTS) using MATLAB.}

\subsection{Experimental design}

\label{sec:Monte_Carlo_design}

The DGP is the binary-response GLM with linear single-index structure given
by
\begin{equation}
\mathbb{E}(y_{t}|\boldsymbol{x}_{t})=\Lambda(\alpha+\boldsymbol{x}_{t}%
^{\prime}\boldsymbol{\beta}_{0}),\quad t=1,2,\ldots,T,\label{eq_dgp1}%
\end{equation}
where $y_{t}\in\{0,1\}$ are conditionally independent Bernoulli variables,
$\boldsymbol{x}_{t}=(x_{1,t},\ldots,x_{p,t})^{\prime}$ is a covariate vector
and $T\in\{150,300\}$. The number of covariates is $p\in\{100,200\}$ when
$T=150$ and $p\in\{200,400\}$ when $T=300$. The intercept $\alpha\in\{0,1,2\}$
controls the number of observed $0$ and $1$ responses in artificial samples
from \eqref{eq_dgp1} --- the larger the value of $\alpha$ is, the more
unbalanced the set of binary data is likely to be. For all pairs $(T,p)$ under
consideration, $\boldsymbol{\beta}_{0}=2\cdot(\boldsymbol{1}_{k}^{\prime
},\boldsymbol{0}_{p-k}^{\prime})^{\prime}$, with $k\in\{1,4\}$, where
$\boldsymbol{1}_{n}$ and $\mathbf{0}_{n}$ denote $n\times1$ vectors of all
ones and zeros, respectively. The DGP is high-dimensional and sparse, with $k$
signal variables.

Signals $x_{1,t},\ldots,x_{k,t}$ and pseudo-signals $x_{k+1,t},\ldots
,x_{2k,t}$ are generated as
\begin{equation}
x_{j,t}=\frac{\varepsilon_{j,t}+\nu_{g}g_{t}+\nu_{f}f_{t}}{(1+\nu_{g}^{2}%
+\nu_{f}^{2})^{1/2}},\quad j=1,\ldots,2k, \label{eq_dgp2}%
\end{equation}
while the remaining noise variables $x_{2k+1,t},\ldots,x_{p,t}$ are obtained
as%
\begin{equation}
x_{j,t}=\frac{\varepsilon_{j-1,t}+\varepsilon_{j,t}+\nu_{f}f_{t}}{(2+\nu
_{f}^{2})^{1/2}},\quad j=2k+1,\ldots,p. \label{eq_dgp3}%
\end{equation}
In \eqref{eq_dgp2}--\eqref{eq_dgp3}, $\varepsilon_{j,t}$ are stochastic
components that represent idiosyncratic variations across the covariates,
while $g_{t}$ and $f_{t}$ are stochastic components that capture local and
global common factors, respectively. Note that the local factor $g_{t}$
affects only the signals and pseudo-signals, whereas the global factor $f_{t}$
influences all $p$ covariates. These components are simulated as stationary
first-order autoregressive processes with zero mean, unit variance and a
common autocovariance structure. Specifically,
\begin{align}
\varepsilon_{j,t}  &  =\rho\varepsilon_{j,t-1}+(1-\rho^{2})^{1/2}%
e_{j,t}^{\varepsilon},\label{eq_dgp41}\\
f_{t}  &  =\rho f_{t-1}+(1-\rho^{2})^{1/2}e_{t}^{f},\label{eq_dgp42}\\
g_{t}  &  =\rho g_{t-1}+(1-\rho^{2})^{1/2}e_{t}^{g}, \label{eq_dgp43}%
\end{align}
where the innovations $\{e_{j,t}^{\varepsilon}\}$, $\{e_{t}^{f}\}$ and
$\{e_{t}^{g}\}$ are collections of independent standard normal random
variables and are independent of each other. We set $\rho=0.6$.

The inclusion of local and common factor components in the covariates
introduces collinearity among the covariates, as is typical in real-world
high-dimensional datasets where correlation among signals and pseudo-signals
arises as a result of shared latent structures. Additionally, such components
reflect the co-movement among covariates that is often observed in big-data
settings due to the influence of underlying systemic forces. Examples of such
forces include macroeconomic trends, technological innovations, regulatory
changes, demographic shifts, geopolitical developments, and sector-wide shocks
such as fluctuations in commodity prices, global supply chain disruptions, or
monetary policy shifts. Modelling $g_{t}$, $f_{t}$, and $\varepsilon_{j,t}$ as
autoregressive processes allows for serial correlation and persistence that
are typical of economic time series.

The role of collinearity is assessed by considering a range of values for the
parameters $\nu_{g}$ and $\nu_{f}$. Defining the variance inflation factor
associated with a signal or pseudo-signal $x_{j,t}$ as $\mathrm{VIF}%
_{j}=1/(1-\mathrm{R}_{j}^{2})$, where $\mathrm{R}_{j}^{2}$ is the coefficient
of determination for the population regression of $x_{j,t}$ on all other $p-1$
candidate covariates, it is easy to show that $\mathrm{VIF}_{j}=1+k(\nu
_{f}^{2}+\nu_{g}^{2})\equiv\mathrm{VIF}$ for all $1\leq j\leq2k$. Hence, we
set
\[
\nu_{f}=\left[  \frac{\omega(\mathrm{VIF}-1)}{k}\right]  ^{1/2},\quad\nu
_{g}=\left[  \frac{(1-\omega)(\mathrm{VIF}-1)}{k}\right]  ^{1/2},
\]
for some $\omega\in(0,1)$. When $\mathrm{VIF}=1$, signals and pseudo-signals
are mutually uncorrelated and the factors $g_{t}$ and $f_{t}$ do not affect
any of the covariates ($\nu_{g}=\nu_{f}=0$). When $\mathrm{VIF}>1$,
$100\omega\%$ of the collinearity among covariates, as quantified by
$\mathrm{VIF}$, is attributable to the global factor $f_{t}$, while the
remaining $100(1-\omega)\%$ arises from the local factor $g_{t}$. We consider
$(\mathrm{VIF},\omega)\in\{(1,0),(2,0.25),(2,0.75),(3,0.25),(3,0.75)\}$ in the experiments.

The BMT procedure is implemented as described in Section~\ref{sec:algo}, using
a correctly specified likelihood and family-wise threshold $c_{T}=\Phi^{-1}\{1-0.05/(2aT)\}$, with $a=1$ in the first stage ($\ell=1$) and $a=2$ in
subsequent stages ($\ell\geq2$). For comparison, we also include in the
experiments the OCMT and LASSO procedures (adapted to the GLM setting). The
OCMT implementation uses the same threshold $c_{T}$ as BMT, while the
regularisation parameter for the LASSO penalty is chosen by 10-fold cross-validation.

\subsection{Performance metrics}

\label{sec:diagnostics}

The properties of covariate-selection procedures are assessed using a variety
of metrics. These include: the true positive rate $TPR=TP/(TP+FN)$; the false
positive rate $FPR=FP/(FP+TN)$; the true discovery rate $TDR=TP/(TP+FP)$; the
false discovery rate $FDR=FP/(FP+TP)$; the $F1$ score:
\(
F1=\frac{2\cdot TP}{(2\cdot TP)+FP+FN};
\)
the Matthews correlation coefficient%

\(MCC=\frac{(TP\cdot TN)-(FP\cdot FN)}{[(TP+FP)(TP+FN)(TN+FP)(TN+FN)]^{1/2}}.\)
The components of these performance measures are defined as
\begin{align}
TP  &  =\sum_{j=1}^{p}\mathbb{I}(\widehat{D}_{j}=1,\,\beta_{0,j}\neq0),\quad
FP=\sum_{j=1}^{p}\mathbb{I}(\widehat{D}_{j}=1,\,\beta_{0,j}=0),\label{eq_tp}\\
TN  &  =\sum_{j=1}^{p}\mathbb{I}(\widehat{D}_{j}=0,\,\beta_{0,j}=0),\quad
FN=\sum_{j=1}^{p}\mathbb{I}(\widehat{D}_{j}=0,\,\beta_{0,j}\neq0),
\label{eq_tn}%
\end{align}
where $\widehat{D}_{j}=1$ if covariate $x_{j,t}$ is selected and
$\widehat{D}_{j}=0$ otherwise, $\beta_{0,j}$ is the true value of the
coefficient on $x_{j,t}$ in \eqref{eq_dgp1}, and $\mathbb{I}(\cdot)$ is the
indicator function.

The $TPR$ metric (sensitivity) captures the share of signals that are
correctly selected, whereas $TDR$ (precision) measures the share of selected
covariates that are signals. Conversely, $FPR$ is the share of pseudo-signal
and noise variables that are erroneously selected, whereas $FDR$ measures the
share of selected covariates that are not signals. Large values of $TPR$ and
$TDR$ are, evidently, indicative of accurate selection, as are small values of
$FPR$ and $FDR$. The $F1$ measure is the harmonic mean of $TPR$ and $TDR$; its
lowest value 0 is obtained when no signal variables are selected, while the
maximum value 1 is associated with perfectly accurate signal selection. The
$MCC$ is, arguably, the most comprehensive accuracy measure, since it
synthesises information about all four possible outcomes associated with the
quantities in \eqref{eq_tp}--\eqref{eq_tn}; its extreme values of $-1$ and $1$
indicate complete failure and success, respectively, in selecting the signals,
while a zero value represents no better than random
selection.\footnote{\citet{Baldi2020} and \citet{Chico2020}, among others,
provide useful discussions on the relative merits of these accuracy metrics.}

In addition to the above set of classification-based performance measures, we
also consider: the number of covariates $(\widehat{k})$ selected (referred to
henceforth as model size); the high-dimensional Bayesian information
criterion
\(
BIC=-2\widehat{L}_{T}+\widehat{k}\ln(pT),
\)
where $\widehat{L}_{T}$ is the maximised log-likelihood associated with the
selected $\widehat{k}$-covariate model \citep{Chen2012}; (iii)~the root mean
squared error of the MLE, computed as%
\[
RMSE=\left(  N^{-1}\sum_{i=1}^{N}\left\Vert \widehat{\boldsymbol{\beta}}%
^{(i)}-\boldsymbol{\beta}_{0}\right\Vert ^{2}\right)  ^{1/2},
\]
where $\widehat{\boldsymbol{\beta}}^{(i)}$ is the $p$-dimensional vector with
$\widehat{k}$ entries equal to the estimated coefficients in the $\widehat{k}%
$-covariate model selected in the $i$th out of $N$ Monte Carlo replications
and all other entries equal to zero. As in the case of the estimation error,
the value of each performance criterion for a given design point is computed
as an average over $N=1000$ replications.

\subsection{Simulation results}

\label{sec:MCresults}

Our discussion here focuses primarily on the $MCC$, the $F1$ score, the model
size, the $BIC$ for the selected models, and the $RMSE$ of the estimated
coefficients. For each of the two values of $k$ (number of signals) under
consideration, Tables~\ref{tab_summary_mcc_alpha0}$-$%
\ref{tab_summary_rmse_alpha0} contain summary statistics for the performance
measures over all 20 design points (combinations of $p$, $T$ and
$(\mathrm{VIF},\omega)$) with $\alpha=0$. The statistics reported include the
median (MED), the interquartile range (IQR), the minimum (MIN) and maximum
(MAX) values, and the average ranking across the three selection methods (RANK).

\begin{table}[h]
\caption{Summary Statistics for Matthews Correlation Coefficient ($MCC$):
$\alpha=0$}%
\label{tab_summary_mcc_alpha0}
\centering%
\begin{tabular}
[c]{ccccccc}\hline\hline
Signals & Method & MED & IQR & MIN & MAX & RANK\\
    \hline
$k=1$ & BMT   & 0.99  & 0.00  & 0.99  & 1.00  & 1.00  \\
      & OCMT  & 0.53  & 0.65  & 0.02  & 0.99  & 2.30  \\
      & LASSO & 0.33  & 0.05  & 0.29  & 0.39  & 2.70  \\
$k=4$ & BMT   & 0.99  & 0.08  & 0.84  & 1.00  & 1.00  \\
      & OCMT  & 0.76  & 0.40  & 0.13  & 1.00  & 2.10  \\
      & LASSO & 0.35  & 0.07  & 0.29  & 0.42  & 2.90  \\
      \hline\hline
\end{tabular}
\begin{minipage}{11cm}\vspace{0.25cm}
\footnotesize{Notes: $MCC$ ranges between $-1$ (perfect misclassification) and $1$ (perfect classification), with 0 indicating random classification.}
\end{minipage}\end{table}

It is clear from the results in Table~\ref{tab_summary_mcc_alpha0} that BMT
outperforms the two competing methods on the $MCC$ metric, ranking first
across all experiments on average (mean rank is $1.0$). It attains a median
$MCC$ of 0.99, with only a narrow interquartile range in the
case of multiple signals. By contrast, the second-ranked OCMT has median $MCC$
of only $0.53$ when $k=1$, rising to $0.76$ when $k=4$. For both values of
$k$, the interquartile range for OCMT is substantial. The $MCC$ values for
LASSO exhibit low dispersion, by comparison, but they never exceed $0.42$.

\begin{table}[h]
\caption{Summary Statistics for $F1$ Score: $\alpha=0$}%
\label{tab_summary_f1_alpha0}%
\centering%
\begin{tabular}
[c]{ccccccc}\hline\hline
Signals & Method & MED & IQR & MIN & MAX & RANK\\
\hline
$k=1$ & BMT   & 0.99  & 0.00  & 0.99  & 1.00  & 1.00  \\
      & OCMT  & 0.45  & 0.70  & 0.01  & 0.99  & 2.30  \\
      & LASSO & 0.22  & 0.05  & 0.17  & 0.29  & 2.70  \\
$k=4$ & BMT   & 0.99  & 0.09  & 0.83  & 1.00  & 1.00  \\
      & OCMT  & 0.74  & 0.46  & 0.06  & 1.00  & 2.10  \\
      & LASSO & 0.25  & 0.09  & 0.18  & 0.35  & 2.90  \\
      \hline\hline
\end{tabular}
\begin{minipage}{11cm}\vspace{0.25cm}
\footnotesize{Notes: $F1$ ranges between 0 (no signal selected) and 1 (all signals selected).}
\end{minipage}\end{table}

The superior performance of BMT is also evident in the summary statistics for
the $F1$ score reported in Table~\ref{tab_summary_f1_alpha0}. Once again, BMT
ranks first in all experiments, with median $F1$ values of $0.99$ for both $k=1$ and $k=4$, and low interquartile range. OCMT does better than LASSO in terms of the median $F1$ value, although this is at
the cost of considerable variation.

\begin{table}[h]
\caption{Summary Statistics for Model Size: $\alpha=0$}%
\label{tab_summary_khat_alpha0}%
\centering%
\begin{tabular}
[c]{ccccccc}\hline\hline
Signals & Method & MED & IQR & MIN & MAX & RANK\\\hline
$k=1$ & BMT   & 1.02  & 0.01  & 1.01  & 1.03  & 1.00  \\
      & OCMT  & 6.30  & 52.93  & 1.02  & 325.76  & 2.40  \\
      & LASSO & 12.87  & 4.00  & 9.36  & 17.48  & 2.60  \\
$k=4$ & BMT   & 3.97  & 0.58  & 2.94  & 4.02  & 1.10  \\
      & OCMT  & 7.06  & 15.30  & 2.94  & 160.72  & 2.05  \\
      & LASSO & 31.11  & 13.06  & 20.45  & 46.13  & 2.85  \\
      \hline\hline
\end{tabular}
\begin{minipage}{11cm}\vspace{0.25cm}
\footnotesize{Notes: Model size is the number of selected covariates.}
\end{minipage}
\end{table}

Another attractive feature of BMT is that selection accuracy is not achieved
by sacrificing model parsimony. This can be seen in
Table~\ref{tab_summary_khat_alpha0}, which contains summary statistics for the
number of selected covariates. When there is a single signal, the median model
size for BMT is $1.02$ with interquartile range of only $0.01$. OCMT, on the
other hand, has median model size $6.30$ and very large dispersion, the
maximum number of covariates selected being $325.76$, on average. LASSO also
selects many additional covariates beyond the single signal, although no more
than $17.48$, on average. Qualitatively similar results are obtained when
multiple signals are present in the DGP, LASSO again selecting very large
models (a median of $31.11$ covariates), as opposed to the median model size
of $7.06$ and $3.97$ of OCMT and BMT, respectively.

\begin{table}[h]
\caption{Summary Statistics for Bayesian Information Criterion ($BIC$):
$\alpha=0$}%
\label{tab_summary_bic_alpha0}%
\centering%
\begin{tabular}
[c]{ccccccc}\hline\hline
Signals & Method & MED & IQR & MIN & MAX & RANK\\\hline
$k=1$ & BMT   & 215.94  & 140.00  & 145.90  & 287.93  & 1.00  \\
      & OCMT  & 295.24  & 441.51  & 146.08  & 3524.60  & 2.40  \\
      & LASSO & 328.57  & 193.81  & 217.01  & 459.13  & 2.60  \\
$k=4$ & BMT   & 146.91  & 75.18  & 101.77  & 216.45  & 1.00  \\
      & OCMT  & 213.18  & 132.92  & 115.70  & 1852.99  & 2.10  \\
      & LASSO & 417.32  & 248.08  & 247.74  & 678.28  & 2.90  \\
      \hline\hline
\end{tabular}
%\begin{minipage}{11cm}\vspace{0.25cm}
%\footnotesize{Notes: $BIC$ for selected model.}
%\end{minipage}
\end{table}

Turning to the complexity-penalised measure of fit,
Table~\ref{tab_summary_bic_alpha0} records summary statistics for the $BIC$.
BMT delivers models with the smallest median $BIC$, regardless of the number
of signals present. OCMT is, once again, second best, while the sets of
covariates selected by LASSO yield the largest median $BIC$ values.

Summary statistics for the $RMSE$ of the estimated coefficients in the models
selected by each of the three procedures are given in
Table~\ref{tab_summary_rmse_alpha0}. Once again, BMT performs best, with a
median $RMSE$ of $0.33$ and $0.93$ when $k=1$ and $k=4$, respectively. The
corresponding values for OCMT and LASSO are more than twice as large. Although
OCMT occasionally attains almost as low an $RMSE$ as that of BMT, there are
cases where estimation errors are excessively large (the maximum $RMSE$ values
being $1553.86$ for $k=1$ and $166.62$ for $k=4$). In spite of its tendency
to select models with a large number of covariates, LASSO does only slight
worse than OCMT when $k=1$ and outperforms OCMT when $k=4$.

\begin{table}[h]
\caption{Summary Statistics for Root Mean Squared Error ($RMSE$): $\alpha=0$}%
\label{tab_summary_rmse_alpha0}%
\centering%
\begin{tabular}
[c]{ccccccc}\hline\hline
Signals & Method & MED & IQR & MIN & MAX & RANK\\\hline
$k=1$ & BMT   & 0.33  & 0.12  & 0.27  & 0.41  & 1.15  \\
      & OCMT  & 1.63  & 74.23  & 0.26  & 1553.86  & 2.40  \\
      & LASSO & 0.77  & 0.18  & 0.61  & 0.99  & 2.45  \\
$k=4$ & BMT   & 0.93  & 0.30  & 0.69  & 1.20  & 1.15  \\
      & OCMT  & 2.79  & 13.80  & 0.68  & 166.62  & 2.45  \\
      & LASSO & 1.90  & 0.30  & 1.66  & 2.24  & 2.40  \\
      \hline\hline
\end{tabular}
%\begin{minipage}{11cm}\vspace{0.25cm}
%\footnotesize{Notes: The Root Mean Squared Error ($RMSE$) of the estimated parameters for the  selected variables.}
%\end{minipage}
\end{table}

Similar conclusions emerge from the experiments with $\alpha=1$ and
$\alpha=2$. To meet the journal's space constraints, the corresponding
summary tables and the design-level results for all nine performance metrics
are omitted from the online appendix; they are available from the authors on
request. Across these experiments, BMT continues to deliver smaller models,
lower parameter-estimation error and higher overall selection accuracy than
OCMT and LASSO. It also has a computational advantage over LASSO, as shown in
Table~\ref{tab_summary_time_alpha0}, and is faster than OCMT in the very sparse
single-signal design.

\begin{table}[h]
\caption{Computation Time (in hours): $\alpha=0$}%
\label{tab_summary_time_alpha0}
\centering
\begin{tabular}
[c]{ccccc}\hline\hline
\multicolumn{1}{c}{} & \multicolumn{2}{c}{$T=150$} &
\multicolumn{2}{c}{$T=300$}\\\cline{2-5}%
Method & $k=1$ & $k=4$ & $k=1$ & $k=4$\\\hline
BMT & 0.02 & 0.07 & 0.09 & 0.28\\
OCMT & 0.03 & 0.04 & 0.28 & 0.27\\
LASSO & 1.62 & 0.65 & 3.73 & 1.67\\\hline\hline
\end{tabular}
%\begin{minipage}{11cm}\vspace{0.25cm}
%\footnotesize{Notes: ...}
%\end{minipage}
\end{table}

\section{Empirical illustration}

\label{sec:Example}
%#############################################################################

In 2012, the U.S. Federal Reserve (Fed) announced the implementation of an
explicit inflation targeting policy, setting the inflation target at 2\%, as
measured by the annual growth rate of the personal consumption expenditures
(PCE) price index. This monetary policy strategy aims to maintain well-anchored inflation
expectations and a moderate interest rate environment, thereby contributing to
a well-functioning economy. To improve the monitoring of price pressures in
the economy, the Federal Reserve Bank of St. Louis has developed a
factor-augmented ordered probit model, which is used to forecast the
probability that the average annual PCE inflation rate over the next 12 months
will exceed 2.5\% \citep{Jackson2015}. These probability forecasts, known as
the Price Pressures Measure (PPM), are updated and released on a monthly
basis.\footnote{The PPM is available at
\url{https://fred.stlouisfed.org/series/STLPPM}.}

%\begin{figure}[h]
%\caption{Average Inflation Rate and High-Inflation Regimes}%
%\label{fig_data}
%\centering
%\includegraphics[scale=1]{fig_data.png}\end{figure}

Our aim here is to illustrate the practical usefulness of the BMT procedure in
the context of forecasting the probability that U.S. inflation will, on
average, be higher than the 2.5\% PPM value over a 12-month horizon. More
specifically, letting $\pi_{t}$ denote the (year-to-year) inflation rate in
month $t$, we are interested in the probability of the event $\{\bar{\pi
}_{t+12}>2.5\}$, where $\bar{\pi}_{t+12}=(1/12)\sum_{i=1}^{12}\pi_{t+i}$ is
the inflation rate averaged over the following 12 months.

%Figure~\ref{fig_data} shows a time plot of the
%average inflation rate $\bar{\pi}_{t+12}$, from January 1962 to June 2025,
%alongside periods of high inflation (shaded areas), that is, periods during
%which $\bar{\pi}_{t+12}>2.5$. Although the segmentation of the observations on
%$\bar{\pi}_{t+12}$ into high- and low-inflation regimes exhibits persistence,
%it is fairly balanced.\footnote{Approximately $52\%$ of observations lie in
%the high-inflation regime, while approximately $48\%$ lie in the low-inflation
%regime.}

Our predictive model for the probabilities associated with the PPM is the
binary logit model
\begin{equation}
\mathbb{E}(y_{t+12}|\boldsymbol{x}_{t})=\Lambda(\alpha+\boldsymbol{x}%
_{t}^{\prime}\boldsymbol{\beta}),\quad t\geq1, \label{eq_ppm}%
\end{equation}
where $y_{t+12}=\mathbb{I}(\bar{\pi}_{t+12}>2.5)$ and $\boldsymbol{x}_{t}$ is
a $p$-dimensional vector of potential predictors. The latter are taken from
the FRED-MD database of macroeconomic and financial variables.\footnote{A
complete description of this database can be found in \cite{Mccracken2021}.
The data can be downloaded from
\url{https://www.stlouisfed.org/research/economists/mccracken/fred-databases}.}
After eliminating variables with missing values and duplicates of consumer
price indices, our dataset comprises monthly observations on the PCE inflation
rate and on $p=102$ predictors, spanning the period January 1961 to June 2025.
The predictors belong to the following general economic groups: output and
income (Group~1); labour market (Group~2); housing (Group~3); consumption,
orders, and inventories (Group~4); money and credit (Group~5); interest and
exchange rates (Group~6); prices (Group~7); stock market (Group~8). When
necessary, the transformations suggested by \cite{Mccracken2021} are used to
ensure approximate stationarity of the data.

To benchmark the performance of BMT, we also implement the OCMT and LASSO 
variable-selection procedures considered in Section \ref{sec:Monte_Carlo}. The predictive ability of the three methods is evaluated using a pseudo-out-of-sample forecasting exercise based on a $85/15$ sample split. In other words, $85\%$ of available observations
(January 1961 to December 2015) are used for parameter estimation and variable
selection (the training sample), while the remaining $15\%$ (January 2016 to
June 2025) are used for forecast evaluation (the evaluation sample). It is
worth noting that the evaluation sample includes two policy-relevant
episodes: the COVID-19 pandemic in 2020 and the global inflation surge in
2021--2023. This makes the out-of-sample forecasting exercise more challenging
and interesting.

As evaluation metrics, we use the $MCC$, the area under the receiver operating
characteristic curve ($AUROC$) \citep{Bamber1975} and the square root of the
quadratic probability score ($SQPS$). The components of the $MCC$ are obtained
in this case as
\begin{align*}
TP  &  =\sum_{t\in\mathcal{T}}\mathbb{I}(\widehat{y}_{t+12}>\tau
,y_{t+12}=1),\quad FP=\sum_{t\in\mathcal{T}}\mathbb{I}(\widehat{y}_{t+12}%
>\tau,y_{t+12}=0),\\
TN  &  =\sum_{t\in\mathcal{T}}\mathbb{I}(\widehat{y}_{t+12}\leq\tau
,y_{t+12}=0),\quad FN=\sum_{t\in\mathcal{T}}\mathbb{I}(\widehat{y}_{t+12}%
\leq\tau,\,y_{t+12}=1),
\end{align*}
where $\widehat{y}_{t+12}\in(0,1)$ is the fitted (predicted) conditional mean
of $y_{t+12}$ from \eqref{eq_ppm}, $\tau\in(0,1)$ is a prespecified threshold
value, and $\mathcal{T}$ is the set of all $t$ for which $\widehat{y}_{t+12}$
is computed. We adopt the conventional choice $\tau=0.5$, although other
threshold values may also be used. The $SQPS$ and $AUROC$ measures are defined
as%
\[
SQPS=\left(  \left\vert \mathcal{T}\right\vert ^{-1}\sum_{t\in\mathcal{T}%
}\left(  \widehat{y}_{t+12}-y_{t+12}\right)  ^{2}\right)  ^{1/2},
\]
and
\[
AUROC=\left(  \left\vert \mathcal{T}_{0}\right\vert \left\vert \mathcal{T}%
_{1}\right\vert \right)  ^{-1}\sum_{t\in\mathcal{T}_{0}}\sum_{h\in
\mathcal{T}_{1}}\mathbb{I}\left(  \widehat{y}_{h+12}-\widehat{y}%
_{t+12}>0\right)  ,
\]
where $\mathcal{T}_{i}=\{t\in\mathcal{T}:$ $y_{t+12}=i\}$, $i\in\{0,1\}$.
Values of $SQPS$ and $AUROC$ near $0$ and $1$, respectively, are an indication
of accurate (in-sample or out-of-sample) identification of high-inflation
regimes, that is, periods associated with $\bar{\pi}_{t+12}>2.5$.

\subsection{In-sample results}

In-sample results for the three covariate-selection methods are summarised in
Table~\ref{tab_fit}.\footnote{The BMT, OCMT and LASSO procedures are
implemented as in the simulations (see Section~\ref{sec:Monte_Carlo_design}). To account for serial dependence, standard errors are obtained from a kernel HAC covariance estimator with the
Bartlett kernel and data-dependent bandwidth determined by the method of 
\cite{Newey1994}.}
The table reports three measures of predictive performance, namely $MCC$, $AUROC$ and $SQPS$, the Bayesian Information Criterion ($BIC$), and the number of selected predictors $(\widehat{k})$.

\begin{table}[h]
\caption{In-Sample Performance}%
\centering%
\begin{tabular}
[c]{c|ccc}\hline\hline
& BMT & OCMT & LASSO\\
    \hline
$MCC$   & 0.803  &    0.939  &    0.994  \\
$AUROC$ & 0.974  &    0.996  &    1.000  \\
$SQPS$  & 0.255  &    0.159  &    0.070  \\
$BIC$   & 318.0  &  470.6    &  575.2  \\
$\widehat{k}$ & 5 & 33 & 49\\
    \hline\hline
\end{tabular}
%\begin{minipage}{15cm}\vspace{0.25cm}
%\end{minipage}
\label{tab_fit}%
\end{table}

All three methods achieve perfect or near-perfect performance with respect to the $AUROC$ metric. This outcome is perhaps unsurprising for OCMT and LASSO, which select 33 and 49 predictors, respectively. By contrast, BMT attains a comparable level of discriminatory accuracy using only 5 predictors.\footnote{These include industrial production measures, capacity utilization, wages and personal consumption expenditures price index.} The resulting gains in parsimony are reflected in the $BIC$, which favours BMT and suggests that it provides a more efficient representation of the predictive information contained in the data

\subsection{Out-of-sample results}

The out-of-sample performance of the variable-selection methods is assessed in
the context of two different forecasting schemes. In the first, the estimated
parameters of the logit model are updated recursively when generating
forecasts (by extending the training sample by one month), while the set of
predictors determined by each method in the initial training sample remains
the same. By contrast, in the second forecasting scheme, both the selected
variables and their estimated coefficients are updated recursively after each
out-of-sample forecast. Although this is computationally more expensive, it
represents a more realistic approach to out-of-sample evaluation.

Results are reported in Table~\ref{tab_oos}. BMT consistently outperforms the
OCMT and LASSO methods on all three evaluation metrics, irrespective of the
forecasting scheme used. The ability of BMT to select predictors that deliver
accurate out-of-sample identification of high-inflation regimes is
particularly impressive in the second setup, with $AUROC$ around $0.92$. 

\begin{table}[h]
\caption{Out-of-Sample Performance}%
\label{tab_oos}%
\centering%
\begin{tabular}
[c]{c|ccc|ccc}\hline\hline
& \multicolumn{3}{c|}{Scheme 1: Fixed Predictors,} & \multicolumn{3}{c}{Scheme
2: Updated Predictors,}\\
& \multicolumn{3}{c|}{Updated Coefficients} & \multicolumn{3}{c}{Updated
Coefficients}\\\hline
& BMT & OCMT & LASSO & BMT & OCMT & LASSO\\\hline
$MCC$   & 0.524  &    0.433  &    0.059  & 0.704  &    0.479  &    0.347  \\
$AUROC$ & 0.863  &    0.640  &    0.281  & 0.917  &    0.695  &    0.616  \\
$SQPS$  & 0.419  &    0.497  &    0.631  & 0.333  &    0.474  &    0.549  \\
 \hline\hline
\end{tabular}
%\begin{minipage}{15cm}\vspace{0.25cm}
%\footnotesize{Notes:...}
%\end{minipage}
\end{table}

\begin{figure}[h]
\caption{Selected Predictors Across Economic Groups Under Forecasting Scheme 1
(Fixed Predictors and Updated Coefficients)}%
\label{fig_heat1}%
\centering  \includegraphics[width=\textwidth]{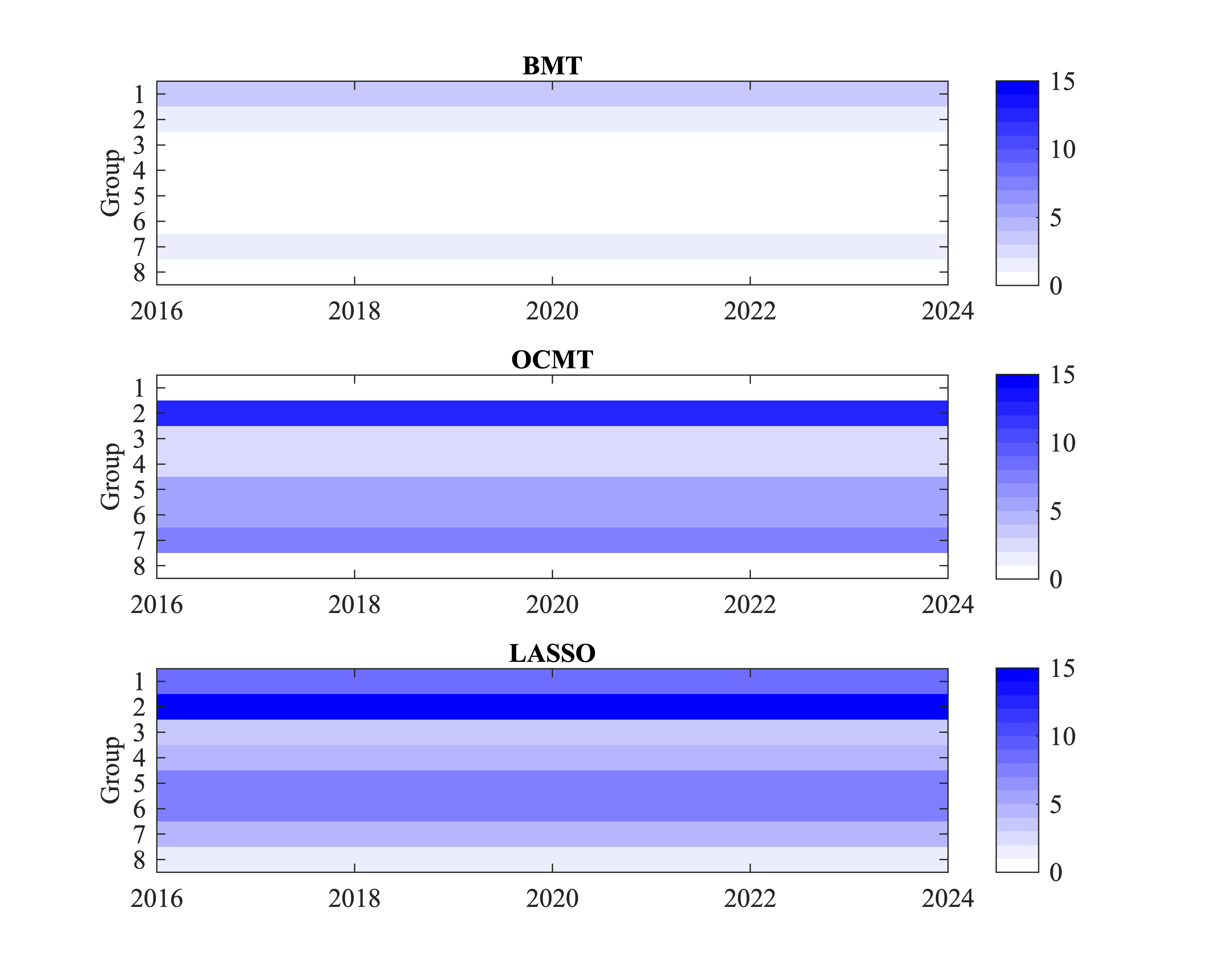}\end{figure}

\begin{figure}[h]
\caption{Selected Predictors Across Economic Groups Under Forecasting Scheme 2
(Updated Predictors and Coefficients)}%
\label{fig_heat2}
\centering
\includegraphics[width=\textwidth]{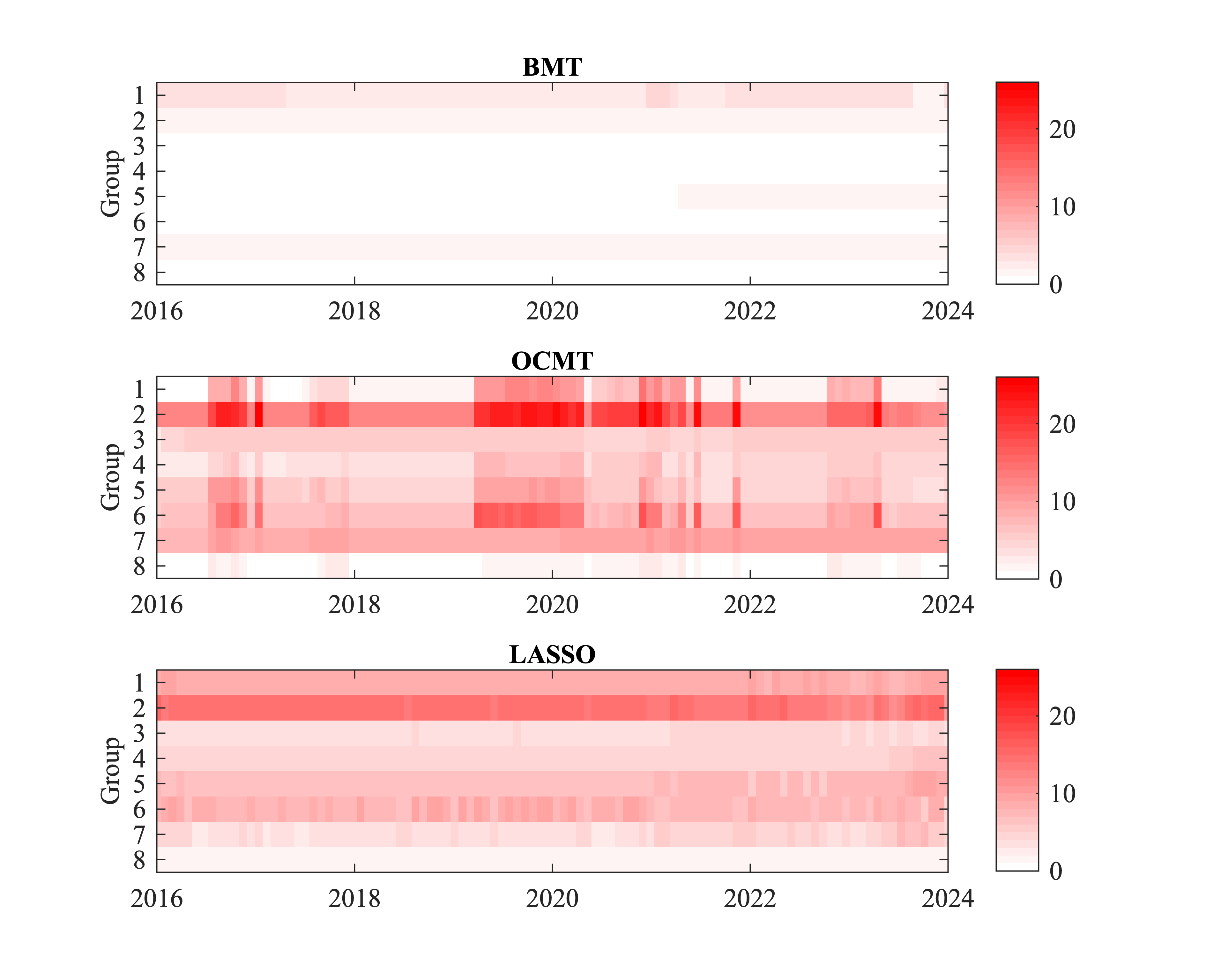}\end{figure}

Figures~\ref{fig_heat1} and \ref{fig_heat2} offer some insight into the
differences in the out-of-sample performance of the three variable-selection
methods under each of the two forecasting schemes used. The heatmaps in these
figures show the number of selected predictors within each of the eight
economic groups in our dataset at each point in the evaluation period; white
indicates that no variables are selected, while dark blue/red indicates that
up to 15/25 variables are selected from a given economic group. As can be seen in
Figure~\ref{fig_heat2}, LASSO selects a stable subset of predictors across the
evaluation period, but this subset tends to be large. OCMT also selects a
large number of predictors across the evaluation period, although not
necessarily the same ones. BMT, on the other hand, selects parsimonious models
with a stable set of predictors, resulting in lower model and parameter
uncertainties compared to OCMT and LASSO. Needless to say, the set of
predictors chosen by each selection procedure remains unaltered over the
evaluation period in Figure~\ref{fig_heat1} due to the design of the
forecasting scheme.

Based on the results obtained by the three methods considered here, most 
selected predictors for price pressures in the U.S. fall into the economic 
groups 1 (output and income), 2 (labour market), 6 (interest/exchange rates)
and 7 (prices). These align with traditional specifications of the Phillips 
curve, which often involve labour market indicators, such as the unemployment 
rate or unit labour costs, as well as economic activity measures, such as the
output gap \citep[e.g.,][]{Stock1999,Faust2013}. They also align with popular
forecasting models for inflation in which factors such as financial variables
\citep{Forni2003} and commodity prices \citep{Chen2014} have been found to
have significant predictive power.\newline

%\subsection*{Appendix}
%Figure~\ref{fig_regimes} shows the average value of the binary dependent variable (i.e. the relative frequency of high-inflation observations) in the training and evaluation samples based on different sample splits. For example, following the $80/20$ rule, there are $60\%$ of observations in high-inflation regime in training sample, whereas only $25\%$ of observations in high-inflation regime in the evaluation sample. In contrast, following the $90/10$ rule, there are $52\%$ of observations in high-inflation regime in the training sample and $49\%$ of observations in high-inflation regime in the evaluation sample.

\section{Conclusion}

\label{sec:conclusion}

This paper has considered the problem of variable selection in binary-response GLMs when the number of possible covariates is large (and potentially exceeds
the number of available observations). Extending the approach of \citet{kapetanios2026}, we propose a nonlinear BMT procedure to select covariates, combining forward stagewise variable addition (boosting) with a family-wise multiple testing stopping rule. A distinctive feature of the procedure is that covariates are introduced sequentially, with at most one covariate added at each stage based on its conditional statistical significance. Following each addition, all remaining candidate covariates are re-evaluated relative to the expanded model before any further selection takes place. This cautious stagewise strategy helps limit the inclusion of pseudo-signals that are correlated with true signals but do not contribute independently to the conditional mean of the response variable.

We have shown that, under the stated conditions, the nonlinear BMT recovers the true
model by selecting all the signals and no other covariates with probability
tending to one. What is more, the post-BMT (quasi-) MLE of the parameters of the
model enjoys an oracle property, in the sense that it has the same asymptotic
properties as the (infeasible) oracle estimator that knows the signal
variables in advance.

An extensive simulation study has demonstrated the good finite-sample
performance of BMT, and its superiority over OCMT and LASSO, in sparse
high-dimensional logit models. Across a variety of designs with varying
collinearity patterns and numbers of covariates and signals, BMT was found to
deliver parsimonious specifications, highly accurate covariate selection and
low parameter estimation errors. Combined with the relatively low
computational cost of BMT (relative to competing procedures such as OCMT and
LASSO) and the attractive asymptotic properties of the procedure, these
findings show that BMT provides a useful and powerful tool for variable selection.

Finally, an empirical application has illustrated the practical usefulness of
the BMT procedure in the context of forecasting the probability that U.S.
inflation will, on average, be higher than 2.5\% over a 12-month horizon. We
have found that BMT delivers a relatively small set of predictors that can
identify accurately out-of-sample periods in which average inflation is high.

\bigskip
\bibliographystyle{plainnat}
\bibliography{literature}

@article{merlevede2011bernstein,
  author  = {Merlev{\`e}de, Florence and Peligrad, Magda and Rio, Emmanuel},
  title   = {A {Bernstein} Type Inequality and Moderate Deviations for
             Weakly Dependent Sequences},
  journal = {Probability Theory and Related Fields},
  year    = {2011},
  volume  = {151},
  number  = {3--4},
  pages   = {435--474},
  doi     = {10.1007/s00440-010-0304-9}
}

@article{yokoyama1980moment,
  author  = {Yokoyama, Ryozo},
  title   = {Moment Bounds for Stationary Mixing Sequences},
  journal = {Zeitschrift f{\"u}r Wahrscheinlichkeitstheorie und
             Verwandte Gebiete},
  year    = {1980},
  volume  = {52},
  number  = {1},
  pages   = {45--57}
}

@article{merlevede2013rosenthal,
  author  = {Merlev{\`e}de, Florence and Peligrad, Magda},
  title   = {{Rosenthal}-Type Inequalities for the Maximum of Partial Sums
             of Stationary Processes and Examples},
  journal = {The Annals of Probability},
  year    = {2013},
  volume  = {41},
  number  = {2},
  pages   = {914--960},
  doi     = {10.1214/11-AOP694}
}

@article{Jackson2015,
	title={A measure of price pressures},
	author={Jackson, Laura E. and Kliesen, Kevin L. and Owyang, Michael},
	journal={Federal Reserve Bank of St. Louis Review},
	volume={97},
	number={1},
	pages={25--52},
	year={2015}}

@article{tibshirani1996,
	title={Regression Shrinkage and Selection Via the Lasso},
	author={Tibshirani, Robert},
	journal={Journal of the Royal Statistical Society B},
	volume={58},
	number={1},
	pages={267--288},
	year={1996}}

@article{fan2001scan,
	title={Variable selection via nonconcave penalized likelihood and its oracle properties},
	author={Fan, Jianqing and Li, Runze},
	journal={Journal of the American statistical Association},
	volume={96},
	number={456},
	pages={1348--1360},
	year={2001}}

@article{park2007lassologit,
  author  = {Park, Mee Young and Hastie, Trevor},
  title   = {${L}_{1}$-regularization path algorithm for generalized linear models},
  journal = {Journal of the Royal Statistical Society B},
  year    = {2007},
  volume  = {69},
  number  = {4},
  pages   = {659--677}
}

@article{zou2005regularization,
	title={Regularization and variable selection via the elastic net},
	author={Zou, Hui and Hastie, Trevor},
	journal={Journal of the Royal Statistical Society B},
	volume={67},
	number={2},
	pages={301--320},
	year={2005}}

@article{chudik2018ocmt,
	title={A one covariate at a time, multiple testing approach to variable selection in high-dimensional linear regression models},
	author={Chudik, Alexander and Kapetanios, George and Pesaran, M. Hashem},
	journal={Econometrica},
	volume={86},
	number={4},
	pages={1479--1512},
	year={2018}}

@article{Mccracken2021,
  title={{FRED-QD}: A Quarterly Database for Macroeconomic Research},
  author={McCracken, Michael W and Ng, Serena},
  journal={Federal Reserve Bank of St. Louis Review},
  volume={103},
  number={1},
  pages={1--44},
  year={2021}}

@article{white1982misspec,
  author  = {White, Halbert},
  title   = {Maximum Likelihood Estimation of Misspecified Models},
  journal = {Econometrica},
  year    = {1982},
  volume  = {50},
  number  = {1},
  pages   = {1--25}
}

@book{mccullagh1989glm,
  author    = {McCullagh, Peter and Nelder, John A.},
  title     = {Generalized Linear Models},
  publisher = {Chapman and Hall},
  year      = {1989},
  edition   = {2nd},
  address   = {London}
  }

@article{fan2008sure,
  author  = {Fan, Jianqing and Lv, Jinchi},
  title   = {Sure independence screening for ultrahigh dimensional feature space},
  journal = {Journal of the Royal Statistical Society B},
  year    = {2008},
  volume  = {70},
  number  = {5},
  pages   = {849--911}
}

@article{fansong2010,
	author  = {Fan, Jianqing and Song, Rui},
	title   = {Sure independence screening for generalized linear models with {NP}-dimensionality},
	journal = {Annals of Statistics},
	year    = {2010},
	volume  = {38},
	number  = {6},
	pages   = {3567--3604}
}

@article{BuhlmannHothorn2007,
	title={Boosting algorithms: regularization, prediction and model fitting},
	author={B\"{u}hlmann, Peter and Hothorn, Torsten},
	journal={Statistical Science},
	volume={22},
	number={7},
	pages={477--505},
	year={2007}}

@incollection{Faust2013,
  	title={Forecasting inflation},
  	author={Faust, Jon and Wright, Jonathan H.},
  	booktitle={Handbook of Economic Forecasting},
  	editor={Elliott, Graham and Timmermann, Allan},
  	volume={2A},
  	publisher={North Holland},
  	address={Amsterdam},
  	pages={2--56},
  	year={2013}}

@article{Chico2020,
  	title={The advantages of the {Matthews} correlation coefficient {(MCC)} over {F1} score
  	and accuracy in binary classification evaluation},
  	author={Chicco, Davide and Jurman, Giuseppe},
  	journal={BMC Genomics},
  	volume={21},
  	number={6},
  	pages={1--13},
  	year={2020}}

@article{Baldi2020,
  	title={Assessing the accuracy of prediction algorithms for classification: an overview},
  	author={Baldi, Pierre and Brunak, S\o{}ren and Chauvin, Yves and Andersen, Claus A. F. and Nielsen, Henrik },
  	journal={Bioinformatics Review},
  	volume={16},
  	number={5},
  	pages={412--424},
  	year={2000}}

@article{Bamber1975,
  		title={The area above the ordinal dominance graph and the area below the 
  		receiver operating characteristic graph},
  		author={Bamber, Donald },
  		journal={Journal of Mathematical Psychology},
  		volume={12},
  		number={4},
  		pages={387--415},
  		year={1975}}

@article{Chen2014,
  	title={Forecasting inflation using commodity price aggregates},
  	author={Chen, Yu-chin and Turnovsky, Stephen J. and Zivot, Eric},
  	journal={Journal of Econometrics},
  	volume={183},
  	number={1},
  	pages={117--134},
  	year={2014}}

@article{Chen2012,
  		title={Extended {BIC} for small-$n$-large-${P}$ sparse {GLM}},
  		author={Chen, Jiahua and Chen, Zehua},
  		journal={Statistica Sinica},
  		volume={22},
  		number={2},
  		pages={555--574},
  		year={2012}}

@article{Forni2003,
  	title={Do financial variables help forecasting inflation and real activity in the euro area?},
  	author={Forni, Mario and Hallin, Marc and Lippi, Marco and Reichlin, Lucrezia},
  	journal={Journal of Monetary Economics},
  	volume={50},
  	number={6},
  	pages={1243--1255},
  	year={2003}}

@unpublished{kapetanios2026,
  author = {Kapetanios, George and Sarafidis, Vasilis and Ventouri, Alexia},
  title  = {Model Selection in High-Dimensional Linear Regression Using Boosting with Multiple Testing},
  note   = {arXiv:2602.19705 [econ.EM]},
  year   = {2026}
}

@article{Newey1994,
  title={{Automatic lag selection in covariance matrix estimation}},
  author={Newey, W. K. and West, K. D.},
  journal={Review of Economic Studies},
  volume={61},
  number={4},
  pages={631--653},
  year={1994}}

@article{Stock1999,
	title={Forecasting inflation},
	author={Stock, James H. and Watson, Mark W.},
	journal={Journal of Monetary Economics},
	volume={44},
	number={2},
	pages={293--335},
	year={1999}}

@BOOK{white94,
  author = {White, Halbert},
  year = 1994,
  title = {Estimation, Inference and Specification Analysis},
  publisher = {Cambridge University Press},
  address = {Cambridge, U.K.}}

\clearpage
\appendix
\setlength{\abovedisplayskip}{4pt plus 1pt minus 1pt}
\setlength{\belowdisplayskip}{4pt plus 1pt minus 1pt}
\setlength{\abovedisplayshortskip}{2pt plus 1pt minus 1pt}
\setlength{\belowdisplayshortskip}{3pt plus 1pt minus 1pt}

\section{Proofs}

\label{app:prelim}
%========================

Definitions and notation used in this Appendix are as in the main text. Limits
are taken as $T\rightarrow\infty$, unless stated otherwise. The constants
$C,C_1,C_2,\ldots$ are finite, positive and may change from line to line.
Throughout this appendix, all maxima over stagewise models are understood over
the admissible class
\(\mathcal A_T = \{(S,j):S^{(0)}\subseteq S\subseteq F_p, |S\setminus S^{(0)}|\le k_{\max},\ j\notin S\}.\)
Recall that
\(
\mathcal E_T
=
\max_{(S,j)\in\mathcal A_T}
\left|W_{T,j}(S)-\mathrm{NC}_j^*(S)\right|.
\)

\subsection{Score implementation and local Wald--score equivalence}

\label{app:wald_score}

The exact-recovery proofs are stated for the Wald statistic $W_{T,j}(S)$ and
only use the uniform approximation encoded in $\mathcal E_T$. They do not
require a global Wald--score equivalence. This distinction matters because
Wald, score and likelihood-ratio statistics are first-order equivalent under
null and local alternatives, but need not be uniformly equivalent over fixed
or diverging alternatives.

If the score statistic $\widetilde W_{T,j}(S)$ in \eqref{eq:score_form} is
used as the implemented stagewise statistic, the results remain valid after
replacing $W_{T,j}(S)$ by $\widetilde W_{T,j}(S)$ in Assumption~\ref{ass:A3}.4
and in the definition of $\mathcal E_T$. Under the usual additional uniform
profiling conditions, the replacement is automatic over population-null and
local alternatives. Specifically, if
\(\mathcal A_T(M)=\{(S,j)\in\mathcal A_T: \mathrm{NC}_{j}^{*}(S)\le M a_{T,p}\},\)
where $M<\infty$ is fixed and $a_{T,p}=\sqrt{\ell_T}$ under A2-ET or
$a_{T,p}=a_{T,p}^{\mathrm{PM}}$ under A2-PM, and if the profiled score admits
a uniform Taylor expansion around the restricted estimator with uniformly
consistent Schur-complement curvature and score denominator, then
\[
\max_{(S,j)\in\mathcal A_T(M)}
|W_{T,j}(S)-\widetilde W_{T,j}(S)|=o_p(a_{T,p}).
\]
For fixed or diverging alternatives, the score implementation should instead
be treated through its own population noncentrality and its own analogue of
Assumption~\ref{ass:A3}.4.

%========================

\subsection{Proof of Theorem~\ref{thm:approx_stats}}

\label{app:thm1}
%========================

\begin{proof}
For every realised stage $\ell\le k_{\max}$ and every candidate
$j\notin S^{(\ell)}$, the pair $(S^{(\ell)},j)$ belongs to $\mathcal A_T$.
Hence,
\[
\max_{\ell\le k_{\max}}
\max_{j\notin S^{(\ell)}}
\left|W_{T,j}(S^{(\ell)})-
\mathrm{NC}_j^*(S^{(\ell)})\right|
\le \mathcal E_T.
\]
Assumption~\ref{ass:A3}.4 gives $\mathcal E_T=O_p(\sqrt{\ell_T})$ under A2-ET and
$\mathcal E_T=O_p(a_{T,p}^{\mathrm{PM}})$ under A2-PM, proving parts (a) and
(b).

For part (c), let
\(
E_T^{(0)}=\{\mathcal E_T\le (1-\bar{\eta}_{0})c_T\}.
\)
By Assumption~\ref{ass:A5}, $\Pr(E_T^{(0)})\to1$. On this event, if
$\mathrm{NC}_{j}^{*}(S^{(\ell)})=0$, then
\[
W_{T,j}(S^{(\ell)})
\le
\left|W_{T,j}(S^{(\ell)})-
\mathrm{NC}_{j}^{*}(S^{(\ell)})\right|
\le \mathcal E_T
\le (1-\bar{\eta}_{0})c_T<c_T.
\]
Thus, no population-null candidate can cross the threshold at any admissible
stage on $E_T^{(0)}$, which proves \eqref{eq:thm1_null}.
\end{proof}

%========================

\subsection{Proof of Theorem~\ref{thm:approx_stop}}

\label{app:thm2}
%========================

\begin{proof}
Part (i) follows immediately from Theorem~\ref{thm:approx_stats}(c), because
BMT can select a covariate only if its statistic is at least $c_T$.

For part (ii), use the event $E_T^{(0)}$ defined in the proof of
Theorem~\ref{thm:approx_stats}. On $E_T^{(0)}$,
\[
\max_{j\notin S^{(\ell^\dagger)}}
W_{T,j}(S^{(\ell^\dagger)})
\le
M_{\ell^\dagger}^*+\mathcal E_T.
\]
By assumption, $M_{\ell^\dagger}^*=o_p(c_T)$, so with probability tending to one, $M_{\ell^\dagger}^*\le \bar{\eta}_{0} c_T/2$. Combining this with
$\mathcal E_T\le (1-\bar{\eta}_{0})c_T$ yields
\[
\max_{j\notin S^{(\ell^\dagger)}}
W_{T,j}(S^{(\ell^\dagger)})
\le (1-\bar{\eta}_{0}/2)c_T<c_T,
\]
with probability tending to one. The stopping rule, therefore, stops at or
before $\ell^\dagger$. The final displayed statement is the selected-stage
form of part (i).
\end{proof}

%========================

\subsection{Proof of Theorem~\ref{thm:k1_recovery}}

\label{app:thm3}
%========================

\begin{proof}
Let $S_0=\{j_0\}$ and define the event
\(E_{1T}=\{\mathcal E_T\le d_{1T}/4\}.\)
By Assumption~\ref{ass:A41}, $\Pr(E_{1T})\to1$. On $E_{1T}$, for every
admissible $S$ not containing $j_0$,
\(W_{T,j_0}(S) \ge \mathrm{NC}_{j_0}^*(S)-d_{1T}/4,\)
whereas for every $m\ne j_0$, $m\notin S$,
\(W_{T,m}(S) \le \mathrm{NC}_{m}^*(S)+d_{1T}/4.\)
Using \eqref{eq:dom_k1},
\[
W_{T,j_0}(S)
\ge
\max_{\substack{m\ne j_0\\m\notin S}}W_{T,m}(S)+d_{1T}/2.
\]
Thus, the signal is the unique maximiser of the stagewise statistic at any
admissible stage before it has been selected. In particular, at $S=S^{(0)}$,
the first-stage maximiser is $j_0$ on $E_{1T}$.

The signal also crosses the testing threshold. Indeed, by
\eqref{eq:dom_k1_threshold}, on $E_{1T}$,
\(W_{T,j_0}(S^{(0)}) \ge c_T+d_{1T}-d_{1T}/4>c_T.\)
Therefore, the first-stage passing set is nonempty and BMT selects $j_0$ with
probability tending to one. This proves parts (a) and (b).

For part (c), after the event $\{j^{(0)}=j_0\}$ occurs, the next conditioning
set is $S^{(0)}\cup\{j_0\}$. Condition \eqref{eq:k1_exhaustion} is exactly the
approximating-model exhaustion condition of Theorem~\ref{thm:approx_stop} at
that stage. Hence, no remaining candidate crosses the threshold with
probability tending to one, and the algorithm stops with
$\widehat S=S^{(0)}\cup\{j_0\}$.
\end{proof}

%========================

\subsection{Proof of Theorem~\ref{thm:fixedk_recovery}}

\label{app:thm4}
%========================

\begin{proof}
Let
\(E_T^{(d)}=\{\mathcal E_T\le d_T/4\}.\)
By Assumption~\ref{ass:A4}, $\Pr(E_T^{(d)})\to1$. We prove by induction that,
on $E_T^{(d)}$, BMT selects a new signal at every stage
$\ell=0,\ldots,k-1$.

The induction hypothesis at stage $\ell<k$ is
\(S^{(\ell)}\subseteq S^{(0)}\cup S_0, \qquad |S^{(\ell)}\cap S_0|=\ell.\)
It is true at $\ell=0$. Suppose it holds at some $\ell<k$. Then
$S^{(\ell)}$ is admissible for Assumption~\ref{ass:A4}. For any remaining
signal $j\in S_0\setminus S^{(\ell)}$ and any remaining non-signal
$m\notin S_0$, $m\notin S^{(\ell)}$, Assumption~\ref{ass:A4} gives
\(\mathrm{NC}_j^*(S^{(\ell)}) \ge \mathrm{NC}_m^*(S^{(\ell)})+d_T.\)
On $E_T^{(d)}$,
\[
W_{T,j}(S^{(\ell)})
\ge
\mathrm{NC}_j^*(S^{(\ell)})-d_T/4,
\qquad
W_{T,m}(S^{(\ell)})
\le
\mathrm{NC}_m^*(S^{(\ell)})+d_T/4.
\]
Consequently,
\[
\min_{j\in S_0\setminus S^{(\ell)}}
W_{T,j}(S^{(\ell)})
\ge
\max_{\substack{m\notin S_0\\m\notin S^{(\ell)}}}
W_{T,m}(S^{(\ell)})+d_T/2.
\]
A non-signal, therefore, cannot maximise the stagewise statistic.

It remains to check that at least one remaining signal passes the threshold.
By \eqref{eq:stg_threshold}, for every remaining signal $j$,
\(W_{T,j}(S^{(\ell)}) \ge c_T+d_T-d_T/4>c_T\)
on $E_T^{(d)}$. Thus, the passing set is nonempty and the selected covariate is
an element of $S_0\setminus S^{(\ell)}$. This closes the induction and proves
\eqref{eq:thm4_true_each_stage} and \eqref{eq:thm4_after_k}.

After $k$ successful selections, $S^{(k)}=S^{(0)}\cup S_0$ on an event whose
probability tends to 1. The exhaustion condition \eqref{eq:thm4_null_after_true}
gives
\(
\max_{j\notin S^{(k)}}\mathrm{NC}_j^*(S^{(k)})=o(c_T),
\)
so Theorem~\ref{thm:approx_stop} implies that no remaining candidate crosses
$c_T$ with probability tending to one. Hence BMT stops with
$\widehat S=S^{(0)}\cup S_0$.
\end{proof}

%========================

\subsection{Proof of Theorem~\ref{thm:oracle}}
\label{app:thm5}
\begin{proof}
By Assumption~\ref{ass:A6},
\[
\sqrt T(\widehat{\boldsymbol\beta}_{\rm or}-\boldsymbol\beta_{0,S^\dagger})
=
\boldsymbol J_0^{-1}T^{-1/2}\sum_{t=1}^T\boldsymbol\psi_{0,t}+o_p(1)
\Rightarrow
N(\boldsymbol{0},\boldsymbol J_0^{-1}\boldsymbol\Omega_0\boldsymbol J_0^{-1}).
\]
On $\{\widehat S=S^\dagger\}$,
$\widehat{\boldsymbol\beta}_{{\rm post},S^\dagger}
=\widehat{\boldsymbol\beta}_{\rm or}$. Hence, for every $\varepsilon>0$,
\[
\Pr\!\left(
\sqrt T\left\|
\widehat{\boldsymbol\beta}_{{\rm post},S^\dagger}
-\widehat{\boldsymbol\beta}_{\rm or}
\right\|>\varepsilon
\right)
\leq
\Pr(\widehat S\neq S^\dagger)\longrightarrow0.
\]
Thus
$\sqrt T(\widehat{\boldsymbol\beta}_{{\rm post},S^\dagger}
-\widehat{\boldsymbol\beta}_{\rm or})=o_p(1)$, and the asserted limiting
distribution follows by Slutsky's theorem. Consistency follows from the same
argument without the $\sqrt T$ scaling.
\end{proof}

%========================
\subsection{Proof of Lemma \ref{lem:logit_probit}}

\begin{proof}
For logit and probit links, the score and Hessian of the binary-response
quasi-log-likelihood are smooth functions of the single index. On the stated
compact index sets, the derivative envelopes required in
Assumption~\ref{ass:A2} are finite. Conditions A3.2--A3.3 are imposed in the
statement, and Theorem~\ref{thm:primitive_wald} gives A3.4 under its
additional growth and standard-error conditions. Hence
Assumption~\ref{ass:A3} holds.

If the postulated conditional mean is correctly specified and
$S\supseteq S^{\dagger}=S^{(0)}\cup S_0$, then
\(\mathbb E(y_t\mid\boldsymbol x_t)=
G(\boldsymbol x_{S^{\dagger},t}'\boldsymbol\beta_{0,S^{\dagger}}).\)
The augmented model then contains the true index. For every candidate
$m\notin S$, the population score in the candidate direction is zero at
$\theta_m=0$. Pseudo-true uniqueness gives $\theta_m^*(S)=0$, and
\eqref{eq:nc_def_prim} gives $\mathrm{NC}_m^*(S)=0$.
\end{proof}

\subsection{Proof of Lemma \ref{lem:local_implies_A4}}

\label{app:lem3}
%========================

\begin{proof}
Write \(M_j(S)=|\mathbb E[r_{j,t}(S)\boldsymbol x_t'\boldsymbol\beta_0]|\).
By \eqref{eq:local_uniform_expansion}, uniformly over admissible $S$,
\[
\min_{j\in S_0\setminus S}
\frac{|\theta_j^*(S)|}{\{V_{\theta\theta,j}^*(S)\}^{1/2}}
\geq
C_L\min_{j\in S_0\setminus S}M_j(S)-r_T,
\]
whereas
\[
\max_{\substack{m\notin S_0\\m\notin S}}
\frac{|\theta_m^*(S)|}{\{V_{\theta\theta,m}^*(S)\}^{1/2}}
\leq
C_U\max_{\substack{m\notin S_0\\m\notin S}}M_m(S)+r_T.
\]
Assumption~\ref{ass:L3} therefore gives a variance-normalised
pseudo-true-coefficient gap of at least
$\underline\kappa_T-2r_T\geq\underline\kappa_T/2$ for all sufficiently large
$T$. It also gives
\[
\min_{j\in S_0\setminus S}
\frac{|\theta_j^*(S)|}{\{V_{\theta\theta,j}^*(S)\}^{1/2}}
\geq \underline\kappa_T-r_T
\geq 3\underline\kappa_T/4.
\]
After multiplication by $\sqrt T$, the population noncentrality gap is at
least $\sqrt T\,\underline\kappa_T/2$, while every remaining signal has
noncentrality at least $3\sqrt T\,\underline\kappa_T/4$. Put
$d_T=\sqrt T\,\underline\kappa_T/4$. The rate condition in the lemma gives
$\mathcal E_T/d_T=o_p(1)$. The same rate condition implies the deterministic
relation $c_T/(\sqrt T\,\underline\kappa_T)\to0$, so, for all sufficiently
large $T$, $c_T+d_T\leq3\sqrt T\,\underline\kappa_T/4$. Both the dominance and
threshold-crossing parts of Assumption~\ref{ass:A4} follow.
\end{proof}

%========================

\subsection{Proof of Lemma \ref{lem:GM2_implies_A4}}

\label{app:lem4}
%========================

\begin{proof}
By Assumption~\ref{ass:GM2}, for every admissible $S$,
\[
\max_{\substack{m\notin S_0\\m\notin S}}\mathrm{NC}_m^*(S)
\le
\rho
\min_{j\in S_0\setminus S}\mathrm{NC}_j^*(S).
\]
Let
\(
b_T=
\inf_S\min_{j\in S_0\setminus S}\mathrm{NC}_j^*(S).
\)
Then
\[
\min_{j\in S_0\setminus S}\mathrm{NC}_j^*(S)
-
\max_{\substack{m\notin S_0\\m\notin S}}\mathrm{NC}_m^*(S)
\ge (1-\rho)b_T.
\]
With $d_T=(1-\rho)b_T/2$, the dominance gap condition holds. The signal-strength condition \eqref{eq:GM2_signal_strength} implies
$\mathcal E_T/d_T=o_p(1)$. Since $c_T$ and $b_T$ are deterministic and
$\mathcal E_T\geq0$, it also implies $c_T/b_T\to0$. Hence, for all
sufficiently large $T$,
\(
\min_{j\in S_0\setminus S}\mathrm{NC}_j^*(S)
\ge b_T\ge c_T+d_T.
\)
Thus Assumption~\ref{ass:A4} follows.
\end{proof}

%#############################################################################

\subsection{Proof of Lemma \ref{lem:G1_implies_GM2}}

\label{app:lem5}
%========================

\begin{proof}
Let
\(
M_j(S)=\mathbb E[\widetilde x_{j,t}(S)\xi_t(S)].
\)
The two-sided moment-to-coefficient bound in Lemma~\ref{lem:G1_implies_GM2}
and Assumption~\ref{ass:G1} imply
\[
\max_{\substack{m\notin S_0\\m\notin S}}
\frac{|\theta_m^*(S)|}{\{V_{\theta\theta,m}^*(S)\}^{1/2}}
\le
C_U\max_{\substack{m\notin S_0\\m\notin S}}|M_m(S)|
\le
C_U\rho\min_{j\in S_0\setminus S}|M_j(S)|.
\]
Using the lower bound for the signal coefficients,
\[
\min_{j\in S_0\setminus S}|M_j(S)|
\le
C_L^{-1}
\min_{j\in S_0\setminus S}
\frac{|\theta_j^*(S)|}{\{V_{\theta\theta,j}^*(S)\}^{1/2}}.
\]
Thus, \eqref{eq:GM2_theta} holds with constant $\rho C_U/C_L<1$. The
signal-strength part of Assumption~\ref{ass:GM2} is imposed directly, so
Assumption~\ref{ass:GM2} follows.
\end{proof}

%========================

\subsection{Proof of Lemma \ref{lem:logit_probit_mapping}}

\label{app:lem7}

\begin{proof}
Fix an admissible pair $(S,j)$. Let
$\eta_{S,t}^*=\boldsymbol x_{S,t}'\boldsymbol{\gamma}^*(S)$ denote the restricted
pseudo-true index and let $\xi_t(S)=\dot l_t(\eta_{S,t}^*,y_t)$ be the
restricted population score residual. Define the population nuisance profile
$\boldsymbol{\gamma}_j(\theta;S)$ as the solution to
\(
\mathbb E\!\bigl[
\boldsymbol x_{S,t}\dot l_t(\boldsymbol x_{S,t}'\boldsymbol{\gamma}+\theta x_{j,t},y_t)
\bigr]=0,
\)
with $\boldsymbol{\gamma}_j(0;S)=\boldsymbol{\gamma}^*(S)$. The profiled candidate score is
\(
\Psi_j(\theta;S)
=
\mathbb E\!\bigl[
 x_{j,t}\dot l_t(\boldsymbol x_{S,t}'\boldsymbol{\gamma}_j(\theta;S)+\theta x_{j,t},y_t)
\bigr].
\)
Equivalently, after differentiating the nuisance first-order condition, the
profile derivative is the Schur complement
\[
-\dot\Psi_j(\theta;S)
=
J_{\theta\theta,j}(\theta;S)
-J_{\theta\gamma,j}(\theta;S)J_{\gamma\gamma,j}(\theta;S)^{-1}
J_{\gamma\theta,j}(\theta;S).
\]
By the curvature condition in Lemma~\ref{lem:logit_probit_mapping}, this
profile curvature is bounded and bounded away from zero uniformly on the line
segment between $0$ and $\theta_j^*(S)$:
\(0<C_1\le -\dot\Psi_j(\theta;S) \le C_2<\infty .\)
The equality $\Psi_j(0;S)=\mathbb E[\widetilde x_{j,t}(S)\xi_t(S)]$ follows
from the weighted-residual representation of the Schur complement; the
residualisation removes the nuisance-score component. The pseudo-true
candidate coefficient satisfies
\(\Psi_j(\theta_j^*(S);S)=0.\)
The mean-value theorem, therefore, gives, for some intermediate
$\bar\theta_j(S)$ between $0$ and $\theta_j^*(S)$,
\(|\Psi_j(0;S)| = |\dot\Psi_j(\bar\theta_j(S);S)|\,|\theta_j^*(S)|.\)
The displayed curvature bounds imply the uniform two-sided inequality
\(C_2^{-1}|\Psi_j(0;S)| \le |\theta_j^*(S)| \le C_1^{-1}|\Psi_j(0;S)|.\)
Since $V_{\theta\theta,j}^*(S)$ is uniformly bounded above and away from zero
by Assumption~\ref{ass:A3}, the same two-sided inequality holds, up to finite
positive constants, after division by
$\{V_{\theta\theta,j}^*(S)\}^{1/2}$. This proves the displayed two-sided bound.
\end{proof}

%========================

\clearpage

%#############################################################################

\section{Extension of Nonlinear BMT to General Exponential-Family GLM}
\label{app:addGLM}

\setcounter{equation}{0} \renewcommand{\theequation}{B.\arabic{equation}}\renewcommand{\theHequation}{B.\arabic{equation}}
\setcounter{theorem}{0} \renewcommand{\thetheorem}{B.\arabic{theorem}}\renewcommand{\theHtheorem}{B.\arabic{theorem}}
\setcounter{lemma}{0} \renewcommand{\thelemma}{B.\arabic{lemma}}\renewcommand{\theHlemma}{B.\arabic{lemma}}
\setcounter{remark}{0} \renewcommand{\theremark}{B.\arabic{remark}}\renewcommand{\theHremark}{B.\arabic{remark}}
\setcounter{assumption}{0} \renewcommand{\theassumption}{B.\arabic{assumption}}\renewcommand{\theHassumption}{B.\arabic{assumption}}

In this Appendix, we provide results on the properties of nonlinear BMT in the
context of high-dimensional GLMs\ in which the response variable takes values in an arbitrary
subset of the real line and its conditional distribution given a vector of
covariates is assumed to belong to a regular one-parameter exponential family.
The notation used here is largely self-contained and, unless indicated
otherwise, should not be interpreted as referring to the same quantities as
elsewhere in the main text. Throughout, an overdot indicates differentiation,
with multiple overdots indicating the order of the derivative. All limits are
taken as $T\rightarrow\infty$, unless stated otherwise.

\subsection{Setup}

\label{app:glm_setup}
%----------------------------------------------------

For $t=1,2,\ldots T$, let $y_{t}\in\mathcal{Y}\subseteq\mathbb{R}$ be a scalar
response and $\boldsymbol{x}_{t}=(x_{1,t},\dots,x_{p,t})^{\prime}\in
\mathbb{R}^{p}$ be a covariate vector whose dimension $p=p_{T}$ may increase
with $T$. To avoid notational clutter, intercepts and deterministic controls
are not explicitly included in $\boldsymbol{x}_{t}$. Such controls are to be
considered as always-in in all stagewise submodels and, if present, all
selected sets below are to be understood as referring to only non-always-in covariates.

Conditional on $\boldsymbol{x}_{t}$, $y_{t}$ is assumed to have a (quasi-) density of
the form
\begin{equation}
f(y_{t}|\boldsymbol{x}_{t};\boldsymbol{\beta},\phi)=\exp\left\{  \frac
{y_{t}\zeta_{t}-B(\zeta_{t})}{A(\phi)}+K(y_{t},\phi)\right\}  ,
\label{eq:app_ef_density}%
\end{equation}
where $\zeta_{t}\in\mathcal{Z}\subseteq\mathbb{R}$ is the scalar natural
parameter, $B:\mathcal{Z}\rightarrow\mathbb{R}$ is the cumulant function,
$\phi\in\bar{\Phi}\subseteq\mathbb{R}_{+}$ is a dispersion parameter (possibly
known), $A:\bar{\Phi}\rightarrow(0,\infty)$ is the dispersion or scale
function, and $K:\mathcal{Y}\times\bar{\Phi}\rightarrow\mathbb{R}$ is the
carrier term. Hence, the conditional mean of $y_{t}$ is $\mu_{t}%
=\mathbb{E}(y_{t}|\boldsymbol{x}_{t})=\dot{B}(\zeta_{t})$. The mean space is
denoted by $\mathcal{M}=\dot{B}(\mathcal{Z})$. It is assumed that $B$ is
strictly convex on $\mathcal{Z}$, so that $\dot{B}$ is one-to-one and $\dot
{B}^{-1}$ is well-defined on $\mathcal{M}$.

The conditional mean $\mu_{t}$ is related to the index $\eta_{t}%
=\boldsymbol{x}_{t}^{\prime}\boldsymbol{\beta}$, where $\boldsymbol{\beta}%
\in\mathbb{R}^{p}$ is an unknown parameter, by the GLM link, that is, an
one-to-one function $g:\mathcal{M}\rightarrow\mathcal{H}$, $\mathcal{H}%
\subseteq\mathbb{R}$, such that $g(\mu_{t})=\eta_{t}$. Let $G:\mathcal{H}%
\rightarrow\mathcal{M}$ denote the inverse link function $g^{-1}$. Since the
conditional mean at index value $\eta\in\mathcal{H}$ is $\mu(\eta)=G(\eta)$,
the natural parameter corresponding to $\eta\in\mathcal{H}$ is $\zeta
(\eta)=\dot{B}^{-1}\{G(\eta)\}\in\mathcal{Z}$. Consequently, the covariates
enter (\ref{eq:app_ef_density}) through $\zeta_{t}=\dot{B}^{-1}%
\{G(\boldsymbol{x}_{t}^{\prime}\boldsymbol{\beta})\}$. For canonical-link
GLMs, $\zeta(\eta)=\eta$; for non-canonical links, $\zeta(\eta)\neq\eta$ in general.

It is assumed throughout that: (i)~there exists a sparse vector
$\boldsymbol{\beta}_{0}=(\beta_{0,1},\ldots,\beta_{0,p})^{\prime}\in
\mathbb{R}^{p}$ such that $G(\boldsymbol{x}_{t}^{\prime}\boldsymbol{\beta}%
_{0})$ equals the true conditional expectation of $y_{t}$ given
$\boldsymbol{x}_{t}$; (ii)~the true active set $S_{0}=\{j\in\{1,\ldots
,p\}:\beta_{0,j}\neq0\}$ has fixed and finite cardinality $k$.

Given (\ref{eq:app_ef_density}), the (quasi-) log-likelihood for observation
$t$ and index value $\eta$ is
\begin{equation}
l_{t}(\eta)=l(\eta,y_{t};\phi)=\frac{y_{t}\zeta(\eta)-B\{\zeta(\eta)\}}%
{A(\phi)}+K(y_{t},\phi). \label{eq:app_index_loglik}%
\end{equation}
Hence, since $\dot{B}\{\zeta(\eta)\}=G(\eta)$ and $\dot{G}(\eta)=\ddot
{B}\{\zeta(\eta)\}\dot{\zeta}(\eta)$, we have
\begin{equation}
\dot{l}_{t}(\eta)=\frac{\{y_{t}-G(\eta)\}\dot{\zeta}(\eta)}{A(\phi)},
\label{eq:app_index_score}%
\end{equation}%
\begin{equation}
\ddot{l}_{t}(\eta)=\frac{\{y_{t}-G(\eta)\}\ddot{\zeta}(\eta)-\dot{G}(\eta
)\dot{\zeta}(\eta)}{A(\phi)}=\frac{\{y_{t}-G(\eta)\}\ddot{\zeta}(\eta
)-\ddot{B}\{\zeta(\eta)\}\{\dot{\zeta}(\eta)\}^{2}}{A(\phi)},
\label{eq:app_index_hessian}%
\end{equation}
where all derivatives of $l_{t}(\eta)$ are with respect to the index $\eta$.

Next, we define the stagewise submodels of interest. To this end, let
$S\subseteq\{1,\ldots,p\}$ be a conditioning set with $|S|=s$ and write
$\boldsymbol{x}_{S,t}\in\mathbb{R}^{s}$ for the subvector of $\boldsymbol{x}%
_{t}$ indexed by $S$. For a candidate $j\notin S$, define the augmented
stagewise index
\begin{equation}
\eta_{j,t}(\boldsymbol{\vartheta}_{j};S)=\boldsymbol{x}_{S,t}^{\prime
}\boldsymbol{\gamma}_{j}+\tau_{j}x_{j,t},\qquad\boldsymbol{\vartheta}%
_{j}=(\boldsymbol{\gamma}_{j}^{\prime},\tau_{j})^{\prime}\in\mathbb{R}^{s+1},
\label{eq:app_aug_index}%
\end{equation}
so that the scalar $\tau_{j}$ is the coefficient on the candidate covariate
$x_{j,t}$. The pseudo-true parameter associated with the augmented submodel is
defined as
\begin{equation}
\boldsymbol{\vartheta}_{j}^{\ast}(S)\in\arg\max_{\boldsymbol{\vartheta}_{j}%
\in\Theta_{j,S}}\mathbb{E}\left[  l_{t}\{\eta_{j,t}(\mathbf{\vartheta}%
_{j};S)\}\right]  , \label{eq:app_pseudotrue}%
\end{equation}
where $\Theta_{j,S}\subseteq\mathbb{R}^{s+1}$ is the parameter space for
$\boldsymbol{\vartheta}_{j}$ and the expectation is taken under the true
conditional distribution of $y_{t}$ given $\boldsymbol{x}_{t}$. We write
\(\boldsymbol{\vartheta}_{j}^{\ast}(S)=(\boldsymbol{\gamma}_{j}^{\ast }(S)^{\prime},\tau_{j}^{\ast}(S))^{\prime},\)
with $\boldsymbol{\gamma}_{j}^{\ast}(S)\in\mathbb{R}^{s}$ and $\tau_{j}^{\ast
}(S)\in\mathbb{R}$. It is assumed throughout that, for every $S\in
\mathcal{S}_{k_{\max}}$ (to be defined) and every $j\notin S$, the population
criterion $\vartheta_{j}\mapsto\mathbb{E}\left[  l_{t}\{\eta_{j,t}%
(\vartheta_{j};S)\}\right]  $ has a unique maximiser $\boldsymbol{\vartheta
}_{j}^{\ast}(S)$. Thus, $\tau_{j}^{\ast}(S)$ and the quantities
$\boldsymbol{J}_{j}^{\ast}(S)$, $\boldsymbol{V}_{j}^{\ast}(S)$ and
$\Lambda_{T,j}^{\ast}(S)$ appearing in (\ref{eq:app_population_hessian}%
)--(\ref{eq:app_lambda}) below are uniquely defined. If uniqueness is not
imposed, all results below should instead be read under a fixed deterministic
tie-breaking rule, together with the additional requirement that all tied
pseudo-true maximisers yield the same value of $\tau_{j}^{\ast}(S)$ and the
same relevant curvature and variance quantities.

The estimator
\[
\widehat{\boldsymbol{\vartheta}}_{j}(S)=(\widehat{\boldsymbol{\gamma}}%
_{j}(S)^{\prime},\widehat{\tau}_{j}(S))^{\prime}%
\]
of $\boldsymbol{\vartheta}_{j}$ is defined as a maximiser over $\Theta_{j,S}$
of the average (quasi-) log-likelihood $L_{T,j}(\boldsymbol{\vartheta}%
_{j};S)=T^{-1}\sum_{t=1}^{T}l_{t}\{\eta_{j,t}(\boldsymbol{\vartheta}_{j}%
;S)\}$. The associated stagewise Wald statistic for testing $\tau_{j}=0$ is
\begin{equation}
W_{T,j}(S)=\frac{|\widehat{\tau}_{j}(S)|}{\widehat{\mathrm{se}}_{j}(S)},
\label{eq:app_wald}%
\end{equation}
where $\widehat{\mathrm{se}}_{j}(S)$ is the estimated standard error of
$\widehat{\tau}_{j}(S)$.

Putting $\boldsymbol{q}_{j,t}(S)=(\boldsymbol{x}_{S,t}^{\prime},x_{j,t}%
)^{\prime}\in\mathbb{R}^{s+1}$, the (quasi-) score vector for the augmented
model is
\begin{equation}
\boldsymbol{s}_{j,t}(\boldsymbol{\vartheta}_{j};S)=\frac{\partial}%
{\partial\boldsymbol{\vartheta}_{j}}l_{t}\{\eta_{j,t}(\boldsymbol{\vartheta
}_{j};S)\}=\boldsymbol{q}_{j,t}(S)\dot{l}_{t}\{\eta_{j,t}%
(\boldsymbol{\vartheta}_{j};S)\}. \label{eq:app_score_vector}%
\end{equation}
The population Hessian matrix is
\begin{align}
\boldsymbol{J}_{j}^{\ast}(S)  &  =-\mathbb{E}\left[  \frac{\partial^{2}%
}{\partial\boldsymbol{\vartheta}_{j}\partial\boldsymbol{\vartheta}_{j}%
^{\prime}}l_{t}\{\eta_{j,t}(\boldsymbol{\vartheta}_{j}^{\ast}(S);S)\}\right]
\label{eq:app_population_hessian}\\
&  =-\mathbb{E}\left[  \boldsymbol{q}_{j,t}(S)\boldsymbol{q}_{j,t}(S)^{\prime
}\ddot{l}_{t}\{\eta_{j,t}(\boldsymbol{\vartheta}_{j}^{\ast}(S);S)\}\right]  .
\label{eq:app_population_hessian_alt}%
\end{align}
For robust inference, define the long-run score covariance matrix at
$\boldsymbol{\vartheta}_{j}^{\ast}(S)$ as
\[
\boldsymbol{\Omega}_{j}^{\ast}(S)=\sum_{h=-\infty}^{\infty}\mathbb{E}\left[
\boldsymbol{s}_{j,t}(\boldsymbol{\vartheta}_{j}^{\ast}(S);S)\boldsymbol{s}%
_{j,t-h}(\boldsymbol{\vartheta}_{j}^{\ast}(S);S)^{\prime}\right]  ,
\]
whenever the series converges to a nonsingular sum. The corresponding sandwich
covariance matrix is
\begin{equation}
\boldsymbol{V}_{j}^{\ast}(S)=\boldsymbol{J}_{j}^{\ast}(S)^{-1}%
\boldsymbol{\Omega}_{j}^{\ast}(S)\boldsymbol{J}_{j}^{\ast}(S)^{-1}%
\in\mathbb{R}^{(s+1)\times(s+1)}. \label{eq:app_population_sandwich}%
\end{equation}
Note that, if the observations are conditionally independent, or the (quasi-)
score sequence is a martingale difference, $\boldsymbol{\Omega}_{j}^{\ast
}(S)=\mathbb{E}\left[  \boldsymbol{s}_{j,t}(\boldsymbol{\vartheta}_{j}^{\ast
}(S);S)\boldsymbol{s}_{j,t}(\boldsymbol{\vartheta}_{j}^{\ast}(S);S)^{\prime
}\right]  $. The standard error $\widehat{\mathrm{se}}_{j}(S)$ in
(\ref{eq:app_wald}) is assumed to estimate
$T^{-1/2}\{V_{\tau\tau,j}^{\ast}(S)\}^{1/2}$ uniformly, where
$V_{\tau\tau,j}^{\ast}(S)$ denotes the $(s+1,s+1)$ element of
$\boldsymbol{V}_{j}^{\ast}(S)$ corresponding to $\tau_j$.

Finally, the population noncentrality strength of $x_{j,t}$ is defined as
\begin{equation}
\Lambda_{T,j}^{\ast}(S)=\frac{\sqrt{T}\,|\tau_{j}^{\ast}(S)|}{\{V_{\tau\tau
,j}^{\ast}(S)\}^{1/2}}. \label{eq:app_lambda}%
\end{equation}
This is the deterministic population component of $W_{T,j}(S)$.

The BMT algorithm is as follows. Given a current conditioning set $S^{(\ell)}%
$, $\ell\geq0$, compute $W_{T,j}(S^{(\ell)})$ for all $j\notin S^{(\ell)}$. At
each stage $\ell$, BMT adds the candidate covariate with the largest Wald
statistic among those satisfying $W_{T,j}(S^{(\ell)})\geq c_{p}$, where
$c_{p}$ is a multiple-testing threshold. If no candidate crosses $c_{p}$, the
procedure stops; otherwise, the procedure continues until it either stops or a
prespecified maximum number of stages $k_{\max}$ is reached (which is assumed
to be no smaller than $k$).

The threshold $c_{p}$ may be set as $c_{p}(\alpha,\delta)=\Phi^{-1}\left(
1-\alpha/(2Cp^{\delta})\right)  $, where $0<\alpha<1$, $C>0$ and $\delta>0$. More
generally, the only properties used below are that $c_{p}$ is indexed by the
number of simultaneous candidate tests and that $c_{p}\rightarrow\infty$. If
the number of remaining candidates at stage $\ell$ is $p_{\ell}$, one may use
$c_{p_{\ell}}$; using $c_{p}$ throughout is a conservative simplification.
Since $p=p_{T}$, $c_{p}$ is implicitly sample-size dependent, but its origin
is multiplicity of candidate testing.

In subsequent analysis, we use two classes of conditioning sets. For
approximating-model arguments, we consider
\(\mathcal{S}_{k_{\max}}=\{S\subseteq\{1,\ldots,p\}:|S|\leq k_{\max}\},\)
whereas
\(\mathcal{S}_{0}=\{S\subseteq S_{0}:|S|<k\},\)
is used for exact-recovery arguments. Fixed controls, if present, can be added
to every conditioning set in either class.

%----------------------------------------------------

\subsection{Assumptions}

\label{app:glm_assumptions}
%----------------------------------------------------

The assumptions are divided into three groups. The first group gives uniform
likelihood expansions and maximal score control over all low-dimensional
conditioning sets. These conditions are used for the approximating-model and
stopping results. The second group gives the high-level dominance conditions
needed for exact recovery. The third group provides primitive sufficient
conditions for dominance through a local score-expansion argument.

%....................................................

\subsubsection*{Uniform likelihood and score conditions}

\begin{assumption}
\label{ass:app_smooth} The index log-likelihood $l_{t}(\eta)$ in
\eqref{eq:app_index_loglik} is three times continuously differentiable in
$\eta$ on the relevant index region. Its first three derivatives $\dot{l}%
_{t}(\eta)$, $\ddot{l}_{t}(\eta)$ and $\dddot{l}_{t}(\eta)$ are dominated by
random variables with moments sufficient for the uniform laws of large numbers
and maximal inequalities stated below.
\end{assumption}

\begin{assumption}
\label{ass:app_compact} For every $S\in\mathcal{S}_{k_{\max}}$ and every
$j\notin S$, the pseudo-true parameter $\boldsymbol{\vartheta}_{j}^{\ast}(S)$
lies in the interior of a compact subset of $\mathbb{R}^{|S|+1}$. Moreover,
the pseudo-true parameters are uniformly bounded and uniformly separated
from the boundaries of these compact sets. The estimator
$\widehat{\boldsymbol{\vartheta}}_{j}(S)$ lies in a neighbourhood of
$\boldsymbol{\vartheta}_{j}^{\ast}(S)$, uniformly over
$S\in\mathcal{S}_{k_{\max}}$ and $j\notin S$, with probability tending to one.
\end{assumption}

\begin{assumption}
\label{ass:app_hessian_ulln} For each $S\in\mathcal{S}_{k_{\max}}$, $j\notin
S$, and $\boldsymbol{\vartheta}_{j}\in\Theta_{j,S}$, define
\[
\boldsymbol{J}_{j}(\vartheta_{j};S)=-\mathbb{E}\left[  \frac{\partial
\boldsymbol{s}_{j,t}(\boldsymbol{\vartheta}_{j};S)}{\partial
\boldsymbol{\vartheta}_{j}^{\prime}}\right]  .
\]
Uniformly over $S\in\mathcal{S}_{k_{\max}}$, $j\notin S$, and
$\boldsymbol{\vartheta}_{j}$ in a neighbourhood of $\boldsymbol{\vartheta}%
_{j}^{\ast}(S)$,
\[
\left\Vert -T^{-1}\sum_{t=1}^{T}\frac{\partial\boldsymbol{s}_{j,t}%
(\boldsymbol{\vartheta}_{j};S)}{\partial\boldsymbol{\vartheta}_{j}^{\prime}%
}-\boldsymbol{J}_{j}(\boldsymbol{\vartheta}_{j};S)\right\Vert =o_{p}(1).
\]
In addition, $\boldsymbol{J}_{j}(\boldsymbol{\vartheta}_{j};S)$ is uniformly
continuous at $\boldsymbol{\vartheta}_{j}^{\ast}(S)$: whenever
\[
\sup_{S\in\mathcal{S}_{k_{\max}}}\sup_{j\notin S}\Vert\boldsymbol{\vartheta
}_{j}-\boldsymbol{\vartheta}_{j}^{\ast}(S)\Vert\rightarrow0,
\]
we have
\[
\sup_{S\in\mathcal{S}_{k_{\max}}}\sup_{j\notin S}\Vert\boldsymbol{J}%
_{j}(\boldsymbol{\vartheta}_{j};S)-\boldsymbol{J}_{j}^{\ast}(S)\Vert
\rightarrow0.
\]

\end{assumption}

\begin{assumption}
\label{ass:app_info} There exist constants $0<c<C<\infty$ such that, uniformly
over $S\in\mathcal{S}_{k_{\max}}$ and $j\notin S$,
\(c\leq\lambda_{\min}\{\boldsymbol{J}_{j}^{\ast}(S)\}\leq\lambda_{\max }\{\boldsymbol{J}_{j}^{\ast}(S)\}\leq C,\) where $\lambda_{\min}\{\boldsymbol{A}\}$ and $\lambda_{\max}\{\boldsymbol{A}%
\}$ denote the minimum and maximum eigenvalues of $\boldsymbol{A}$,
respectively.
When robust standard errors are used, the same type of uniform boundedness
holds for $\boldsymbol{V}_{j}^{\ast}(S)$. In particular, there exist constants
$0<\underline{v}<\overline{v}<\infty$ such that
\(\underline{v}\leq V_{\tau\tau,j}^{\ast}(S)\leq\overline{v},\)
uniformly over $S\in\mathcal{S}_{k_{\max}}$ and $j\notin S$.
\end{assumption}

\begin{assumption}
\label{ass:app_score} For $S\in\mathcal{S}_{k_{\max}}$ and $j\notin S$,
define
\[
Z_{T,j}(S)=\frac{\boldsymbol{e}_{\tau}^{\prime}\boldsymbol{J}_{j}^{\ast
}(S)^{-1}T^{-1/2}\sum_{t=1}^{T}\boldsymbol{s}_{j,t}(\boldsymbol{\vartheta}%
_{j}^{\ast}(S);S)}{\{V_{\tau\tau,j}^{\ast}(S)\}^{1/2}},
\]
where
\(\boldsymbol{e}_{\tau}=(0,\ldots,0,1)^{\prime}\in\mathbb{R}^{|S|+1}\)
selects the candidate coefficient. Let
\[
\mathcal{Z}_{T}=\max_{S\in\mathcal{S}_{k_{\max}}}\max_{j\notin S}|Z_{T,j}(S)|
\]
and
\[
\mathcal{L}_{T}=\max_{S\in\mathcal{S}_{k_{\max}}}\max_{j\notin S}\Lambda
_{T,j}^{\ast}(S).
\]
There exist nonnegative random sequences $\rho_{T}$ and $\Delta_{T}$ such
that, uniformly over $S\in\mathcal{S}_{k_{\max}}$ and $j\notin S$,
\[
\left\vert \frac{\sqrt{T}\{\widehat{\tau}_{j}(S)-\tau_{j}^{\ast}%
(S)\}}{\{V_{\tau\tau,j}^{\ast}(S)\}^{1/2}}-Z_{T,j}(S)\right\vert \leq\rho
_{T},
\]
with probability tending to one, and
\[
\Delta_{T}=\max_{S\in\mathcal{S}_{k_{\max}}}\max_{j\notin S}\left\vert
\frac{\widehat{\mathrm{se}}_{j}(S)}{T^{-1/2}\{V_{\tau\tau,j}^{\ast}%
(S)\}^{1/2}}-1\right\vert =o_{p}(1).
\]

\end{assumption}

Define the uniform approximation remainder
\(\mathcal{R}_{T}=2\rho_{T}+2\Delta_{T}(\mathcal{L}_{T}+\mathcal{Z}_{T}).\)
This term explicitly records the effect of standard-error estimation on
possibly diverging population noncentralities.

\begin{assumption}
\label{ass:app_se} The multiple-testing threshold $c_{p}$ satisfies
$c_{p}\rightarrow\infty$. Moreover, there exists a constant $\bar{\eta}_{0}%
\in(0,1)$ such that
\(\mathbb{P}\left( \mathcal{Z}_{T}+\mathcal{R}_{T}\leq(1-\bar{\eta}_{0})c_{p}\right) \rightarrow1.\)
Equivalently, the threshold is asymptotically separated from the maximal
uniform stochastic error.
\end{assumption}

%....................................................

\subsubsection*{High-level conditions for exact recovery}

\begin{assumption}
\label{ass:app_dom} There exists a deterministic sequence $d_{T}>0$ such that
\(\frac{\mathcal{Z}_{T}+\mathcal{R}_{T}}{d_{T}}=o_{p}(1),\)
and, for every $S\in\mathcal{S}_{0}$,
\[
\max_{j\in S_{0}\setminus S}\Lambda_{T,j}^{\ast}(S)\geq\max_{\substack{m\notin
S_{0}\\m\notin S}}\Lambda_{T,m}^{\ast}(S)+d_{T}.
\]
In addition, at every such $S$, the strongest remaining signal crosses the
multiple-testing threshold by the same population margin:
\[
\max_{j\in S_{0}\setminus S}\Lambda_{T,j}^{\ast}(S)\geq c_{p}+d_{T}.
\]

\end{assumption}

\begin{assumption}
\label{ass:app_exhaust} Once all signals are included, the remaining
population noncentralities are asymptotically below the multiple-testing
threshold:
\(\max_{m\notin S_{0}} \Lambda_{T,m}^{*}(S_{0}) = o(c_{p}).\)

\end{assumption}

%....................................................

\subsubsection*{Primitive local conditions for dominance}

The next set of conditions give one route to Assumption~\ref{ass:app_dom}.
They avoid a deterministic bound on the full index $\boldsymbol{x}_{t}%
^{\prime}\boldsymbol{\beta}_{0}$. Instead, they impose a local population
score expansion, stable curvature, and a negligible nonlinear remainder.

Before stating the local primitive conditions, define the restricted
pseudo-true parameter for a model containing only the variables in $S$:
\begin{equation}
\boldsymbol{\gamma}^{\ast}(S)\in\arg\max_{\boldsymbol{\gamma}}\mathbb{E}%
\left[  l_{t}(\boldsymbol{x}_{S,t}^{\prime}\boldsymbol{\gamma})\right]  .
\label{eq:app_gammas}%
\end{equation}
This maximiser is assumed to be unique. Let
\begin{equation}
\eta_{S,t}^{\ast}=\boldsymbol{x}_{S,t}^{\prime}\boldsymbol{\gamma}^{\ast
}(S)\label{eq:app_etas}%
\end{equation}
and define the restricted population score residual
\begin{equation}
\xi_{t}(S)=\dot{l}_{t}(\eta_{S,t}^{\ast}). \label{eq:app_ksi}%
\end{equation}
Furthermore, let
\begin{equation}
w_{t}(S)=-\mathbb{E}\left[  \ddot{l}_{t}(\eta_{S,t}^{\ast})|\boldsymbol{x}%
_{t}\right]  \label{eq:app_weight}%
\end{equation}
denote the local curvature weight. For $j\notin S$, define $\boldsymbol{\pi
}_{j}(S)$ and $r_{j,t}(S)$ by the weighted population projection
\begin{equation}
x_{j,t}=\boldsymbol{x}_{S,t}^{\prime}\boldsymbol{\pi}_{j}(S)+r_{j,t}%
(S),\qquad\mathbb{E}\left[  w_{t}(S)\boldsymbol{x}_{S,t}r_{j,t}(S)\right]  =0.
\label{eq:app_proj}%
\end{equation}
The projection matrix $\mathbb{E}\left[  w_{t}(S)\boldsymbol{x}_{S,t}%
\boldsymbol{x}_{S,t}^{\prime}\right]  $ is assumed to be nonsingular uniformly
over the relevant conditioning sets. For non-canonical links, $w_{t}(S)$ need
not be positive automatically. Positivity and stability of the profile
curvature are imposed below as primitive conditions rather than derived from
the GLM form alone.

\begin{assumption}
\label{ass:app_local_expansion}
For $S\in\mathcal S_0$ and $j\notin S$, let
$\boldsymbol\vartheta_j^0(S)=(\boldsymbol\gamma^*(S)',0)'$ and define the
negative population Hessian at this restricted point by
\[
\boldsymbol J_j^0(S)
=
-\mathbb E\!\left[
\boldsymbol q_{j,t}(S)\boldsymbol q_{j,t}(S)'
\ddot l_t(\eta_{S,t}^*)
\right]
=
\begin{pmatrix}
\boldsymbol J_{\gamma\gamma,j}^0(S)&J_{\gamma\tau,j}^0(S)\\
J_{\tau\gamma,j}^0(S)&J_{\tau\tau,j}^0(S)
\end{pmatrix}.
\]
The restricted nuisance block is nonsingular, and the profile curvature
\[
A_j(S)
=
J_{\tau\tau,j}^0(S)
-J_{\tau\gamma,j}^0(S)
\{\boldsymbol J_{\gamma\gamma,j}^0(S)\}^{-1}
J_{\gamma\tau,j}^0(S)
=
\mathbb E\{w_t(S)r_{j,t}^2(S)\}
\]
satisfies
\(0<\underline A\leq\inf_{S\in\mathcal S_0}\inf_{j\notin S}A_j(S)
\leq\sup_{S\in\mathcal S_0}\sup_{j\notin S}A_j(S)\leq\overline A<\infty\).
For every such $(S,j)$, the profiled population score equation admits the
uniform expansion
\[
0=A_j(S)\tau_j^*(S)-B_j(S)+R_j(S),
\qquad
B_j(S)=\mathbb E\{r_{j,t}(S)\xi_t(S)\}.
\]
\end{assumption}

\begin{assumption}
\label{ass:app_local_remainder} There exists a deterministic sequence
$a_{T}=o(1)$ such that
\(\frac{\mathcal{Z}_{T}+\mathcal{R}_{T}}{\sqrt{T}\,a_{T}}=o_{p}(1)\)
and
\[
\sup_{S\in\mathcal{S}_{0}}\sup_{j\notin S}|R_{j}(S)|=o(a_{T}).
\]

\end{assumption}

\begin{assumption}
\label{ass:app_local_gap} For every $S\in\mathcal{S}_{0}$,
\[
\max_{j\in S_{0}\setminus S}\frac{|B_{j}(S)|}{A_{j}(S)\{V_{\tau\tau,j}^{\ast
}(S)\}^{1/2}}\geq\max_{\substack{m\notin S_{0}\\m\notin S}}\frac{|B_{m}%
(S)|}{A_{m}(S)\{V_{\tau\tau,m}^{\ast}(S)\}^{1/2}}+a_{T}.
\]
Moreover, the strongest remaining signal crosses the threshold with a margin:
there exists a deterministic sequence $d_{T}^{(2)}>0$ such that
\(\frac{\mathcal{Z}_{T}+\mathcal{R}_{T}}{d_{T}^{(2)}}=o_{p}(1)\)
and
\[
\max_{j\in S_{0}\setminus S}\Lambda_{T,j}^{\ast}(S)\geq c_{p}+d_{T}^{(2)}%
\]
for every $S\in\mathcal{S}_{0}$.
\end{assumption}

%....................................................

\subsubsection*{Alternative global route}

The local expansion is not the only way to verify stagewise dominance. A more
direct, but higher-level, alternative is to impose a global separation
condition on variance-normalised pseudo-true coefficients.

\begin{assumption}
[Global variance-normalised pseudo-true separation]\label{ass:app_global_gap}
There exists a deterministic sequence $b_{T}>0$ such that
\(\frac{\mathcal{Z}_{T}+\mathcal{R}_{T}}{\sqrt{T}\,b_{T}}=o_{p}(1)\)
and, for every $S\in\mathcal{S}_{0}$,
\[
\max_{j\in S_{0}\setminus S}\frac{|\tau_{j}^{\ast}(S)|}{\{V_{\tau\tau,j}%
^{\ast}(S)\}^{1/2}}\geq\max_{\substack{m\notin S_{0}\\m\notin S}}\frac
{|\tau_{m}^{\ast}(S)|}{\{V_{\tau\tau,m}^{\ast}(S)\}^{1/2}}+b_{T}.
\]
Moreover, the strongest remaining signal crosses the threshold with a margin:
there exists a deterministic sequence $d_{T}^{(2)}>0$ such that
\(\frac{\mathcal{Z}_{T}+\mathcal{R}_{T}}{d_{T}^{(2)}}=o_{p}(1)\)
and
\[
\max_{j\in S_{0}\setminus S}\Lambda_{T,j}^{\ast}(S)\geq c_{p}+d_{T}^{(2)}%
\]
for every $S\in\mathcal{S}_{0}$.
\end{assumption}

Assumptions~\ref{ass:app_smooth}--\ref{ass:app_se} yield the uniform bound
\[
 |W_{T,j}(S)-\Lambda_{T,j}^{\ast}(S)|
 \leq |Z_{T,j}(S)|+\mathcal R_T,
\]
over all admissible conditioning sets and candidates. Uniformity is needed
because the BMT path is random. The threshold condition allows the maximal
score and approximation errors to be of order $\sqrt{\ell_T}$, provided that
the multiplicity correction leaves a strict margin below $c_p$.

Assumption~\ref{ass:app_dom} is the central exact-recovery condition. Since
BMT adds one covariate at a time, it is enough that the strongest remaining
signal dominate all non-signals by a gap $d_T$ that exceeds
$\mathcal Z_T+\mathcal R_T$, and that this signal also exceed $c_p+d_T$.
Assumption~\ref{ass:app_exhaust} then requires only that the remaining
population noncentralities be $o(c_p)$ after $S_0$ has been selected; exact
nullity is sufficient but not necessary.

Assumptions~\ref{ass:app_local_expansion}--\ref{ass:app_local_gap} provide a
primitive route through
\(0=A_j(S)\tau_j^{\ast}(S)-B_j(S)+R_j(S)\), with
\(R_j(S)=o(a_T)\) uniformly. Thus the relevant score moment is
$|B_j(S)|$ after normalization by
$A_j(S)\{V_{\tau\tau,j}^{\ast}(S)\}^{1/2}$, rather than the unnormalised
moment itself. Assumption~\ref{ass:app_global_gap} is the more
direct alternative, imposing the required separation on the
variance-normalised pseudo-true coefficients themselves. These are sufficient,
not necessary, conditions; their purpose is to isolate the population
ordering and threshold crossing needed by the algorithm.

\subsection{Model-specific verification of the local score expansion}
\label{app:model_specific_local_verification}

The expansion in Assumption~\ref{ass:app_local_expansion} follows from the
profile likelihood and therefore accounts for re-optimisation of the nuisance
coefficient. Fix $S\in\mathcal S_0$ and $j\notin S$. In a neighbourhood of
zero, let $\boldsymbol\gamma_j(\tau;S)$ solve
\[
\mathbb E\!\left[
\boldsymbol x_{S,t}\dot l_t\{\boldsymbol x_{S,t}'
\boldsymbol\gamma_j(\tau;S)+\tau x_{j,t}\}
\right]=\boldsymbol0,
\qquad
\boldsymbol\gamma_j(0;S)=\boldsymbol\gamma^*(S).
\]
Uniform nonsingularity of the nuisance Hessian and the implicit function
theorem give a locally unique differentiable path. Define
\begin{equation}
Q_j(\tau;S)
=
\mathbb E\!\left[
l_t\{\boldsymbol x_{S,t}'\boldsymbol\gamma_j(\tau;S)+\tau x_{j,t}\}
\right],
\qquad
\Psi_j(\tau;S)=\dot Q_j(\tau;S).
\label{eq:app_profile_score}
\end{equation}
The envelope theorem gives
\(\Psi_j(\tau;S)=\mathbb E[x_{j,t}\dot l_t\{
\boldsymbol x_{S,t}'\boldsymbol\gamma_j(\tau;S)+\tau x_{j,t}\}]\).
At $\tau=0$, the restricted nuisance score is zero; hence, using
$x_{j,t}=\boldsymbol x_{S,t}'\boldsymbol\pi_j(S)+r_{j,t}(S)$,
\(\Psi_j(0;S)=\mathbb E[r_{j,t}(S)\xi_t(S)]=B_j(S)\).

Differentiating the nuisance first-order condition at zero gives
\[
\dot{\boldsymbol\gamma}_j(0;S)
=
-\{\mathbb E[w_t(S)\boldsymbol x_{S,t}\boldsymbol x_{S,t}']\}^{-1}
\mathbb E[w_t(S)\boldsymbol x_{S,t}x_{j,t}]
=
-\boldsymbol\pi_j(S).
\]
Consequently,
\[
-\dot\Psi_j(0;S)
=
J_{\tau\tau,j}^0(S)
-J_{\tau\gamma,j}^0(S)
\{\boldsymbol J_{\gamma\gamma,j}^0(S)\}^{-1}
J_{\gamma\tau,j}^0(S)
=
\mathbb E\{w_t(S)r_{j,t}^2(S)\}
=
A_j(S).
\]
The pseudo-true coefficient is an interior maximiser of the profile
criterion, so $\Psi_j\{\tau_j^*(S);S\}=0$. Taylor expansion around zero gives
\begin{equation}
0=A_j(S)\tau_j^*(S)-B_j(S)+R_j(S),
\qquad
R_j(S)
=
-\frac12\ddot\Psi_j\{\bar\tau_j(S);S\}\{\tau_j^*(S)\}^2,
\label{eq:app_model_specific_expansion}
\end{equation}
for some $\bar\tau_j(S)$ between zero and $\tau_j^*(S)$.

One convenient set of primitive sufficient conditions is as follows. Along
the profile paths with $|\tau|\leq Ca_T$, the nuisance Hessian is uniformly
nonsingular with uniformly bounded inverse. The index log-likelihood has three
derivatives on the corresponding index neighbourhoods and, for fixed
$k_{\max}$,
\[
\sup_{S\in\mathcal S_0}\sup_{j\notin S}
\mathbb E\!\left[
\sup_{\eta\in\mathcal N_{j,S}}|\dddot l_t(\eta)|
\{\|\boldsymbol x_{S,t}\|+|x_{j,t}|\}^3
\right]<\infty,
\]
together with the analogous second-derivative moment needed to differentiate
the nuisance first-order condition twice. These conditions imply
$\sup_{S,j,|\tau|\leq Ca_T}|\ddot\Psi_j(\tau;S)|<\infty$. If, in addition,
$\sup_{S\in\mathcal S_0}\sup_{j\notin S}|\tau_j^*(S)|=O(a_T)$ with
$a_T\to0$, then \eqref{eq:app_model_specific_expansion} gives
$\sup_{S,j}|R_j(S)|=O(a_T^2)=o(a_T)$. Uniform upper and lower bounds on
$A_j(S)$ complete the verification of Assumptions
\ref{ass:app_local_expansion} and \ref{ass:app_local_remainder}.

For the canonical Poisson model,
\(l_t(\eta)=y_t\eta-\exp(\eta)-\log(y_t!)\), so
$\ddot l_t(\eta)=\dddot l_t(\eta)=-\exp(\eta)$ and
$w_t(S)=\exp(\eta_{S,t}^*)$. The preceding conditions reduce to stable
exponential-weighted residual second moments and uniform integrability of the
corresponding third moments along the profile paths. For Gamma and related
GLMs, the same argument applies with weights determined by the relevant mean
and variance functions; the required weighted moment and curvature bounds
must be checked for the chosen link and parametrisation.

\subsection{Results}

\label{app:glm_theorems}
%----------------------------------------------------

We now state the main results. The first two results are approximating-model
results; they do not require any dominance condition. The third result gives
exact recovery under the relaxed stagewise dominance condition. The fourth and
fifth results show how exact recovery follows from the local and global
primitive conditions, respectively. The final result gives oracle inference.

\begin{lemma}
\label{lem:app_wald} Suppose Assumptions~\ref{ass:app_smooth},
\ref{ass:app_compact}, \ref{ass:app_hessian_ulln}, \ref{ass:app_info}, and
\ref{ass:app_score} hold. Then, uniformly over $S\in\mathcal{S}_{k_{\max}}$
and $j\notin S$,
\begin{equation}
\left\vert W_{T,j}(S)-\Lambda_{T,j}^{\ast}(S)\right\vert \leq|Z_{T,j}%
(S)|+\mathcal{R}_{T}, \label{eq:app_wald_expansion}%
\end{equation}
with probability tending to one.
\end{lemma}

\begin{theorem}
\label{thm:app_no_null} Suppose Assumptions~\ref{ass:app_smooth},
\ref{ass:app_compact}, \ref{ass:app_hessian_ulln}, \ref{ass:app_info},
\ref{ass:app_score}, and \ref{ass:app_se} hold. Then,
\[
\mathbb{P}\left(  \exists\,S\in\mathcal{S}_{k_{\max}},\ \exists\,j\notin
S:\Lambda_{T,j}^{\ast}(S)=0,\ W_{T,j}(S)\geq c_{p}\right)  \rightarrow0.
\]
Consequently, along the BMT path, no candidate with zero population
noncentrality at the stage where it is tested is selected, with probability
tending to one.
\end{theorem}

\begin{theorem}
[Stopping under exhaustion]\label{thm:app_stop} Suppose
Assumptions~\ref{ass:app_smooth}, \ref{ass:app_compact},
\ref{ass:app_hessian_ulln}, \ref{ass:app_info}, \ref{ass:app_score}, and
\ref{ass:app_se} hold. Let
\[
M_{\ell,T}^{*} = \max_{j\notin S^{(\ell)}} \Lambda_{T,j}^{*}(S^{(\ell)}).
\]
If there exists a stage $\ell^{\dagger}\le k_{\max}$ such that
\(M_{\ell^{\dagger},T}^{*}=o_{p}(c_{p}),\)
then
\(\mathbb{P}(\widehat{L}\le\ell^{\dagger})\to1,\)
where $\widehat L$ is the number of successful selections before stopping
(and equals $k_{\max}$ if the stage cap is reached). For a possibly random
$\ell^\dagger$, the path is extended after stopping by setting
$S^{(\ell)}=S^{(\widehat L)}$ for $\ell\geq\widehat L$.
\end{theorem}

\begin{theorem}
[Exact recovery under relaxed stagewise dominance]\label{thm:app_exact}
Suppose Assumptions~\ref{ass:app_smooth}, \ref{ass:app_compact},
\ref{ass:app_hessian_ulln}, \ref{ass:app_info}, \ref{ass:app_score},
\ref{ass:app_se}, \ref{ass:app_dom}, and \ref{ass:app_exhaust} hold. Suppose
also that $k_{\max}\ge k$. Then
\(\mathbb{P}(\widehat{S}=S_{0})\to1.\)
Here $\widehat{S}$ denotes the selected set of non-always-in covariates. If
always-in controls are present, the final selected model is the union of those
controls and $S_{0}$.
\end{theorem}

\begin{theorem}
\label{thm:app_local} Suppose Assumptions~\ref{ass:app_smooth},
\ref{ass:app_compact}, \ref{ass:app_hessian_ulln}, \ref{ass:app_info},
\ref{ass:app_score}, \ref{ass:app_se}, \ref{ass:app_local_expansion},
\ref{ass:app_local_remainder}, \ref{ass:app_local_gap}, and
\ref{ass:app_exhaust} hold and $k_{\max}\geq k$. Then
\(\mathbb{P}(\widehat{S}=S_{0})\to1.\)

\end{theorem}

\begin{theorem}
\label{thm:app_global} Suppose Assumptions~\ref{ass:app_smooth},
\ref{ass:app_compact}, \ref{ass:app_hessian_ulln}, \ref{ass:app_info},
\ref{ass:app_score}, \ref{ass:app_se}, \ref{ass:app_global_gap}, and
\ref{ass:app_exhaust} hold and $k_{\max}\geq k$. Then
\(\mathbb{P}(\widehat{S}=S_{0})\to1.\)

\end{theorem}

For the oracle result, let $\widehat{\boldsymbol\beta}_{\mathrm{or}}$ denote the
infeasible MLE or quasi-MLE of the subvector $\boldsymbol{\beta}_{0,S_{0}}$ of
$\boldsymbol{\beta}_{0}$ corresponding to the true active set $S_{0}$. It is
computed using the same maximisation and deterministic tie-breaking rule as
the post-BMT refit. Define the oracle score by
\[
\boldsymbol{s}_{t}(\boldsymbol{\beta}_{0,S_{0}})=\boldsymbol{x}_{S_{0},t}%
\dot{l}_{t}(\boldsymbol{x}_{S_{0},t}^{\prime}\boldsymbol{\beta}_{0,S_{0}}).
\]
If serial dependence is allowed, let
\[
\boldsymbol{\Omega}_{0}=\sum_{h=-\infty}^{\infty}\mathbb{E}\left[
\boldsymbol{s}_{t}(\boldsymbol{\beta}_{0,S_{0}})\boldsymbol{s}_{t-h}%
(\boldsymbol{\beta}_{0,S_{0}})^{\prime}\right]  ,
\]
whenever the series converges to a nonsingular sum. Under independent
observations or an uncorrelated martingale-difference score, this reduces to
\[
\boldsymbol{\Omega}_{0}=\mathbb{E}\left[  \boldsymbol{s}_{t}(\boldsymbol{\beta
}_{0,S_{0}})\boldsymbol{s}_{t}(\boldsymbol{\beta}_{0,S_{0}})^{\prime}\right]
.
\]

\begin{assumption}
[Oracle regularity]\label{ass:app_oracle_regularity} The oracle score
satisfies
\[
T^{-1/2}\sum_{t=1}^{T}\boldsymbol{s}_{t}(\boldsymbol{\beta}_{0,S_{0}%
})\Rightarrow N(\boldsymbol{0},\boldsymbol{\Omega}_{0}).
\]
The oracle Hessian
\[
\boldsymbol{J}_{0}=-\mathbb{E}\left[
\boldsymbol{x}_{S_{0},t}\boldsymbol{x}_{S_{0},t}^{\prime}
\ddot{l}_{t}(\boldsymbol{x}_{S_{0},t}^{\prime}\boldsymbol{\beta}_{0,S_{0}})
\right]
\]
is nonsingular, and the oracle estimator admits the expansion
\[
\sqrt{T}(\widehat{\boldsymbol{\beta}}_{\mathrm{or}}-\boldsymbol{\beta
}_{0,S_{0}})=\boldsymbol{J}_{0}^{-1}T^{-1/2}\sum_{t=1}^{T}\boldsymbol{s}%
_{t}(\boldsymbol{\beta}_{0,S_{0}})+o_{p}(1).
\]

\end{assumption}

\begin{theorem}
[Oracle inference]\label{thm:app_oracle} Suppose
Assumption~\ref{ass:app_oracle_regularity} holds. Suppose also that
\(\mathbb{P}(\widehat{S}=S_{0})\rightarrow1.\)
Define the post-BMT estimator as the MLE or quasi-MLE on the selected
support $\widehat S$. Embed this estimator in $\mathbb R^p$ by setting
coefficients outside $\widehat S$ equal to zero, and write
$\widehat{\boldsymbol\beta}_{\widehat S,S_0}$ for its $S_0$ subvector.
Then,
\[
\sqrt{T}(\widehat{\boldsymbol{\beta}}_{\widehat{S},S_0}-\boldsymbol{\beta
}_{0,S_{0}})\Rightarrow N(\boldsymbol{0},\boldsymbol{J}_{0}^{-1}%
\boldsymbol{\Omega}_{0}\boldsymbol{J}_{0}^{-1}).
\]
Thus, the post-BMT estimator has the same limiting distribution as the oracle
estimator that knows $S_{0}$ in advance.
\end{theorem}

%----------------------------------------------------

\subsection{Proofs}
\label{app:glm_proofs}

\begin{proof}[Proof of Lemma~\ref{lem:app_wald}]
Fix $S\in\mathcal S_{k_{\max}}$ and $j\notin S$, and put
\(A_{T,j}(S)=\sqrt T\,\tau_j^{\ast}(S)/\{V_{\tau\tau,j}^{\ast}(S)\}^{1/2}\),
so that $\Lambda_{T,j}^{\ast}(S)=|A_{T,j}(S)|$.  Define $R_{T,j}(S)$ by
\[
 \frac{\sqrt T\{\widehat\tau_j(S)-\tau_j^{\ast}(S)\}}
 {\{V_{\tau\tau,j}^{\ast}(S)\}^{1/2}}
 =Z_{T,j}(S)+R_{T,j}(S).
\]
Assumption~\ref{ass:app_score} gives
$\sup_{S,j}|R_{T,j}(S)|\leq\rho_T$ with probability tending to one. Write
\(\widehat{\mathrm{se}}_j(S)=T^{-1/2}\{V_{\tau\tau,j}^{\ast}(S)\}^{1/2}
\{1+\delta_{T,j}(S)\}\), where
$\sup_{S,j}|\delta_{T,j}(S)|\leq\Delta_T=o_p(1)$. Let $H_T$ be the joint
event on which both uniform bounds hold and $\Delta_T\leq1/2$; then
$\mathbb P(H_T)\to1$. On $H_T$,
\[
 \frac{\widehat\tau_j(S)}{\widehat{\mathrm{se}}_j(S)}
 =\frac{A_{T,j}(S)+Z_{T,j}(S)+R_{T,j}(S)}{1+\delta_{T,j}(S)}.
\]
Moreover, $1+\delta_{T,j}(S)>0$ and
$|(1+\delta)^{-1}-1|\leq2|\delta|$.  The reverse triangle inequality therefore yields
\begin{align*}
 |W_{T,j}(S)-\Lambda_{T,j}^{\ast}(S)|
 &\leq |Z_{T,j}(S)|+|R_{T,j}(S)| \\
 &\quad+2|\delta_{T,j}(S)|
 \{|A_{T,j}(S)|+|Z_{T,j}(S)|+|R_{T,j}(S)|\} \\
 &\leq |Z_{T,j}(S)|+2\rho_T
       +2\Delta_T(\mathcal L_T+\mathcal Z_T) \\
 &=|Z_{T,j}(S)|+\mathcal R_T.
\end{align*}
The event on which these inequalities hold has probability tending to one,
and the bound is uniform over all admissible $(S,j)$.
\end{proof}

\begin{proof}[Proof of Theorem~\ref{thm:app_no_null}]
Let $H_T$ denote the event on which the bound in Lemma~\ref{lem:app_wald}
holds uniformly; then $\mathbb P(H_T)\to1$.  On $H_T$, every population-null
candidate satisfies \(W_{T,j}(S)\leq |Z_{T,j}(S)|+\mathcal R_T\), and hence
\[
 \max_{\substack{S\in\mathcal S_{k_{\max}},\ j\notin S\\
                  \Lambda_{T,j}^{\ast}(S)=0}}
 W_{T,j}(S)
 \leq \mathcal Z_T+\mathcal R_T.
\]
Intersecting $H_T$ with the event in Assumption~\ref{ass:app_se}, whose
probability also tends to one, makes the right-hand side no larger than
$(1-\bar\eta_0)c_p<c_p$.  This proves the uniform statement.  The conclusion
for the realised BMT path follows because every conditioning set visited by
the procedure belongs to $\mathcal S_{k_{\max}}$.
\end{proof}

\begin{proof}[Proof of Theorem~\ref{thm:app_stop}]
Let $\ell^{\dagger}$ satisfy the stated condition and let $H_T$ be the
uniform high-probability event in Lemma~\ref{lem:app_wald}.  On $H_T$,
\[
 \max_{j\notin S^{(\ell^{\dagger})}}
 W_{T,j}(S^{(\ell^{\dagger})})
 \leq M_{\ell^{\dagger},T}^{\ast}+\mathcal Z_T+\mathcal R_T.
\]
Since $M_{\ell^{\dagger},T}^{\ast}=o_p(c_p)$, the event
$\{M_{\ell^{\dagger},T}^{\ast}\leq\bar\eta_0c_p/2\}$ has probability
tending to one.  Assumption~\ref{ass:app_se} supplies another such event on
which $\mathcal Z_T+\mathcal R_T\leq(1-\bar\eta_0)c_p$.  On the intersection
of these three events,
\[
 \max_{j\notin S^{(\ell^{\dagger})}}
 W_{T,j}(S^{(\ell^{\dagger})})
 \leq (1-\bar\eta_0/2)c_p<c_p.
\]
No remaining candidate can then pass the multiple-testing filter, and BMT
stops no later than stage $\ell^{\dagger}$.
\end{proof}

\begin{proof}[Proof of Theorem~\ref{thm:app_exact}]
Let $H_T$ be the uniform high-probability event in Lemma~\ref{lem:app_wald},
put $\mathcal E_T^{\ast}=\mathcal Z_T+\mathcal R_T$, and define
$E_T=H_T\cap\{\mathcal E_T^{\ast}\leq d_T/4\}$.  Assumption~\ref{ass:app_dom}
then gives $\mathbb P(E_T)\to1$.  We show by induction that, on $E_T$, every
variable selected during the first $k$ stages belongs to $S_0$.

Suppose that only signals have been selected up to some stage $\ell<k$.
Then $S^{(\ell)}\in\mathcal S_0$.  Choose
\[
 j_{\ell}^{\ast}\in
 \arg\max_{j\in S_0\setminus S^{(\ell)}}
 \Lambda_{T,j}^{\ast}(S^{(\ell)}).
\]
For each non-signal $m\notin S_0$ with $m\notin S^{(\ell)}$, Lemma~\ref{lem:app_wald},
Assumption~\ref{ass:app_dom}, and the definition of $E_T$ give
\begin{align*}
 W_{T,j_{\ell}^{\ast}}(S^{(\ell)})-W_{T,m}(S^{(\ell)})
 &\geq
 \Lambda_{T,j_{\ell}^{\ast}}^{\ast}(S^{(\ell)})
 -\Lambda_{T,m}^{\ast}(S^{(\ell)})-2\mathcal E_T^{\ast} \\
 &\geq d_T/2>0.
\end{align*}
Thus no non-signal can maximise the sample Wald statistic at stage $\ell$.
The threshold part of Assumption~\ref{ass:app_dom} also gives
\(W_{T,j_{\ell}^{\ast}}(S^{(\ell)})\geq c_p+d_T-
\mathcal E_T^{\ast}\geq c_p+3d_T/4\).  Hence the passing set is nonempty and
BMT selects a signal.  This completes the induction, so after $k$ stages the
selected set is $S_0$.

It remains to exclude further selections.  By Assumption~\ref{ass:app_exhaust},
$\max_{m\notin S_0}\Lambda_{T,m}^{\ast}(S_0)=o(c_p)$, so this maximum is at
most $\bar\eta_0c_p/2$ for all sufficiently large $T$.  Intersecting $H_T$
with the event in Assumption~\ref{ass:app_se}, we obtain
\[
 \max_{m\notin S_0}W_{T,m}(S_0)
 \leq \bar\eta_0c_p/2+(1-\bar\eta_0)c_p
 =(1-\bar\eta_0/2)c_p<c_p
\]
with probability tending to one.  Thus BMT stops after selecting $S_0$, and
$\mathbb P(\widehat S=S_0)\to1$.
\end{proof}

\begin{proof}[Proof of Theorem~\ref{thm:app_local}]
Fix $S\in\mathcal S_0$ and $j\notin S$.  The expansion in
Assumption~\ref{ass:app_local_expansion} gives
\(\tau_j^{\ast}(S)=B_j(S)/A_j(S)-R_j(S)/A_j(S)\).  The lower bounds on
$A_j(S)$ and $V_{\tau\tau,j}^{\ast}(S)$ therefore imply, uniformly in $(S,j)$,
\[
 \left|
 \frac{|\tau_j^{\ast}(S)|}{\{V_{\tau\tau,j}^{\ast}(S)\}^{1/2}}
 -\frac{|B_j(S)|}{A_j(S)\{V_{\tau\tau,j}^{\ast}(S)\}^{1/2}}
 \right|
 \leq C|R_j(S)|=o(a_T).
\]
Let the uniform remainder in this display be $e_T=o(a_T)$.  The gap in
Assumption~\ref{ass:app_local_gap} then implies, for every $S\in\mathcal S_0$,
\begin{align*}
 \max_{j\in S_0\setminus S}
 \frac{|\tau_j^{\ast}(S)|}{\{V_{\tau\tau,j}^{\ast}(S)\}^{1/2}}
 &\geq
 \max_{\substack{m\notin S_0\\m\notin S}}
 \frac{|\tau_m^{\ast}(S)|}{\{V_{\tau\tau,m}^{\ast}(S)\}^{1/2}}
 +a_T-2e_T \\
 &\geq
 \max_{\substack{m\notin S_0\\m\notin S}}
 \frac{|\tau_m^{\ast}(S)|}{\{V_{\tau\tau,m}^{\ast}(S)\}^{1/2}}
 +a_T/2
\end{align*}
for all sufficiently large $T$.  After multiplication by $\sqrt T$, the
population noncentrality gap is at least
$d_T^{(1)}=\sqrt T\,a_T/2$.  Assumption~\ref{ass:app_local_remainder} gives
$(\mathcal Z_T+\mathcal R_T)/d_T^{(1)}=o_p(1)$.

Let $d_T^{(2)}$ be the threshold margin in Assumption
\ref{ass:app_local_gap}.  Set $d_T=\min\{d_T^{(1)},d_T^{(2)}\}$.  Since
\[
 \frac{\mathcal Z_T+\mathcal R_T}{d_T}
 \leq
 \frac{\mathcal Z_T+\mathcal R_T}{d_T^{(1)}}
 +
 \frac{\mathcal Z_T+\mathcal R_T}{d_T^{(2)}}=o_p(1),
\]
the population ordering and threshold inequalities above verify both parts of
Assumption~\ref{ass:app_dom}.  Assumption~\ref{ass:app_exhaust} is imposed,
so Theorem~\ref{thm:app_exact} yields the result.
\end{proof}

\begin{proof}[Proof of Theorem~\ref{thm:app_global}]
Assumption~\ref{ass:app_global_gap} gives, uniformly over $S\in\mathcal S_0$,
\[
 \max_{j\in S_0\setminus S}\Lambda_{T,j}^{\ast}(S)
 \geq
 \max_{\substack{m\notin S_0\\m\notin S}}
 \Lambda_{T,m}^{\ast}(S)+\sqrt T\,b_T.
\]
Set $d_T^{(1)}=\sqrt T\,b_T$, let $d_T^{(2)}$ be the threshold margin in the
same assumption, and put $d_T=\min\{d_T^{(1)},d_T^{(2)}\}$.  The two rate
conditions in Assumption~\ref{ass:app_global_gap} imply
$(\mathcal Z_T+\mathcal R_T)/d_T=o_p(1)$, by the same inequality used in the
proof of Theorem~\ref{thm:app_local}.  The population gap is at least $d_T$
and the strongest remaining signal exceeds $c_p+d_T$.  Hence
Assumption~\ref{ass:app_dom} holds.  Theorem~\ref{thm:app_exact}, together
with Assumption~\ref{ass:app_exhaust}, proves exact recovery.
\end{proof}

\begin{proof}[Proof of Theorem~\ref{thm:app_oracle}]
Assumption~\ref{ass:app_oracle_regularity} and the score central limit theorem
give, by Slutsky's theorem,
\[
 \sqrt T(\widehat{\boldsymbol\beta}_{\mathrm{or}}-
 \boldsymbol\beta_{0,S_0})
 \Rightarrow
 N(\boldsymbol0,\boldsymbol J_0^{-1}\boldsymbol\Omega_0\boldsymbol J_0^{-1}).
\]
By the embedding used in the theorem, on $\{\widehat S=S_0\}$ the
$S_0$ subvector of the post-BMT refit equals the oracle estimator. Therefore,
for every $\varepsilon>0$,
\[
 \mathbb P\!\left(
 \sqrt T\,\|\widehat{\boldsymbol\beta}_{\widehat S,S_0}-
 \widehat{\boldsymbol\beta}_{\mathrm{or}}\|>\varepsilon
 \right)
 \leq \mathbb P(\widehat S\neq S_0)\longrightarrow0.
\]
Thus
$\sqrt T(\widehat{\boldsymbol\beta}_{\widehat S,S_0}-
\widehat{\boldsymbol\beta}_{\mathrm{or}})=o_p(1)$, and a second application
of Slutsky's theorem gives the stated oracle limit.
\end{proof}

\clearpage\newpage

\section{Comparison with Binary LASSO}
\label{app:binary_lasso}

\setcounter{equation}{0} \renewcommand{\theequation}{C.\arabic{equation}}\renewcommand{\theHequation}{C.\arabic{equation}}
\setcounter{theorem}{0} \renewcommand{\thetheorem}{C.\arabic{theorem}}\renewcommand{\theHtheorem}{C.\arabic{theorem}}
\setcounter{lemma}{0} \renewcommand{\thelemma}{C.\arabic{lemma}}\renewcommand{\theHlemma}{C.\arabic{lemma}}
\setcounter{remark}{0} \renewcommand{\theremark}{C.\arabic{remark}}\renewcommand{\theHremark}{C.\arabic{remark}}

This appendix compares the stagewise condition used by BMT with the
irrepresentability condition used in model-selection results for binary
LASSO. The comparison is confined to correctly specified logit and probit
models. The observations are independent across $t$. The first example gives a design in which BMT recovers the true
model while binary LASSO is not sign consistent. The second gives the
reverse comparison. We then give an additional condition under which
generalised irrepresentability implies the weaker, one-signal-at-a-time
dominance condition that is sufficient for BMT. 

For clarity, the examples have no always-in covariates, so that
$S^{(0)}=\varnothing$.  With always-in controls, the same calculations apply
after replacing the covariates and covariance matrices below by their
population residual counterparts conditional on $S^{(0)}$.  The ordinary
linear projection residual appears because the calculations are local around
a zero binary index.  At that point, the GLM curvature weight is constant, so
this residual is the leading-order version of the GLM-weighted residual used
in Section~\ref{sec:primitive}.

Let
\begin{equation}
R_T(\boldsymbol\beta)
=-T^{-1}\sum_{t=1}^T
l_t(\boldsymbol x_t'\boldsymbol\beta,y_t)
\label{eq:app_lasso_risk}
\end{equation}
be the average negative Bernoulli log-likelihood.  The binary LASSO estimator
is
\begin{equation}
\widehat{\boldsymbol\beta}^{L}
\in\operatorname*{arg\,min}_{\boldsymbol\beta}
\left\{R_T(\boldsymbol\beta)+\lambda_T\|\boldsymbol\beta\|_1\right\},
\label{eq:app_lasso_def}
\end{equation}
where $\left\Vert \cdot \right\Vert _{1}$ is the $L_{1}$ norm. Under correct specification, put
\begin{equation}
\boldsymbol Q(\boldsymbol\beta_0)
=
\mathbb E\!\left[
\omega(\boldsymbol x_t'\boldsymbol\beta_0)
\boldsymbol x_t\boldsymbol x_t'
\right],
\qquad
\omega(v)=\frac{\{G'(v)\}^2}{G(v)\{1-G(v)\}}.
\label{eq:app_lasso_information}
\end{equation}
If $\boldsymbol\sigma_0=\operatorname{sign}(\boldsymbol\beta_{0,S_0})$, the
generalised irrepresentability condition is
\begin{equation}
\left\|
\boldsymbol Q_{S_0^cS_0}(\boldsymbol\beta_0)
\boldsymbol Q_{S_0S_0}(\boldsymbol\beta_0)^{-1}
\boldsymbol\sigma_0
\right\|_\infty<1,
\label{eq:app_lasso_gi}
\end{equation}
where $\left\Vert \cdot \right\Vert _{\infty }$\ is the $L_{\infty }$ norm. 
\subsection{Two preliminary results}
\label{app:binary_lasso_prelim}

The following fixed-dimensional Karush–Kuhn–Tucker (KKT) result is used in both examples.

\begin{lemma}
\label{lem:app_lasso_kkt}
Suppose $p$ is fixed, the binary model is correctly specified, and the
population risk is uniquely minimised at $\boldsymbol\beta_0$.
Suppose $\boldsymbol Q_{S_0S_0}(\boldsymbol\beta_0)$ is positive definite, the
sample Hessian is uniformly consistent in a neighbourhood of
$\boldsymbol\beta_0$, and the score is $O_p(T^{-1/2})$. Suppose also that
the objective function in \eqref{eq:app_lasso_def} is convex and has a unique
minimiser with probability tending to one. Let
\(\lambda_T\rightarrow0, \sqrt T\lambda_T\rightarrow\infty,\)
and suppose that $\min_{j\in S_0}|\beta_{0,j}|$ is bounded away from zero.

If, for some $\eta>0$,
\begin{equation}
\left\|
\boldsymbol Q_{S_0^cS_0}(\boldsymbol\beta_0)
\boldsymbol Q_{S_0S_0}(\boldsymbol\beta_0)^{-1}
\boldsymbol\sigma_0
\right\|_\infty
\leq1-\eta,
\label{eq:app_lasso_gi_strict}
\end{equation}
then
\[
\mathbb P\!\left\{
\operatorname{sign}(\widehat{\boldsymbol\beta}^{L})
=
\operatorname{sign}(\boldsymbol\beta_0)
\right\}\rightarrow1.
\]
If, for some $m\notin S_0$ and some $\eta>0$,
\begin{equation}
\left|
\boldsymbol Q_{mS_0}(\boldsymbol\beta_0)
\boldsymbol Q_{S_0S_0}(\boldsymbol\beta_0)^{-1}
\boldsymbol\sigma_0
\right|
\geq1+\eta,
\label{eq:app_lasso_gi_failure}
\end{equation}
then
\[
\mathbb P\!\left\{
\operatorname{sign}(\widehat{\boldsymbol\beta}^{L})
=
\operatorname{sign}(\boldsymbol\beta_0)
\right\}\rightarrow0.
\]
\end{lemma}

\begin{proof}
Let $\widetilde{\boldsymbol\beta}_{S_0}$ minimise the penalised criterion
subject to $\boldsymbol\beta_{S_0^c}=\boldsymbol0$. Since $\lambda_T\to0$,
standard fixed-dimensional convex $M$-estimation results give
$\widetilde{\boldsymbol\beta}_{S_0}-\boldsymbol\beta_{0,S_0}=o_p(1)$.
The beta-min condition therefore implies
$\operatorname{sign}(\widetilde{\boldsymbol\beta}_{S_0})=\boldsymbol\sigma_0$
with probability tending to one. On this event, the active KKT equations and
a Taylor expansion give
\[
\widetilde{\boldsymbol\beta}_{S_0}-\boldsymbol\beta_{0,S_0}
=
-\boldsymbol Q_{S_0S_0}^{-1}
\{\nabla_{S_0}R_T(\boldsymbol\beta_0)+\lambda_T\boldsymbol\sigma_0\}
+o_p(\lambda_T),
\]
because $T^{-1/2}=o(\lambda_T)$. A second Taylor expansion then yields
\[
\lambda_T^{-1}\nabla_{S_0^c}R_T
(\widetilde{\boldsymbol\beta}_{S_0},\boldsymbol0)
=
-\boldsymbol Q_{S_0^cS_0}\boldsymbol Q_{S_0S_0}^{-1}\boldsymbol\sigma_0
+o_p(1).
\]
Under \eqref{eq:app_lasso_gi_strict}, the inactive KKT inequalities hold
strictly with probability tending to one. The restricted solution is then a
global minimiser of the full convex problem, and uniqueness identifies it
with $\widehat{\boldsymbol\beta}^{L}$; its signs are correct.

For the converse, standard fixed-dimensional penalised $M$-estimation and
$\lambda_T\to0$ give
$\widehat{\boldsymbol\beta}^{L}-\boldsymbol\beta_0=o_p(1)$. On the event that the
full estimator has the correct signed support, its active KKT equations therefore
give the same expansion, while the inactive
KKT inequality for coordinate $m$ requires
\[
\left|
\boldsymbol Q_{mS_0}\boldsymbol Q_{S_0S_0}^{-1}\boldsymbol\sigma_0
+o_p(1)
\right|\leq1.
\]
This contradicts \eqref{eq:app_lasso_gi_failure} with probability tending to
one. Hence the probability of correct signed support tends to zero.
\end{proof}

We also use a local expansion of the population Wald noncentrality.  Let
$z_t$ be a fixed linear combination of the covariates and consider
\begin{equation}
\mathbb P(y_t=1\mid\boldsymbol x_t)=G(bz_t),
\qquad b\rightarrow0.
\label{eq:app_lasso_local_model}
\end{equation}
For a conditioning set $S$ and $j\notin S$, let $r_{j,t}(S)$ denote the
residual from the population least-squares projection of $x_{j,t}$ on
$\boldsymbol x_{S,t}$.  Define
\begin{equation}
\omega_0
=
\frac{\{G'(0)\}^2}{G(0)\{1-G(0)\}}.
\label{eq:app_lasso_omega0}
\end{equation}

\begin{lemma}
\label{lem:app_lasso_local_wald}
Suppose the covariance matrix of the covariates in the augmented submodel is
nonsingular and the moments required for differentiation under the expectation
are finite.  For every fixed pair $(S,j)$,
\begin{equation}
\mathrm{NC}_j^*(S)
=
\sqrt T\,|b|\sqrt{\omega_0}
\frac{
\left|\mathbb E\{r_{j,t}(S)z_t\}\right|
}{
\{\mathbb E[r_{j,t}(S)^2]\}^{1/2}
}
+O(\sqrt T\,b^2),
\qquad b\rightarrow0.
\label{eq:app_lasso_local_wald}
\end{equation}
The remainder is uniform over any fixed finite collection of pairs $(S,j)$.
\end{lemma}

\begin{proof}
Let $\boldsymbol q_{j,t}(S)=(\boldsymbol x_{S,t}',x_{j,t})'$ and
$\boldsymbol\vartheta=(\boldsymbol\gamma',\theta)'$. At $b=0$ the pseudo-true
parameter is zero and the negative Jacobian of the population score is
$\omega_0\mathbb E[\boldsymbol q_{j,t}(S)\boldsymbol q_{j,t}(S)']$.
The implicit function theorem and the Schur complement for the candidate
coordinate give
\[
 \theta_j^{\ast}(S)
 =b\frac{\mathbb E\{r_{j,t}(S)z_t\}}{\mathbb E[r_{j,t}(S)^2]}+O(b^2),
 \qquad
 V_{\theta\theta,j}^{\ast}(S)
 =\frac{1}{\omega_0\mathbb E[r_{j,t}(S)^2]}+O(b).
\]
Substitution into \eqref{eq:nc_def_prim} proves
\eqref{eq:app_lasso_local_wald}; uniformity over a finite collection follows
from continuity and uniformly positive residual variances.
\end{proof}

\subsection{A design in which BMT succeeds}
\label{app:binary_lasso_bmt_success}

Let
\begin{equation}
\boldsymbol x_t=(x_{1,t},x_{2,t},x_{3,t})'
\sim N(\boldsymbol0,\boldsymbol\Sigma_\rho),
\qquad
\boldsymbol\Sigma_\rho
=
\begin{pmatrix}
1&0&\rho\\
0&1&\rho\\
\rho&\rho&1
\end{pmatrix},
\label{eq:app_lasso_proxy_covariance}
\end{equation}
where $0<\rho<1/\sqrt2$.  Consider
\begin{equation}
\mathbb P(y_t=1\mid\boldsymbol x_t)
=
G\{b(x_{1,t}+\alpha x_{2,t})\},
\qquad 0<\alpha<1.
\label{eq:app_lasso_proxy_dgp}
\end{equation}
The true support is $S_0=\{1,2\}$ and $x_{3,t}$ is inactive.

\begin{theorem}
\label{thm:app_lasso_bmt_success}
Suppose
\begin{equation}
\frac12<\rho<
\min\left\{\frac{1}{1+\alpha},\frac{1}{\sqrt2}\right\}.
\label{eq:app_lasso_proxy_wedge}
\end{equation}
Suppose also that Assumptions~\ref{ass:A1}--\ref{ass:A3} and
\ref{ass:A5} hold, $k_{\max}\geq2$, and $c_T=o(\sqrt T)$.  There exists
$\bar b>0$ such that, for every fixed $b\in(0,\bar b)$,
\(\mathbb P(\widehat S=S_0)\rightarrow1.\)
For binary LASSO, if
$\lambda_T\rightarrow0$ and $\sqrt T\lambda_T\rightarrow\infty$, then
\[
\mathbb P\!\left\{
\operatorname{sign}(\widehat{\boldsymbol\beta}^{L})
=
\operatorname{sign}\{(b,\alpha b,0)'\}
\right\}\rightarrow0.
\]
\end{theorem}

\begin{proof}
Put $z_t=x_{1,t}+\alpha x_{2,t}$.  At the initial stage,
\[
\mathbb E(x_{1,t}z_t)=1,
\qquad
\mathbb E(x_{2,t}z_t)=\alpha,
\qquad
\mathbb E(x_{3,t}z_t)=\rho(1+\alpha).
\]
All three covariates have unit variance.  The first inequality in
\eqref{eq:app_lasso_proxy_wedge} is not used here; the restriction
$\rho<1/(1+\alpha)$ implies that $x_{1,t}$ has the largest leading population
Wald statistic.  By Lemma~\ref{lem:app_lasso_local_wald}, this ordering is
strict for all sufficiently small fixed $b>0$.

After $x_{1,t}$ has been selected,
\(r_{2,t}(\{1\})=x_{2,t}, \qquad r_{3,t}(\{1\})=x_{3,t}-\rho x_{1,t}.\)
Consequently,
\[
\frac{|\mathbb E\{r_{2,t}(\{1\})z_t\}|}
{\{\mathbb E[r_{2,t}(\{1\})^2]\}^{1/2}}
=\alpha,
\]
whereas
\[
\frac{|\mathbb E\{r_{3,t}(\{1\})z_t\}|}
{\{\mathbb E[r_{3,t}(\{1\})^2]\}^{1/2}}
=
\frac{\rho\alpha}{\sqrt{1-\rho^2}}.
\]
The signal $x_{2,t}$ therefore has the larger population Wald statistic when
$\rho<1/\sqrt2$.  Once $x_{1,t}$ and $x_{2,t}$ have been selected, correct
specification gives $\mathrm{NC}_3^*(\{1,2\})=0$.

The two strict population gaps are proportional to $\sqrt T$.  Since $p$ is fixed, Assumption~\ref{ass:A3}.4 gives
$\mathcal E_T=O_p(\sqrt{\log T})=o_p(\sqrt T)$, and
$c_T=o(\sqrt T)$.  The sample ordering agrees with the population ordering
with probability tending to one, both signals cross the threshold at their
respective stages, and Assumption~\ref{ass:A5} gives stopping after the second
stage.  This proves the BMT statement.

For LASSO, joint normality gives the projection representation
\(x_{3,t}=\rho x_{1,t}+\rho x_{2,t}+u_{3,t},\)
where $u_{3,t}$ is independent of $(x_{1,t},x_{2,t})$.  The information
weight in \eqref{eq:app_lasso_information} is a function of
$(x_{1,t},x_{2,t})$.  Hence, with $S_0=\{1,2\}$,
\(\boldsymbol Q_{3S_0}(\boldsymbol\beta_0) =(\rho,\rho)\boldsymbol Q_{S_0S_0}(\boldsymbol\beta_0).\)
Both active coefficients are positive, so
\[
\boldsymbol Q_{3S_0}(\boldsymbol\beta_0)
\boldsymbol Q_{S_0S_0}(\boldsymbol\beta_0)^{-1}
\boldsymbol\sigma_0
=2\rho>1.
\]
The second conclusion follows from Lemma~\ref{lem:app_lasso_kkt}.
\end{proof}

The example uses the sequential nature of BMT.  At the first stage, the weak
signal $x_{2,t}$ need not dominate the inactive proxy.  It is enough that the
stronger signal $x_{1,t}$ is selected first and that $x_{2,t}$ dominates the
proxy after conditioning on $x_{1,t}$.

\subsection{A design in which LASSO succeeds}
\label{app:binary_lasso_lasso_success}

Let
\begin{equation}
\boldsymbol x_t=(x_{1,t},x_{2,t},x_{3,t},x_{4,t})'
\sim N(\boldsymbol0,\boldsymbol\Sigma_R),
\qquad
\boldsymbol\Sigma_R
=
\begin{pmatrix}
1&0&0&3/5\\
0&1&0&3/5\\
0&0&1&-1/4\\
3/5&3/5&-1/4&1
\end{pmatrix}.
\label{eq:app_lasso_reverse_covariance}
\end{equation}
The matrix is positive definite because the Schur complement of the upper
left $3\times3$ block is
\[
1-\left\{\left(\frac35\right)^2+
\left(\frac35\right)^2+\left(\frac14\right)^2\right\}
=\frac{87}{400}>0.
\]
Consider
\begin{equation}
\mathbb P(y_t=1\mid\boldsymbol x_t)
=
G\!\left\{\tau\left(x_{1,t}+x_{2,t}+\frac15x_{3,t}\right)\right\},
\qquad \tau>0.
\label{eq:app_lasso_reverse_dgp}
\end{equation}
Here $S_0=\{1,2,3\}$ and $x_{4,t}$ is inactive.

\begin{theorem}
\label{thm:app_lasso_lasso_success}
Suppose Assumptions~\ref{ass:A1}--\ref{ass:A3} and
\ref{ass:A5} hold, $k_{\max}\geq3$, and $c_T=o(\sqrt T)$.  There exists
$\bar\tau>0$ such that, for every fixed $\tau\in(0,\bar\tau)$,
\(\mathbb P(\widehat S=S_0)\rightarrow0.\)
If binary LASSO is computed with
$\lambda_T\rightarrow0$ and $\sqrt T\lambda_T\rightarrow\infty$, then
\[
\mathbb P\!\left\{
\operatorname{sign}(\widehat{\boldsymbol\beta}^{L})
=
\operatorname{sign}\{\tau(1,1,1/5,0)'\}
\right\}\rightarrow1.
\]
\end{theorem}

\begin{proof}
Let
$z_t=x_{1,t}+x_{2,t}+x_{3,t}/5$.  At the initial stage,
\[
\mathbb E(x_{1,t}z_t)=1,
\qquad
\mathbb E(x_{2,t}z_t)=1,
\qquad
\mathbb E(x_{3,t}z_t)=\frac15,
\]
while
\(\mathbb E(x_{4,t}z_t) =\frac35+\frac35-\frac14\frac15 =\frac{23}{20}>1.\)
Lemma~\ref{lem:app_lasso_local_wald} implies that, for all sufficiently small
fixed $\tau>0$,
\(
\mathrm{NC}_4^*(\varnothing)
>
\max_{j\in S_0}\mathrm{NC}_j^*(\varnothing).
\)
The gap is proportional to $\sqrt T$.  Assumption~\ref{ass:A3}.4 and
$c_T=o(\sqrt T)$ therefore imply that BMT selects $x_{4,t}$ at the first stage
with probability tending to one.  Since the BMT procedure does not delete a
selected covariate, its final selected set cannot equal $S_0$.

For LASSO, joint normality gives
\(x_{4,t} =\frac35x_{1,t}+\frac35x_{2,t}-\frac14x_{3,t}+u_{4,t},\)
where $u_{4,t}$ is independent of $\boldsymbol x_{S_0,t}$.  The information
weight is measurable with respect to $\boldsymbol x_{S_0,t}$, and hence
\(\boldsymbol Q_{4S_0}(\boldsymbol\beta_0) = (3/5,3/5,-1/4)\boldsymbol Q_{S_0S_0}(\boldsymbol\beta_0).\)
All active coefficients are positive.  It follows that
\[
\boldsymbol Q_{4S_0}(\boldsymbol\beta_0)
\boldsymbol Q_{S_0S_0}(\boldsymbol\beta_0)^{-1}
\boldsymbol\sigma_0
=
\frac35+\frac35-\frac14
=\frac{19}{20}<1.
\]
The first conclusion of Lemma~\ref{lem:app_lasso_kkt} proves sign consistency
of binary LASSO.
\end{proof}

The two examples show that generalised irrepresentability and BMT stagewise
dominance are not nested without additional restrictions.  The difference is
that generalised irrepresentability depends on the signed projection of an
inactive covariate on the full active set, whereas the BMT statistic depends
on the actual coefficient magnitudes and is recomputed after each selection.

\subsection{Generalised irrepresentability and stagewise dominance}
\label{app:binary_lasso_relationship}

We now give a sufficient relation between the two conditions.  Let
$\boldsymbol\Sigma=\mathbb E(\boldsymbol x_t\boldsymbol x_t')$ be positive
definite, with $\Sigma_{jj}=1$, and consider the local binary model
\begin{equation}
\mathbb P(y_t=1\mid\boldsymbol x_t)
=
G\{\tau_T\boldsymbol x_{S_0,t}'\boldsymbol b_{S_0}\},
\qquad
\tau_T>0,\quad \tau_T\rightarrow0,
\label{eq:app_lasso_general_local}
\end{equation}
where $b_j\neq0$ for every $j\in S_0$.  Put
$\boldsymbol\sigma_0=\operatorname{sign}(\boldsymbol b_{S_0})$.  For each
$m\notin S_0$, define the coefficient vector in the population projection of
$x_{m,t}$ on the full active set by
\begin{equation}
\boldsymbol\gamma_m
=
\boldsymbol\Sigma_{S_0S_0}^{-1}\boldsymbol\Sigma_{S_0m}.
\label{eq:app_lasso_gamma}
\end{equation}
At $\tau_T=0$, the Fisher information is $\omega_0\boldsymbol\Sigma$, so the
leading form of generalised irrepresentability is
\begin{equation}
\max_{m\notin S_0}
|\boldsymbol\gamma_m'\boldsymbol\sigma_0|
\leq1-\eta
\qquad\text{for some }\eta\in(0,1).
\label{eq:app_lasso_cov_gi}
\end{equation}
Under uniform moment and continuity conditions for the Fisher-information
blocks, a strict margin in \eqref{eq:app_lasso_cov_gi} also implies the
corresponding Fisher-information condition, uniformly over inactive
candidates, for all sufficiently small $\tau_T$.

Let $A\subsetneq S_0$ be the signals already selected and put
$R=S_0\setminus A$.  For $U,V\subseteq A^c$, define the conditional
covariance block
\begin{equation}
\boldsymbol\Sigma^A_{UV}
=
\boldsymbol\Sigma_{UV}
-
\boldsymbol\Sigma_{UA}
\boldsymbol\Sigma_{AA}^{-1}
\boldsymbol\Sigma_{AV},
\label{eq:app_lasso_schur}
\end{equation}
with $\boldsymbol\Sigma^\varnothing=\boldsymbol\Sigma$.  Write
\(
d_j(A)=\Sigma^A_{jj},
\qquad
\boldsymbol D(A)=\operatorname{diag}\{d_j(A):j\in R\}.
\)
For $j\notin A$, define the leading local population score
\begin{equation}
\mathcal L_j(A)
=
\frac{
|\boldsymbol\Sigma^A_{jR}\boldsymbol b_R|
}{\sqrt{d_j(A)}}.
\label{eq:app_lasso_leading_score}
\end{equation}
By Lemma~\ref{lem:app_lasso_local_wald}, this is the term multiplying
$\sqrt T|\tau_T|\sqrt{\omega_0}$ in the local expansion of
$\mathrm{NC}_j^*(A)$.

Define
\begin{equation}
\begin{split}
\boldsymbol h(A)
&=\boldsymbol D(A)^{-1/2}
\boldsymbol\Sigma^A_{RR}\boldsymbol b_R,\\
M(A)&=\|\boldsymbol h(A)\|_\infty,
\qquad
\boldsymbol u(A)=\frac{\boldsymbol h(A)}{M(A)},\\
\boldsymbol a_m(A)
&=
\frac{
\boldsymbol D(A)^{1/2}\boldsymbol\gamma_{m,R}
}{\sqrt{d_m(A)}}.
\end{split}
\label{eq:app_lasso_stage_objects}
\end{equation}
Positive definiteness and $\boldsymbol b_R\neq\boldsymbol0$ imply
$M(A)>0$.

\begin{lemma}
\label{lem:app_lasso_stage_representation}
For every $A\subsetneq S_0$ and $m\notin S_0$,
\begin{equation}
\max_{j\in R}\mathcal L_j(A)=M(A),
\qquad
\frac{\mathcal L_m(A)}{M(A)}
=
|\boldsymbol a_m(A)'\boldsymbol u(A)|.
\label{eq:app_lasso_stage_representation}
\end{equation}
\end{lemma}

\begin{proof}
For $j\in R$, the $j$th component of $\boldsymbol h(A)$ is
$\boldsymbol\Sigma^A_{jR}\boldsymbol b_R/\sqrt{d_j(A)}$, which proves the
first equality.  The Frisch--Waugh--Lovell identity for the projection in
\eqref{eq:app_lasso_gamma} gives
\(\boldsymbol\gamma_{m,R} = (\boldsymbol\Sigma^A_{RR})^{-1} \boldsymbol\Sigma^A_{Rm}.\)
Therefore,
\[
\frac{\boldsymbol\Sigma^A_{mR}\boldsymbol b_R}{\sqrt{d_m(A)}}
=
\boldsymbol a_m(A)'\boldsymbol h(A).
\]
Taking absolute values and dividing by $M(A)$ proves the second equality.
\end{proof}

Only an increase in the inactive loading after conditioning can use up the
margin in \eqref{eq:app_lasso_cov_gi}.  Define the one-sided directional
distortion
\begin{equation}
\Delta_{\mathrm{dir}}
=
\sup_{A\subsetneq S_0}
\sup_{m\notin S_0}
\left[
|\boldsymbol a_m(A)'\boldsymbol u(A)|
-|\boldsymbol\gamma_m'\boldsymbol\sigma_0|
\right]_+,
\label{eq:app_lasso_directional_distortion}
\end{equation}
where $[v]_+=\max(v,0)$.  This condition is coefficient-specific: it uses the
actual stagewise score direction $\boldsymbol u(A)$ and does not require
independence of the active covariates or a common sign pattern for the
inactive projection coefficients.

\begin{theorem}
\label{thm:app_lasso_gi_bmt}
Suppose \eqref{eq:app_lasso_cov_gi} holds and
\begin{equation}
\Delta_{\mathrm{dir}}\leq\delta<\eta.
\label{eq:app_lasso_directional_stability}
\end{equation}
Let $\varepsilon=\eta-\delta$.  Then, for every
$A\subsetneq S_0$,
\begin{equation}
\max_{m\notin S_0}\mathcal L_m(A)
\leq
(1-\varepsilon)
\max_{j\in S_0\setminus A}\mathcal L_j(A).
\label{eq:app_lasso_relaxed_dominance}
\end{equation}

Suppose, in addition, that the local expansion
\eqref{eq:app_lasso_local_wald} holds uniformly over
$A\subsetneq S_0$ and $j\notin A$, that $k_{\max}\geq k$, and that
Assumption~\ref{ass:A5} holds.  If
\begin{equation}
\frac{c_T+\mathcal E_T}{\sqrt T|\tau_T|}
=o_p(1),
\label{eq:app_lasso_local_rate}
\end{equation}
then BMT recovers the true support:
\(\mathbb P(\widehat S=S_0)\rightarrow1.\)
If the number of candidates is fixed, the same conclusion holds for every
sufficiently small fixed $\tau>0$ when
$c_T+\mathcal E_T=o_p(\sqrt T)$.
\end{theorem}

\begin{proof}
By Lemma~\ref{lem:app_lasso_stage_representation},
\(\frac{\mathcal L_m(A)}{M(A)} = |\boldsymbol a_m(A)'\boldsymbol u(A)|.\)
Equations \eqref{eq:app_lasso_cov_gi} and
\eqref{eq:app_lasso_directional_stability} therefore give
\[
|\boldsymbol a_m(A)'\boldsymbol u(A)|
\leq
|\boldsymbol\gamma_m'\boldsymbol\sigma_0|+\delta
\leq
1-\eta+\delta
=
1-\varepsilon.
\]
This proves \eqref{eq:app_lasso_relaxed_dominance}.

There are finitely many sets $A\subsetneq S_0$.  Since $M(A)>0$ for each of
them,
\(\underline M = \min_{A\subsetneq S_0}M(A)>0.\)
Uniformly over these sets, Lemma~\ref{lem:app_lasso_local_wald} and
\eqref{eq:app_lasso_relaxed_dominance} imply that the strongest remaining
signal exceeds every inactive candidate by a population Wald gap of at least
\(\frac12\varepsilon\sqrt{\omega_0}\,\underline M \sqrt T|\tau_T|\)
for all sufficiently large $T$.  The factor $1/2$ absorbs the uniform
$O(\sqrt T\tau_T^2)$ remainder.  Condition
\eqref{eq:app_lasso_local_rate} transfers this ordering to the sample Wald
statistics and makes the strongest remaining signal cross $c_T$ at every
stage.  The induction used in the proof of
Theorem~\ref{thm:fixedk_recovery} then shows that BMT selects a signal at each
of its first $k$ stages.  Once $S_0$ has been selected, correct specification
gives $\mathrm{NC}_m^*(S_0)=0$ for all $m\notin S_0$, and
Assumption~\ref{ass:A5} gives stopping.  This proves exact recovery.

When the number of candidates is fixed, the strict inequalities at the local
limit persist for all sufficiently small fixed $\tau>0$ by continuity.  Their
Wald gaps are then of order $\sqrt T$, and the same argument applies.
\end{proof}

The directional condition in \eqref{eq:app_lasso_directional_stability} is
one-sided and permits cancellation.  The following bound gives a more
primitive sufficient condition.  Partition
$\boldsymbol\gamma_m=(\boldsymbol\gamma_{m,A}',
\boldsymbol\gamma_{m,R}')'$ and
$\boldsymbol\sigma_0=(\boldsymbol\sigma_A',\boldsymbol\sigma_R')'$.  Since
\[
\boldsymbol a_m(A)'\boldsymbol u(A)
=
\boldsymbol\gamma_{m,R}'
\frac{\boldsymbol D(A)^{1/2}\boldsymbol u(A)}{\sqrt{d_m(A)}},
\]
the reverse triangle inequality gives
\begin{equation}
\begin{split}
&\left|
|\boldsymbol a_m(A)'\boldsymbol u(A)|
-|\boldsymbol\gamma_m'\boldsymbol\sigma_0|
\right|\\
&\quad\leq
|\boldsymbol\gamma_{m,A}'\boldsymbol\sigma_A|
+
\|\boldsymbol\gamma_{m,R}\|_1
\left\|
\frac{\boldsymbol D(A)^{1/2}\boldsymbol u(A)}{\sqrt{d_m(A)}}
-\boldsymbol\sigma_R
\right\|_\infty.
\end{split}
\label{eq:app_lasso_directional_bound}
\end{equation}
Thus \eqref{eq:app_lasso_directional_stability} follows if the right-hand
side of \eqref{eq:app_lasso_directional_bound} is uniformly smaller than the
irrepresentability margin $\eta$.  This allows correlated active covariates
and mixed-sign inactive loadings.  It requires only that deleting the already
selected coordinates and replacing the sign direction by the actual
stagewise score direction do not increase the inactive loading enough to
exhaust the strict margin in \eqref{eq:app_lasso_cov_gi}.

Finally, Assumption~\ref{ass:A4} is stated using the minimum population
noncentrality among all remaining signals.  The weaker maximum condition in
\eqref{eq:app_lasso_relaxed_dominance} is sufficient for the BMT rule because
only one covariate is added at each stage.  If one wishes to verify
Assumption~\ref{ass:A4} exactly as stated, define
\begin{equation}
q_*
=
\inf_{A\subsetneq S_0}
\frac{
\min_{j\in S_0\setminus A}\mathcal L_j(A)
}{
\max_{j\in S_0\setminus A}\mathcal L_j(A)
}.
\label{eq:app_lasso_active_balance}
\end{equation}
If $q_*>1-\varepsilon$, then every remaining signal dominates every inactive
candidate in the local limit.  Together with
\eqref{eq:app_lasso_local_rate}, this gives the dominance and threshold parts
of Assumption~\ref{ass:A4} with a margin proportional to
$\sqrt T|\tau_T|$.

The examples and Theorem~\ref{thm:app_lasso_gi_bmt} give a qualified
comparison.  Neither generalised irrepresentability nor BMT stagewise
dominance implies the other without further restrictions.  A strict
generalised irrepresentability margin does imply the one-signal-at-a-time BMT
dominance condition when that margin is stable along the actual stagewise
score directions.

\section{Primitive Sufficient Conditions for the Uniform Wald Approximation}
\label{app:primitive_wald}
\setlength{\abovedisplayskip}{2pt plus 1pt minus 1pt}
\setlength{\belowdisplayskip}{2pt plus 1pt minus 1pt}
\setlength{\abovedisplayshortskip}{1pt plus 1pt minus 1pt}
\setlength{\belowdisplayshortskip}{2pt plus 1pt minus 1pt}
\setlength{\jot}{2pt}

\setcounter{equation}{0} \renewcommand{\theequation}{D.\arabic{equation}}\renewcommand{\theHequation}{D.\arabic{equation}}
\setcounter{theorem}{0} \renewcommand{\thetheorem}{D.\arabic{theorem}}\renewcommand{\theHtheorem}{D.\arabic{theorem}}
\setcounter{lemma}{0} \renewcommand{\thelemma}{D.\arabic{lemma}}\renewcommand{\theHlemma}{D.\arabic{lemma}}
\setcounter{remark}{0} \renewcommand{\theremark}{D.\arabic{remark}}\renewcommand{\theHremark}{D.\arabic{remark}}

This appendix gives primitive sufficient conditions for Assumption~\ref{ass:A3}.4.
The variance target remains the long-run sandwich matrix in A3.3. The
verification covers two cases in which simpler estimators are uniformly
consistent for that target: observed-information standard errors under the
information equality and one-period sandwich standard errors when the
stagewise scores are serially uncorrelated. A kernel HAC estimator requires a
separate, bandwidth-dependent uniform long-run-variance bound and remains
covered by the high-level formulation of A3.4.

Put $N_T=|\mathcal A_T|$. Because $s_0$ and $k_{\max}$ are fixed, every
candidate parameter has dimension at most $d_0=s_0+k_{\max}+1$, while
\begin{equation}
N_T\le p\sum_{\ell=0}^{k_{\max}}\binom p\ell\le Cp^K,
\qquad K=k_{\max}+1.
\label{eq:primitive_cardinality}
\end{equation}
Thus, uniformity entails a union bound over polynomially many
fixed-dimensional likelihood problems.

For $m=(S,j)\in\mathcal A_T$, abbreviate
$\boldsymbol\vartheta_m^*=\boldsymbol\vartheta_j^*(S)$,
$\widehat{\boldsymbol\vartheta}_m=\widehat{\boldsymbol\vartheta}_j(S)$,
$\boldsymbol J_m^*=\boldsymbol J_j^*(S)$ and
$\boldsymbol V_m^*=\boldsymbol V_j^*(S)$. Write $\theta_m^*$ for the
candidate component of $\boldsymbol\vartheta_m^*$ and
$V_{\theta\theta,m}^*$ for the corresponding diagonal element of
$\boldsymbol V_m^*$. Let
$\widehat Q_m(\boldsymbol\vartheta)=T^{-1}\sum_{t=1}^T
l_{m,t}(\boldsymbol\vartheta)$ and
$Q_m(\boldsymbol\vartheta)=\mathbb E l_{m,t}(\boldsymbol\vartheta)$.
Write $\boldsymbol s_{m,t}(\boldsymbol\vartheta)$ for the score. Set
$\widehat{\boldsymbol J}_m(\boldsymbol\vartheta)
=-T^{-1}\sum_{t=1}^T\partial_{\boldsymbol\vartheta
\boldsymbol\vartheta'}^2l_{m,t}(\boldsymbol\vartheta)$ and
$\widehat{\boldsymbol\Omega}_m(\boldsymbol\vartheta)
=T^{-1}\sum_{t=1}^T\boldsymbol s_{m,t}(\boldsymbol\vartheta)
\boldsymbol s_{m,t}(\boldsymbol\vartheta)'$. Their expectations are denoted by
$\boldsymbol J_m(\boldsymbol\vartheta)$ and
$\boldsymbol\Omega_m(\boldsymbol\vartheta)$, respectively. For the
primitive verification, $\widehat{\boldsymbol\vartheta}_m$ is taken to be a
maximiser over $\mathcal B_m=\mathcal B_{j,S}$. The same argument applies to
a local maximiser known to lie in the uniform neighbourhood of
$\boldsymbol\vartheta_m^*$ used below. In the polynomial-moment result, the
sets $\mathcal B_m$ are also assumed to be uniformly regular compact
Lipschitz domains; compact rectangles of uniformly bounded diameter suffice.

Define
\begin{equation}
r_T^{\mathrm{ET}}=\sqrt{\ell_T},
\qquad
r_T^{\mathrm{PM}}=N_T^{1/q}=O(p^{K/q}),
\label{eq:primitive_rates}
\end{equation}
and let $r_T$ denote the applicable rate.

\begin{lemma}[Uniform empirical bounds]
\label{lem:primitive_empirical}
Suppose A1--A3.3 hold. Under A2-ET and $\ell_T^3/T\to0$,
\begin{align}
\max_{m\in\mathcal A_T}
\left\|T^{-1/2}\sum_{t=1}^T
\boldsymbol s_{m,t}(\boldsymbol\vartheta_m^*)\right\|
&=O_p(r_T^{\mathrm{ET}}),
\label{eq:primitive_score_et}\\
\max_{m\in\mathcal A_T}\sup_{\boldsymbol\vartheta\in\mathcal B_m}
\left|
\{\widehat Q_m(\boldsymbol\vartheta)-\widehat Q_m(\boldsymbol\vartheta_m^*)\}
-\{Q_m(\boldsymbol\vartheta)-Q_m(\boldsymbol\vartheta_m^*)\}
\right|
&=O_p\!\left(\frac{r_T^{\mathrm{ET}}}{\sqrt T}\right),
\label{eq:primitive_criterion_et}\\
\max_{m\in\mathcal A_T}\sup_{\boldsymbol\vartheta\in\mathcal B_m}
\left\|\widehat{\boldsymbol J}_m(\boldsymbol\vartheta)
-\boldsymbol J_m(\boldsymbol\vartheta)\right\|
&=O_p\!\left(\frac{r_T^{\mathrm{ET}}}{\sqrt T}\right),
\label{eq:primitive_hessian_et}\\
\max_{m\in\mathcal A_T}\sup_{\boldsymbol\vartheta\in\mathcal B_m}
\left\|\widehat{\boldsymbol\Omega}_m(\boldsymbol\vartheta)
-\boldsymbol\Omega_m(\boldsymbol\vartheta)\right\|
&=O_p\!\left(\frac{r_T^{\mathrm{ET}}}{\sqrt T}\right).
\label{eq:primitive_omega_et}
\end{align}
Under A2-PM, the same four statements hold with
$r_T^{\mathrm{ET}}$ replaced by $r_T^{\mathrm{PM}}$.
\end{lemma}

\begin{proof}
Every scalar coordinate of the arrays in
\eqref{eq:primitive_score_et}--\eqref{eq:primitive_omega_et} is a measurable
function of the underlying strongly mixing row and, therefore, inherits its
mixing coefficients. At the net points, all relevant coordinates have at
least sub-exponential tails. With geometric mixing, the least-favourable
exponents in Theorem~1 of \citet{merlevede2011bernstein} are
$\gamma_1=\gamma_2=1$ and hence
$\gamma=(\gamma_1^{-1}+\gamma_2^{-1})^{-1}=1/2$. Applied at the unnormalised level $y=M\sqrt{T\ell_T}$, that theorem
bounds each relevant scalar tail probability at $M\sqrt{\ell_T}$ by
$T\exp\{-c(T\ell_T)^{1/4}\}+2\exp(-cM^2\ell_T)$, after changing constants.
The last exponential term in the cited Bernstein bound is no larger than its
Gaussian term and is absorbed in the second term.

Cover each fixed-dimensional set $\mathcal B_m$ by a $T^{-3}$-net. The
number of net points is at most $CT^{3d_0}$, uniformly in $m$, so the total
number of scalar comparisons is at most $Cp^KT^{3d_0}$. For sufficiently
large fixed $M$, the union bound eliminates the Gaussian term. It also
eliminates the semiexponential term because
$(T\ell_T)^{1/4}/\ell_T=(T/\ell_T^3)^{1/4}\to\infty$. Interpolation off the
nets is negligible: the criterion increment is Lipschitz with envelope
$F_1$, the Hessian with $F_3$, and the score outer product with
$F_{\Omega,1}$. A marginal union bound gives maxima of order at most
$O_p(\ell_T^{3/2})$ for the $\psi_{2/3}$ envelopes, and multiplication by
$T^{-3}$ is $o_p(r_T^{\mathrm{ET}}/\sqrt T)$.

For the PM regime, the mixing summability condition in A1.1 is
stronger than the condition in the $q$th-moment inequality of
\citet{yokoyama1980moment}. Together with the $q+\delta$ moments in A2-PM, it
yields the uniform pointwise bound
\begin{equation}
\sup_{m,\boldsymbol\vartheta}
\mathbb E\left|
T^{-1/2}\sum_{t=1}^T
\{g_{m,t}(\boldsymbol\vartheta)
-\mathbb Eg_{m,t}(\boldsymbol\vartheta)\}
\right|^q\le C
\label{eq:primitive_rosenthal}
\end{equation}
for each relevant coordinate and its first parameter derivative; related
Rosenthal-type inequalities are given by
\citet{merlevede2013rosenthal}. Since $q>d_0$, the fixed-dimensional Sobolev
inequality on the uniformly regular sets $\mathcal B_m$, followed by Fubini's
theorem, extends \eqref{eq:primitive_rosenthal} to the supremum over
$\boldsymbol\vartheta$ within each model. Markov's inequality and the union
bound over $N_T$ models then give, for any $M>0$, a bound of order $CM^{-q}$
for the probability that any normalised empirical process exceeds
$MN_T^{1/q}$. This proves the PM statements.
\end{proof}

\begin{lemma}[Uniform consistency and Bahadur expansion]
\label{lem:primitive_bahadur}
If $r_T=o(\sqrt T)$, then
\begin{equation}
\max_{m\in\mathcal A_T}
\|\widehat{\boldsymbol\vartheta}_m-\boldsymbol\vartheta_m^*\|
=O_p\!\left(\frac{r_T}{\sqrt T}\right),
\label{eq:primitive_consistency}
\end{equation}
and
\begin{equation}
\max_{m\in\mathcal A_T}
\left\|
\sqrt T(\widehat{\boldsymbol\vartheta}_m-\boldsymbol\vartheta_m^*)
-(\boldsymbol J_m^*)^{-1}T^{-1/2}
\sum_{t=1}^T\boldsymbol s_{m,t}(\boldsymbol\vartheta_m^*)
\right\|
=O_p\!\left(\frac{r_T^2}{\sqrt T}\right)=o_p(r_T).
\label{eq:primitive_bahadur}
\end{equation}
\end{lemma}

\begin{proof}
Equation \eqref{eq:primitive_criterion_et} is $o_p(1)$ when
$r_T=o(\sqrt T)$. The uniform separation in A3.2 therefore gives
$\max_m\|\widehat{\boldsymbol\vartheta}_m-
\boldsymbol\vartheta_m^*\|=o_p(1)$ by the usual compact-parameter
$M$-estimation argument. The eigenvalue condition in A3.3 and the uniformly
integrable third-derivative envelope imply that the population negative
Hessian remains uniformly positive definite on a fixed neighbourhood of each
$\boldsymbol\vartheta_m^*$. Lemma~\ref{lem:primitive_empirical} transfers
this property to the sample Hessians with probability tending to one. The
estimators are then interior and satisfy the score equations.

Put
$\boldsymbol h_m=\widehat{\boldsymbol\vartheta}_m-
\boldsymbol\vartheta_m^*$. Using the integral form of the multivariate Taylor
expansion, set $\overline{\boldsymbol J}_m=\int_0^1
\widehat{\boldsymbol J}_m(\boldsymbol\vartheta_m^*+u\boldsymbol h_m)du$.
The score equation gives the exact identity
$\boldsymbol h_m=\overline{\boldsymbol J}_m^{-1}T^{-1}
\sum_{t=1}^T\boldsymbol s_{m,t}(\boldsymbol\vartheta_m^*)$.
Uniform nonsingularity and \eqref{eq:primitive_score_et} first yield
\eqref{eq:primitive_consistency}. The empirical Hessian bound, the third-derivative envelope and
\eqref{eq:primitive_consistency} then give
$\max_m\|\overline{\boldsymbol J}_m^{-1}-(\boldsymbol J_m^*)^{-1}\|
=O_p(r_T/\sqrt T)$. Multiplication by the maximal normalised score, which is $O_p(r_T)$, proves
\eqref{eq:primitive_bahadur}.
\end{proof}

\begin{lemma}[Uniform standard-error approximation]
\label{lem:primitive_variance}
Suppose either: (i) the information equality
$\boldsymbol\Omega_m^*=\boldsymbol J_m^*$ holds uniformly and standard errors
are computed from the observed information; or (ii) the stagewise scores are
serially uncorrelated and standard errors are computed from the one-period
sandwich estimator. If $r_T=o(\sqrt T)$, then
\begin{equation}
\Delta_T
=
\max_{m\in\mathcal A_T}
\left|
\frac{\widehat{\mathrm{se}}_m}
{T^{-1/2}\{V_{\theta\theta,m}^*\}^{1/2}}-1
\right|
=O_p\!\left(\frac{r_T}{\sqrt T}\right).
\label{eq:primitive_se}
\end{equation}
\end{lemma}

\begin{proof}
In case (i), Lemmas~\ref{lem:primitive_empirical} and
\ref{lem:primitive_bahadur} give
$\max_m\|\widehat{\boldsymbol J}_m(
\widehat{\boldsymbol\vartheta}_m)-\boldsymbol J_m^*\|
=O_p(r_T/\sqrt T)$. In case (ii), the same argument also gives
$\max_m\|\widehat{\boldsymbol\Omega}_m(
\widehat{\boldsymbol\vartheta}_m)-\boldsymbol\Omega_m^*\|
=O_p(r_T/\sqrt T)$; serial uncorrelatedness identifies the long-run covariance
with the one-period score covariance. The inverse-matrix identity and the
uniform eigenvalue bounds in A3.3 transfer these rates to the corresponding
asymptotic covariance matrices and their candidate-coordinate diagonal
elements. Taking square roots proves \eqref{eq:primitive_se}.
\end{proof}

\begin{theorem}[Primitive verification of A3.4]
\label{thm:primitive_wald}
Suppose A1--A3.3 and the parameter-set regularity stated at the start of this
appendix hold, and the standard errors satisfy one of the two cases in
Lemma~\ref{lem:primitive_variance}. Under A2-ET, if
$\ell_T^3/T\to0$, then $\mathcal E_T=O_p(\sqrt{\ell_T})$. Under A2-PM, if
$p^{K/q}=o(\sqrt T)$, then $\mathcal E_T=O_p(p^{K/q})$. Consequently, A3.4
holds with $a_{T,p}^{\mathrm{PM}}=p^{K/q}$.
\end{theorem}

\begin{proof}
Let
$A_m=\sqrt T\,\theta_m^*/\{V_{\theta\theta,m}^*\}^{1/2}$ and let $Z_m$ be the
candidate coordinate of the linear term in \eqref{eq:primitive_bahadur},
normalised by $\{V_{\theta\theta,m}^*\}^{1/2}$. Lemmas
\ref{lem:primitive_empirical} and \ref{lem:primitive_bahadur} give
$\max_m|Z_m|=O_p(r_T)$ and remainder variables $R_{m,T}$ satisfying
$\max_m|R_{m,T}|=o_p(r_T)$. A3.2--A3.3 imply
$\max_m|A_m|=O(\sqrt T)$. Write the relative standard-error error as
$\delta_m$, so that $\max_m|\delta_m|=\Delta_T$. On the event
$\Delta_T\le1/2$,
\[
\left|
\left|\frac{A_m+Z_m+R_{m,T}}{1+\delta_m}\right|-|A_m|
\right|
\le |Z_m|+|R_{m,T}|
+2|\delta_m|\{|A_m|+|Z_m|+|R_{m,T}|\}.
\]
By Lemma~\ref{lem:primitive_variance}, the final term is $O_p(r_T)$ uniformly
in $m$. Taking the maximum proves $\mathcal E_T=O_p(r_T)$ and hence both
claims.
\end{proof}

The ET condition $\ell_T^3/T\to0$ is equivalent to
$\log(p\vee T)=o(T^{1/3})$, so it permits
$p=\exp\{o(T^{1/3})\}$. In the PM regime, a threshold can satisfy both
$p^{K/q}=o(c_T)$ and $c_T=o(\sqrt T)$ only if
\begin{equation}
p=o\!\left(T^{q/(2K)}\right),
\qquad K=k_{\max}+1.
\label{eq:primitive_pm_growth}
\end{equation}
For example, $d_0<8$, $q=8$ and $k_{\max}=5$ give $p=o(T^{2/3})$; only
the ET regime permits faster-than-polynomial growth. For kernel HAC standard
errors, the uniform long-run-variance estimation error must instead enter
A3.4; the recovery proofs remain valid with the resulting rate in A4--A5.

\end{document}